\documentclass{jfp1}

\usepackage{natbib}

\usepackage{graphicx}

\usepackage{amssymb,amsopn,latexsym,amsmath,textcomp}

\usepackage{upgreek}

\usepackage{multirow}

\usepackage{url}

\usepackage{appendix}

\usepackage[table]{xcolor}

\usepackage[all]{xy}

\usepackage{stmaryrd}

\usepackage{code}

\usepackage{caption}

\usepackage{pfsteps} 
\newcommand{\defterm}[1]{\textbf{#1}}

\newcommand{\ie}{\emph{i.e.}}
\newcommand{\eg}{\emph{e.g.}}

\newcommand{\syn}[1]{\mathsf{#1}}
\newcommand{\var}[1]{\mathit{#1}}
\newcommand{\s}[1]{\mathit{#1}}

\newcommand{\parto}{\rightharpoonup}
\newcommand{\dom}{\var{dom}}
\newcommand{\range}{\var{range}}

\newcommand{\compose}{\mathrel{\circ}}

\newcommand{\monto}{\xrightarrow{\mathrm{mon}}}

\newcommand{\set}[1]{\left\{#1\right\}}
\newcommand{\setbuild}[2]{\left\{ #1 : #2\right\}}

\newcommand{\Pow}[1]{{\mathcal{P}\left(#1\right)}}
\newcommand{\PowSm}[1]{{\mathcal{P}(#1)}}

\newcommand{\union}{\cup}

\newcommand{\Union}{\bigcup}

\newcommand{\vect}[1]{\langle #1\rangle}
\newcommand{\vecp}[1]{\vec{#1}\;'}

\newcommand{\To}{\mathrel{\Rightarrow}}

\newcommand{\wt}{\sqsubseteq}

\newcommand{\join}{\sqcup}

\newcommand{\bigjoin}{\bigsqcup}

\DeclareMathOperator{\lfp}{lfp}

\newcommand{\infer}[2]{{\renewcommand{\arraystretch}{1.4}\begin{array}{c}#1\\ \hline #2\end{array}\renewcommand{\arraystretch}{1.0}}}

\newcommand{\transition}{\delta}

\newcommand{\QStates}{Q}
\newcommand{\FStates}{F}

\newcommand{\StackAlpha}{\Gamma}

\newcommand{\stackchar}{\gamma}

\newcommand{\sembr}[1]{\ensuremath{[\![{#1}]\!]}}

\newcommand{\opor}{\mathrel{|}}

\newcommand{\Alphabet}{A}

\newcommand{\produces}{\mathrel{::=}}

\newcommand{\Lang}{\mathcal{L}} 

\newcommand{\vv}{v}

\newcommand{\lam}{\ensuremath{\var{lam}}}

\newcommand{\lamterm}{$\lambda$-term}
\newcommand{\lc}{$\lambda$-calculus}

\newcommand{\call}{\ensuremath{\var{call}}}

\newcommand{\lt}[2]{\lambda #1.#2}

\newcommand{\ttlp}{\mbox{\tt (}}
\newcommand{\ttrp}{\mbox{\tt )}}

\newcommand{\appform}[2]{\ttlp #1\; #2\ttrp}
\newcommand{\lamform}[2]{\ttlp \uplambda\;\ttlp#1\ttrp\;#2\ttrp}

\newcommand{\letiform}[3]{\ttlp {\tt let}\; \ttlp\ttlp#1\; #2\ttrp\ttrp\; #3\ttrp}

\newcommand{\fexpr}{f}
\newcommand{\expr}{e}
\newcommand{\aexpr}{\mbox{\sl {\ae}}}

\newcommand{\Eval}{{\mathcal{E}}}
\newcommand{\ArgEval}{{\mathcal{A}}}

\newcommand{\Inject}{{\mathcal{I}}}

\newcommand{\QState}{Q}
\newcommand{\qstate}{q}

\newcommand{\tf}{f}

\newcommand{\store}{\sigma}

\newcommand{\env}{\rho}

\newcommand{\clo}{\var{clo}}
\newcommand{\cont}{\kappa}

\newcommand{\alloc}{\mathit{alloc}}

\newcommand{\addr}{a}

\newcommand{\aTo}{\leadsto}

\newcommand{\aInject}{{\hat{\mathcal{I}}}}

\newcommand{\sa}[1]{\widehat{\mathit{#1}}}

\newcommand{\aEval}{{\hat{\mathcal{E}}}}
\newcommand{\aArgEval}{{\hat{\mathcal{A}}}}

\newcommand{\atf}{{\hat{f}}}

\newcommand{\astore}{{\hat{\sigma}}}
\newcommand{\aenv}{{\hat{\rho}}}

\newcommand{\aclo}{{\widehat{\var{clo}}}}

\newcommand{\acont}{{\hat{\kappa}}}

\newcommand{\aaddr}{{\hat{\addr}}}
\newcommand{\aalloc}{{\widehat{alloc}}}

\newcommand{\absmap}{\alpha}
\newcommand{\abs}[1]{|#1|}

\usepackage{xargs}
\newcommandx{\minipagebreak}[6][2=0pt,5=0pt]{%
  \begin{minipage}{#1\linewidth}
    \vspace{#2}
    \begin{align*}
    #3
    \end{align*}\end{minipage}
  \begin{minipage}{#4\linewidth}
    \vspace{#5}
    \begin{align*}
      #6
    \end{align*}
  \end{minipage}}

\newcommand{\ControlStates}{Q}
\newcommand{\transfunction}{\delta}

\newcommand{\conf}{c}
\newcommand{\aconf}{{\hat c}}

\newcommand{\phrame}{\phi}
\newcommand{\aphrame}{\hat{\phi}}

\newcommand{\stackact}{g}

\newcommand{\obsolete}[1]{} 

\DeclareMathOperator*{\PDTrans}{\longmapsto}

\newcommand{\fECG}{\mathcal{ECG}}
\newcommand{\fCRPDS}{\mathcal{C}}
\newcommand{\fCCPDS}{\mathcal{CC}}

\newcommand{\afPDA}{\widehat{\mathcal{PDA}}}

\newcommand{\afIPDS}{\widehat{\mathcal{IPDS}}}

\newcommand{\afRPDS}{\widehat{\mathcal{RPDS}}}

\renewcommand{\Alphabet}{\Sigma}

\DeclareMathOperator*{\pdedge}{\rightarrowtail}

\newcommand{\quadedge}[4]{#1 \mathrel{\overset{#2}{\underset{#3}{\pdedge}}} #4}
\newcommand{\triedge}[3]{#1 \mathrel{\overset{#2}{\pdedge}} #3}
\newcommand{\biedge}[2]{#1 \mathrel{\rightarrowtail} #2}

\newcommand{\DSStates}{S}
\newcommand{\DSEdges}{E}
\newcommand{\DSIEdges}{E}
\newcommand{\DSFrames}{\StackAlpha}

\newcommand{\dsstate}{s}

\newcommand{\mkCRPDS}{\mathcal{F}}
\newcommand{\mkCCPDS}{\mathcal{F}}

\newcommand{\Stacks}{\mathit{Stacks}}

\newcommand{\fnet}[1]{\lfloor #1 \rfloor}
\newcommand{\fstackify}[1]{\lceil #1 \rceil}
\newcommand{\StackRoot}{\mathit{StackRoot}}

\newcommand{\ecg}{$\epsilon$-closure graph}

\DeclareMathOperator*{\RPDTrans}{{\longmapsto\!\!\!\!\!\!\!\!\!\!\!\!\longrightarrow}}
\newcommand{\RPDTransOU}[2]{\mathrel{\overset{#1}{\underset{#2}{\RPDTrans}}}}
\newcommand{\PDTranssOU}[2]{\mathrel{\overset{#1}{\underset{#2}{\PDTrans}}^*}}
\newcommand{\PDTransOU}[2]{\mathrel{\overset{#1}{\underset{#2}{\PDTrans}}}}

\newcommand{\fsprout}{\mathit{sprout}}
\newcommand{\faddpush}{\mathit{addPush}}
\newcommand{\faddpop}{\mathit{addPop}}
\newcommand{\faddempty}{\mathit{addEmpty}}

\newcommand{\touches}{{\mathcal T}}

\newcommand{\eancestor}[1]{\overleftarrow{G}_\epsilon[#1]}
\newcommand{\edescendent}[1]{\overrightarrow{G}_\epsilon[#1]}

\newcommand{\system}{C}
\newcommand{\asystem}{{\hat{C}}}
\newcommand{\att}{\hat{t}}

\newcommand{\apconf}{{\hat{\pi}}}
\newcommand{\apstate}{{\hat{\psi}}}
\newcommand{\aopstate}{{\hat{\Omega}}}

\renewcommand{\aTo}{\mathrel{\widehat{\Rightarrow}}}
\DeclareMathOperator*{\afTo}{\rightarrowtriangle}

\DeclareMathOperator*{\apTo}{\rightharpoondown}

\DeclareMathOperator*{\areaches}{\rightarrowtriangle}

\newcommand{\pdcfato}[4]{#1 \mathrel{\apTo^{#2}_{#3}} #4}
\newcommand{\ipdcfato}[4]{#1 \mathrel{\apTo^{#2}_{#3}} #4}

\newcommand{\aCollect}{{\hat{G}}}

\newcommand{\mtrace}{\pi}
\newcommand{\Prop}{\mathit{Prop}}
\newcommand{\invcrpdsp}{\mathit{inv}}

\newcommand{\mylongtitle}{Pushdown flow analysis with abstract garbage collection}

\newcommand{\mytitle}{\mylongtitle}

\newtheorem{theorem}{Theorem}[section]
\newtheorem{lemma}{Lemma}[section]
\newtheorem{corollary}{Corollary}[section]

\bibpunct();A{},
\let\cite=\citep

\renewcommand{\appendix}{
\renewcommand*{\theHsection}{chY.\the\value{section}}}

\title[PUSHDOWN FLOW ANALYSIS WITH ABSTRACT GARBAGE COLLECTION]
{\mytitle}

\author[]{J.~IAN JOHNSON\\
  Northeastern University}
\author[]{ILYA SERGEY\\
  IMDEA Software Institute}
\author[]{CHRISTOPHER EARL\\
  University of Utah}
\author[]{MATTHEW MIGHT\\
  University of Utah}
\author[J.I.~Johnson, I.~Sergey, C.~Earl, M.~Might, and D.~Van Horn]{DAVID VAN HORN\\
  University of Maryland}

\begin{document}
\maketitle

\begin{abstract}

In the static analysis of functional programs, pushdown flow analysis
and abstract garbage collection push the boundaries of what we can
learn about programs statically.  This work illuminates and poses
solutions to theoretical and practical challenges that stand in the
way of combining the power of these techniques.  Pushdown flow
analysis grants unbounded yet computable polyvariance to the analysis
of return-flow in higher-order programs.  Abstract garbage collection
grants unbounded polyvariance to abstract addresses which become
unreachable between invocations of the abstract contexts in which they
were created.  Pushdown analysis solves the problem of precisely
analyzing recursion in higher-order languages; abstract garbage
collection is essential in solving the ``stickiness'' problem. Alone,
our benchmarks demonstrate that each method can reduce analysis times
and boost precision by orders of magnitude.

We combine these methods.  The challenge in marrying these techniques
is not subtle: computing the reachable control states of a pushdown
system relies on limiting access during transition to the top of the
stack; abstract garbage collection, on the other hand, needs full
access to the entire stack to compute a root set, just as concrete
collection does. \emph{Conditional} pushdown systems were developed
for just such a conundrum, but existing methods are ill-suited for the
dynamic nature of garbage collection.

We show fully precise and approximate solutions to the feasible paths problem for pushdown garbage-collecting control-flow analysis.
Experiments reveal synergistic interplay between garbage collection and pushdown techniques, and the fusion demonstrates
``better-than-both-worlds'' precision.



\end{abstract}

\section{Introduction}

The development of a context-free\footnote{%
As in context-free language, not a context-insensitive analysis.
} approach to control-flow analysis
(CFA2) by \citet{mattmight:Vardoulakis:2010:CFA2} provoked a shift in the
static analysis of higher-order
programs.
Prior to CFA2, a precise analysis of recursive behavior 
had been a challenge---even though flow analyses have an important role to play in 
optimization for functional languages, such as 
flow-driven inlining~\cite{mattmight:Might:2006:DeltaCFA},
interprocedural constant propagation~\cite{mattmight:Shivers:1991:CFA}
and type-check elimination~\cite{mattmight:Wright:1998:Polymorphic}.

While it had been possible to statically analyze
recursion \emph{soundly}, CFA2 made it possible to analyze recursion
\emph{precisely} by matching calls and returns without approximating the stack as $k$-CFA does.
The approximation is only in the binding structure, and not the control structure of the program.
In its pursuit of recursion,
clever engineering steered CFA2 to a \emph{theoretically} intractable complexity, though in practice it performs well.
Its payoff is significant reductions in 
analysis time \emph{as a result of} corresponding increases
in precision.

For a visual measure of the impact, Figure~\ref{fig:diamond}
renders the abstract transition graph (a model of all possible traces through the program) for 
the toy program in Figure~\ref{fig:toy}.
\begin{figure}
\figrule
\begin{code}
(define (id x) x)

(define (f n)
  (cond [(<= n 1)  1]
        [else      (* n (f (- n 1)))]))

(define (g n)
  (cond [(<= n 1)  1]
        [else      (+ (* n n) (g (- n 1)))]))
    
(print (+ ((id f) 3) ((id g) 4)))
\end{code}
\caption{A small example to illuminate the strengths and weaknesses of
  both pushdown analysis and abstract garbage collection.}
\label{fig:toy}
\figrule
\end{figure}
For this example,
pushdown analysis
eliminates spurious return-flow from the
use of  recursion.
But, recursion is just one problem of many for flow analysis.
For instance, pushdown analysis still gets tripped up by
the
spurious cross-flow problem;
at calls to \texttt{(id f)}
and \texttt{(id g)} in the previous example,
it thinks \texttt{(id g)} could be \texttt{f} \emph{or} \texttt{g}.
CFA2 is not confused in this due to its precise stack frames, but can be confused by unreachable heap-allocated bindings.

Powerful techniques such as abstract garbage collection~\cite{mattmight:Might:2006:GammaCFA}
were developed to address the cross-flow problem (here in a way complementary to CFA2's stack frames).
The cross-flow problem arises because monotonicity prevents revoking a
judgment like ``procedure ${\tt f}$ flows to {\tt x},'' or ``procedure
${\tt g}$ flows to {\tt x},'' once it's been made.

\begin{figure}
\begin{center}
\includegraphics[width=3in]{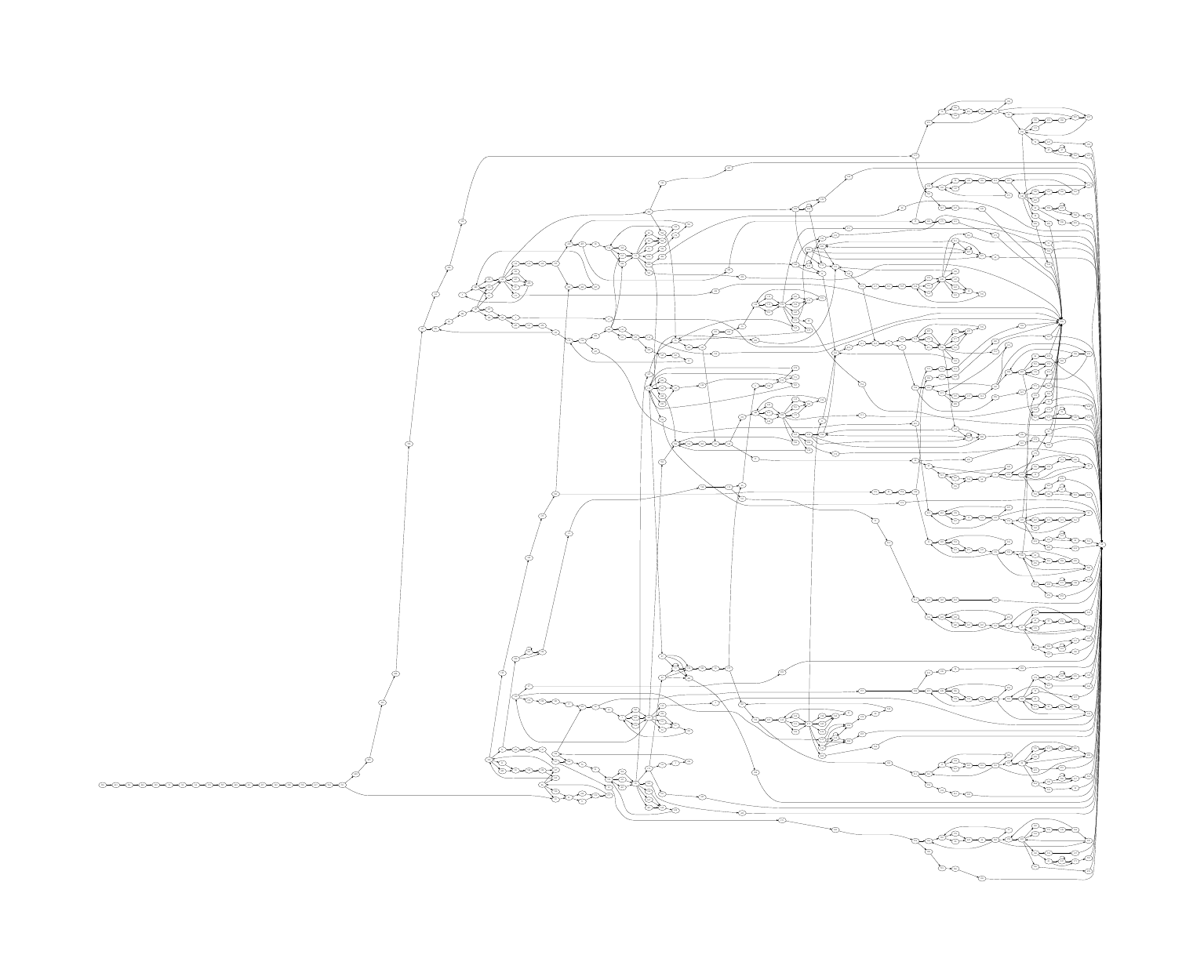}
\\
(1) without pushdown analysis or abstract GC: 653 states
\\
\includegraphics[width=3in]{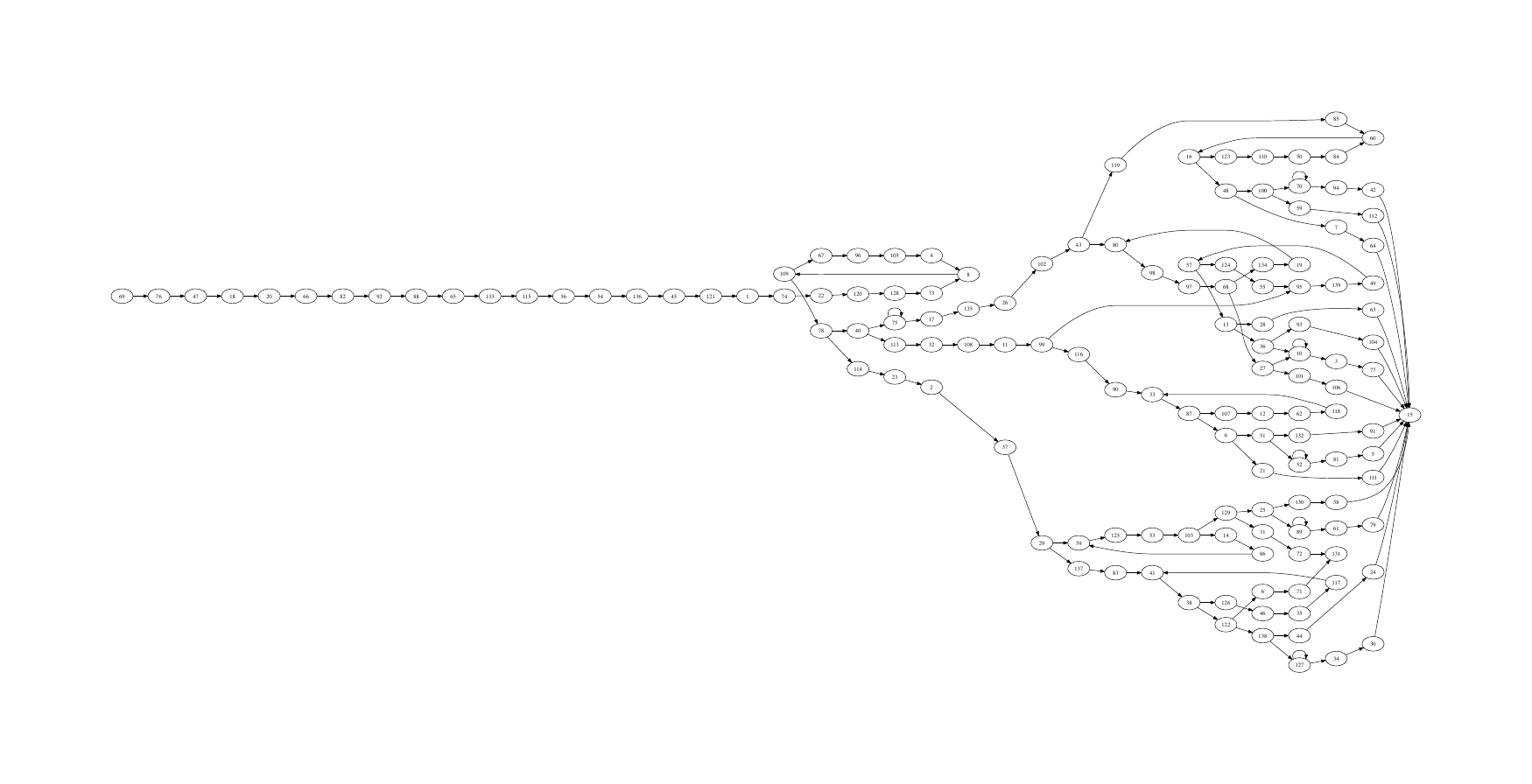}
\\
(2) with pushdown only: 139 states
\\
\includegraphics[width=3.3in]{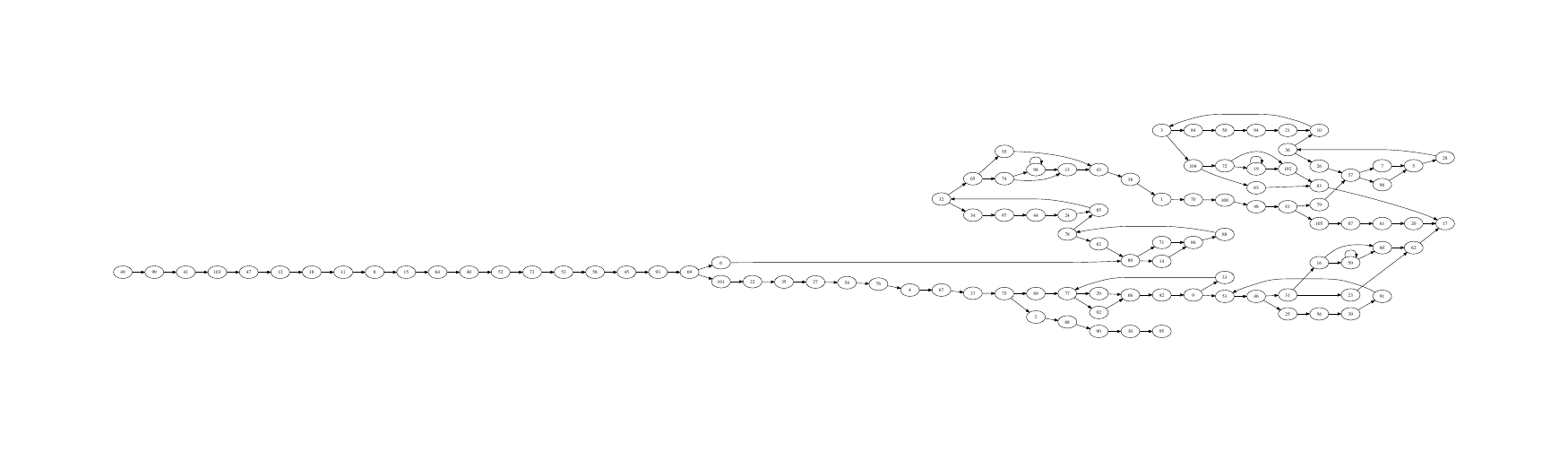}
\\
(3) with GC only: 105 states
\\
\includegraphics[width=2.5in]{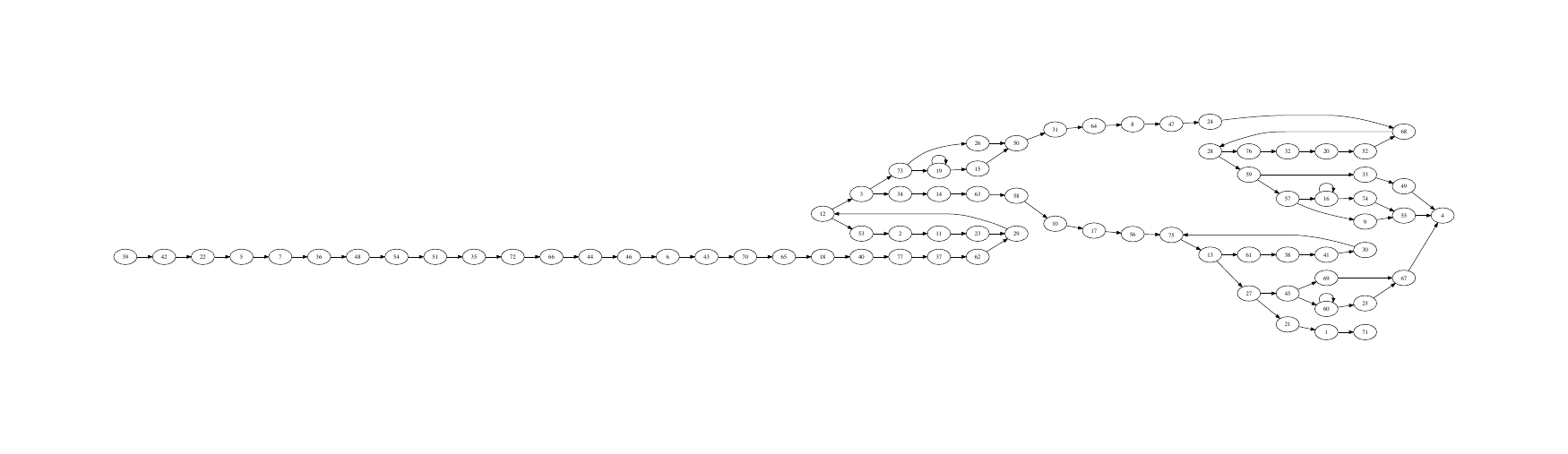}
\\
(4) with pushdown analysis and abstract GC: 77 states
\end{center}

\caption{
We generated an abstract transition graph
for the same program from Figure~\ref{fig:toy} four times: 
 (1) without pushdown analysis or abstract garbage collection;
 (2) with only abstract garbage collection;
 (3) with only pushdown analysis;
 (4) with both pushdown analysis and abstract garbage collection.
With only pushdown or abstract GC, the abstract transition graph shrinks
by an order of magnitude, but in different ways.
The pushdown-only analysis is confused by variables 
that are bound to several different higher-order functions,
but for short durations.
The abstract-GC-only is confused by non-tail-recursive loop structure.
With both techniques enabled, the graph shrinks by nearly half yet again
and fully recovers the control structure of the original program.}
\label{fig:diamond}
\end{figure}

In fact, abstract garbage collection, by itself, also delivers significant
improvements to analytic speed and precision in many benchmarks.
(See Figure~\ref{fig:diamond} again for a visualization of that
impact.)

It is natural to ask: can abstract garbage collection and pushdown 
analysis work together?
Can their strengths be multiplied?
At first, the answer appears to be a disheartening~``\emph{No}.''

%
%
%

\subsection{The problem: The whole stack \emph{versus} just the top}

Abstract garbage collection seems to require more than pushdown analysis can
decidably provide: access to the full stack.
Abstract garbage collection, like its name implies, discards unreachable values
from an abstract store during the analysis.
Like concrete garbage collection, abstract garbage collection also begins its
sweep with a root set, and like concrete garbage collection, it must traverse
the abstract stack to compute that root set.
But, pushdown systems are restricted to viewing the top of the stack (or
a bounded depth)---a condition violated by this traversal.

Fortunately, abstract garbage collection does not need to arbitrarily modify
the stack.
It only needs to know the root set of addresses in the stack.
This kind of system has been studied before in the context of compilers that build a symbol table (a so-called ``one-way stack automaton''~\citep{ianjohnson:one-way-sa:ginsburg:1967}),%
 in the context of first-order model-checking (pushdown systems with checkpoints~\citep{EsparzaKS03}),%
 and also in the context of points-to analysis for Java (conditional weighted pushdown systems (CWPDS) ~\citep{ianjohnson:DBLP:conf/pepm/LiO10}).
We borrow the definition of (unweighted) conditional pushdown system (CPDS) in this work, though our analysis does not take CPDSs as inputs.
Higher-order flow analyses typically do not take a control-flow graph, or similar pre-abstracted object, as input and produce an annotated graph as output.
Instead, they take a program as input and ``run it on all possible inputs'' (abstractly) to build an approximation of the language's reduction relation (semantics), specialized to the given program.
This semantics may be non-standard in such a way that extra-semantic information might be accumulated for later analyses' consumption.
The important distinction between higher-order and first-order analyses is that the \emph{model} to analyze is built \emph{during} the analysis, which involves interpreting the program (abstractly).
When a language's semantics treats the control stack as an actual stack, \ie, it does not have features such as first-class continuations, an interpreter can be split into two parts:
 a function that takes the current state and returns all next states along with a pushed activation frame or a marker that the stack is unchanged;
 and a function that takes the current state, a possible ``top frame'' of the stack, and returns the next states after popping this frame.
This separation is crucial for an effective algorithm, since pushed frames are understood from program text, and popped frames need only be enumerated from a (usually small) set that we compute along the way.
Control-state reachability for the straightforward formulation of stack introspection ends up being uncomputable.
Conditional pushdown systems introduce a relatively weak regularity constraint on transitions' introspection:
a CPDS may match the current stack against a choice of finitely many regular languages of stacks in order to transition from one state to the next along with the stack action.
The general solutions to feasible paths in \emph{conditional} pushdown systems enumerate all languages of stacks that a transition may be conditioned on.
This strategy is a non-starter for garbage collection, since we delineate stacks by the addresses they keep live; this is exponential in the number of addresses.
The abstraction step that finitizes the address space is what makes the problem fall within the realm of CPDSs, even if the target is so big it barely fits.
But abstract garbage collection is special --- we can compute which languages of stacks we need to check against, given the current state of the analysis.
It is therefore possible to fuse the full benefits of abstract garbage collection with pushdown analysis.
The dramatic reduction in abstract transition graph size
from the top to the bottom in Figure~\ref{fig:diamond}
(and echoed by later benchmarks) conveys the impact
of this fusion.

\paragraph{Secondary motivations}
There are four secondary motivations for this work:
\begin{enumerate}
\item bringing context-sensitivity to pushdown analysis;
\item exposing the context-freedom of the analysis;
\item enabling pushdown analysis without continuation-passing style; and
\item defining an
  alternative algorithm for computing pushdown analysis,
  introspectively or otherwise.
\end{enumerate}
In CFA2, monovariant (0CFA-like) context-sensitivity is etched directly into
the abstract ``local'' semantics, which is in turn
phrased in terms of an explicit (imperative) summarization algorithm
for a partitioned continuation-passing style.
Our development exposes the classical parameters (exposed as
allocation functions in a semantics) that allow one to tune the
context-sensitivity and polyvariance (accomplishing (1)), thanks to
the semantics of the analysys formulated in the form of an
``abstracted abstract machine''~\cite{mattmight:VanHorn:2012:AAM}.

In addition, the context-freedom of CFA2 is buried implicitly
inside an imperative summarization algorithm.
No pushdown system or context-free grammar is explicitly identified.
Thus, a motivating factor for our work was to make 
the pushdown system in CFA2 explicit, and to make the control-state
reachability algorithm purely functional (accomplishing (2)).

A third motivation was to show that a transformation to continuation-passing
style is unnecessary for pushdown analysis.
In fact, pushdown analysis is arguably more natural over direct-style programs.
By abstracting all machine components except for the program stack, 
it converts naturally and readily into a pushdown system (accomplishing (3)).
In his dissertation, Vardoulakis showed a direct-style version of CFA2 that exploits the meta-language's runtime stack to get precise call-return matching.
The approach is promising, but its correctness remains unproven, and it does not apply to generic pushdown systems.

Finally, to bring much-needed clarity to algorithmic formulation of
pushdown analysis, we have included an appendix containing a reference
implementation in Haskell (accomplishing (4)).
We have kept the code as close in form to the mathematics as possible, so that
where concessions are made to the implementation, they are obvious.

\subsection{Overview}
We first review preliminaries to set a consistent feel for
terminology and notation, particularly with respect
to pushdown systems.
The derivation of the analysis 
begins with a concrete CESK-machine-style semantics for
A-Normal Form \lc{}.
The next step is an infinite-state abstract interpretation, 
constructed by
bounding the C(ontrol), E(nvironment) and S(tore) portions of the machine.
Uncharacteristically,
we leave the stack component---the K(ontinuation)---unbounded.

A shift in perspective reveals that this abstract interpretation is a pushdown
system.
We encode it as a pushdown automaton explicitly, and pose control state
reachability as a decidable language intersection problem.
We then extract a rooted pushdown system from the pushdown automaton.
For completeness, we fully develop pushdown analysis for higher-order programs,
including an efficient algorithm for computing reachable control states.
We go further by characterizing complexity and demonstrating the approximations
necessary to get to a polynomial-time algorithm.

We then introduce abstract garbage collection and quickly find
that it violates the pushdown model
with its traversals of the stack.
To prove the decidability of control-state reachability,
we formulate introspective pushdown systems, and 
recast abstract garbage collection within this framework.
We then review that control-state reachability is decidable for
introspective pushdown systems as well when
subjected to a straightforward regularity constraint.

We conclude with an implementation and empirical evaluation that shows strong
synergies between pushdown analysis and abstract garbage collection, including
significant reductions in the size of the abstract state transition graph.

\subsection{Contributions}
We make the following contributions:
\begin{enumerate}
\item Our primary contribution is an \emph{online} decision procedure
  for reachability in introspective pushdown systems, with a more
  efficient specialization to abstract garbage collection.

\item We show that classical notions of context-sensitivity, such as
  $k$-CFA and poly/CFA, have direct generalizations in a pushdown
  setting.  CFA2 was presented as a monovariant
  analysis,\footnote{Monovariance refers to an abstraction that groups
    all bindings to the same variable together: there is \emph{one}
    abstract variant for all bindings to each variable.} whereas we
  show polyvariance is a natural extension.

\item We make the context-free aspect of CFA2 explicit: we clearly define and
identify the pushdown system.
We do so by starting with a classical CESK machine and systematically
abstracting until a pushdown system emerges.
We also remove the orthogonal frame-local-bindings aspect of CFA2, so as to
focus solely on the pushdown nature of the analysis.

\item (*) We remove the requirement for a global CPS-conversion
by synthesizing the analysis directly for direct-style (in
    the form of A-normal form lambda-calculus --- a local transformation).

\item We empirically validate claims of improved
precision on a suite of benchmarks.
We find synergies between 
pushdown analysis 
and
abstract garbage collection 
that makes the whole greater that the sum of its parts.

\item
We provide a mirror of the major formal development as working 
Haskell code in the appendix.  
This code illuminates dark corners
of pushdown analysis and it 
provides a concise formal reference implementation.

\end{enumerate}

(*) The CPS requirement distracts from the connection between continuations and stacks.
We do not discuss \texttt{call/cc} in detail, since we believe there are no significant barriers to adapting the techniques of \citet{dvanhorn:Vardoulakis2011Pushdown} to the direct-style setting, given related work in \citet{local:hopa-summaries}.
Languages with exceptions fit within the pushdown model since a throw can be modeled as ``pop until first catch.''

\section{Pushdown Preliminaries}

The literature contains  many equivalent definitions of pushdown machines, so
we adapt our own definitions from \citet{mattmight:Sipser:2005:Theory}.
\emph{Readers familiar with pushdown theory may wish to skip ahead.}

\subsection{Syntactic sugar}

When a triple $(x,\ell,x')$ is an edge in a labeled graph:
\begin{equation*}
x \pdedge^\ell x'  \equiv
(x,\ell,x')
\text.
\end{equation*}
Similarly, when a pair $(x,x')$ is a graph edge:
\begin{equation*}
\biedge{x}{x'} \equiv (x,x')
\text.
\end{equation*}
We use both
string and vector notation for sequences:
\begin{equation*}
a_1 a_2 \ldots a_n \equiv \vect{a_1,a_2,\ldots,a_n}
\equiv
\vec{a} 
\text.
\end{equation*}

\subsection{Stack actions, stack change and stack manipulation}

Stacks are sequences over a stack alphabet $\StackAlpha$.
To reason about stack manipulation concisely,
we first turn stack alphabets into ``stack-action'' sets;
each character represents a change to the stack: push, pop or no
change.

For each character $\stackchar$ in a stack alphabet $\StackAlpha$, the
\defterm{stack-action} set $\StackAlpha_\pm$ contains a push character
$\stackchar_{+}$; a pop character $\stackchar_{-}$;
and a no-stack-change indicator, $\epsilon$:
\begin{align*}
\stackact \in \StackAlpha_\pm &\produces \epsilon && \text{[stack unchanged]} 
\\
    &\;\;\opor\;\; \stackchar_{+}  \;\;\;\text{ for each } \stackchar \in \StackAlpha && \text{[pushed $\stackchar$]}
    \\
      &\;\;\opor\;\; \stackchar_{-}  \;\;\;\text{ for each } \stackchar \in \StackAlpha && \text{[popped $\stackchar$]}
      \text.
\end{align*}
In this paper, the symbol $\stackact$ represents some stack action.

When we develop introspective pushdown systems, we are going
to need formalisms for easily manipulating stack-action strings
and stacks.
Given a string of stack actions, we can compact it into a minimal
string describing net stack change.
We do so through the operator $\fnet{\cdot} : \StackAlpha_\pm^* \to
\StackAlpha_\pm^*$, which cancels out opposing adjacent push-pop stack
actions:
\begin{align*}
\fnet{\vec{\stackact} \; \stackchar_+\stackchar_- \; \vecp{\stackact}} &= 
\fnet{\vec{\stackact} \; \vecp{\stackact}} 
&
\fnet{\vec{\stackact} \; \epsilon \; \vecp{\stackact}} &= 
\fnet{\vec{\stackact} \; \vecp{\stackact}} 
\text,
\end{align*}
so that
$\fnet{\vec{\stackact}} = \vec{\stackact}\text,$
if there are no cancellations to be made in the string $\vec{\stackact}$.

We can convert a net string back into a stack by stripping off the
push symbols with the stackify operator, $\fstackify{\cdot} :
\StackAlpha^{*}_\pm \parto \StackAlpha^*$:
\begin{align*}
\fstackify{\stackchar_+ \stackchar_+' \ldots \stackchar_+^{(n)}} =
\vect{\stackchar^{(n)}, \ldots, \stackchar', \stackchar}
\text,
\end{align*}
and for convenience, $[\vec{\stackact}] = 
\fstackify{\fnet{\vec{\stackact}}}$.
Notice the stackify operator is defined for strings containing
only push actions. 

  %
  %
  %

  \subsection{Pushdown systems}
  A \defterm{pushdown system} is a triple
  $M = (\ControlStates,\StackAlpha,\transfunction)$ where:
  \begin{enumerate}

  \item $\ControlStates$ is a finite set of control states;

  \item $\StackAlpha$ is a stack alphabet; and

  \item $\transfunction \subseteq
  \ControlStates \times \StackAlpha_\pm \times \ControlStates$ is a transition relation.
  \end{enumerate}
  The set $\ControlStates \times \StackAlpha^*$ is 
  called the \defterm{configuration-space} of this pushdown system.
  We use $\mathbb{PDS}$ to denote the class of all pushdown systems.
  \\

  \noindent
  For the following definitions, let $M = (\ControlStates,\StackAlpha,\transfunction)$.
  \begin{itemize}


  \item The labeled \defterm{transition relation} $(\PDTrans_{M}) \subseteq
  (\ControlStates \times \StackAlpha^*) \times 
  \StackAlpha_\pm \times 
  (\ControlStates \times \StackAlpha^*)$
  determines whether one configuration may transition to another while performing the given stack action:
        \begin{align*}
(\qstate, \vec{\stackchar}) 
  \mathrel{\overset{\epsilon}{\underset{M}{\PDTrans}}}
  (\qstate',\vec{\stackchar}) 
  & \text{ iff }
  \qstate \pdedge^\epsilon \qstate'
  \in \transfunction
  && \text{[no change]}
  \\
    (\qstate, \stackchar : \vec{\stackchar}) 
    \mathrel{\overset{\stackchar_{-}}{\underset{M}{\PDTrans}}}
    (\qstate',\vec{\stackchar})
    & \text{ iff }
    \qstate \pdedge^{\stackchar_{-}} \qstate'
    \in \transfunction
    && \text{[pop]}
    \\
      (\qstate, \vec{\stackchar}) 
      \mathrel{\overset{\stackchar_{+}}{\underset{M}{\PDTrans}}}
      (\qstate',\stackchar : \vec{\stackchar}) 
      & \text{ iff }
      \qstate \pdedge^{\stackchar_{+}} \qstate'
      \in \transfunction
      && \text{[push]}
      \text.
      \end{align*}

\item If unlabelled, the transition relation $(\PDTrans)$ checks whether \emph{any} stack action can enable the transition:
\begin{align*}
\conf \mathrel{\underset{M}{\PDTrans}} \conf' \text{ iff }
\conf \mathrel{\overset{\stackact}{\underset{M}{\PDTrans}}} \conf' \text{ for some stack action } \stackact 
\text.
\end{align*}

\item

For a string of stack actions $\stackact_1 \ldots
\stackact_n$:
\begin{equation*}
\conf_0 \mathrel{\overset{\stackact_1\ldots\stackact_n}{\underset{M}{\PDTrans}}} \conf_n
\text{ iff }
\conf_0
\mathrel{\overset{\stackact_1}{\underset{M}{\PDTrans}}}
\conf_1
\mathrel{\overset{\stackact_2}{\underset{M}{\PDTrans}}}
 \cdots
\mathrel{\overset{\stackact_{n-1}}{\underset{M}{\PDTrans}}}
\conf_{n-1} 
\mathrel{\overset{\stackact_n}{\underset{M}{\PDTrans}}}
\conf_n\text,
\end{equation*}
for some configurations $\conf_0,\ldots,\conf_n$.

\item

For the transitive closure:
\begin{equation*}
  \conf \mathrel{\underset{M}{\PDTrans}^{*}} \conf'
  \text{ iff }
  \conf \mathrel{\overset{\vec{\stackact}}{\underset{M}{\PDTrans}}} \conf'
   \text{ for some action string }
   \vec{\stackact}
 \text.
\end{equation*}

\end{itemize}

  %

\paragraph{Note}
Some texts define the transition relation $\transfunction$ so that 
$\transfunction \subseteq \ControlStates \times \StackAlpha \times \ControlStates
\times \StackAlpha^*$.
In these texts, $(\qstate,\stackchar,\qstate',\vec{\stackchar}) \in
\transfunction$ means, ``if in control state $\qstate$ while the
character $\stackchar$ is on top, pop the stack, transition to
control state $\qstate'$ and push $\vec{\stackchar}$.''
Clearly, we can convert between these two representations by
introducing extra control states to our representation when it needs
to push multiple characters.

\subsection{Rooted pushdown systems}

A \defterm{rooted pushdown system} is a quadruple 
$(\ControlStates,\StackAlpha,\transfunction,\qstate_0)$ in which 
$(\ControlStates,\StackAlpha,\transfunction)$ is a pushdown system and
$\qstate_0 \in \ControlStates$ is an initial (root) state.
$\mathbb{RPDS}$  is the class of all rooted pushdown
systems.
For a rooted pushdown system $M =
(\ControlStates,\StackAlpha,\transfunction,\qstate_0)$, we define 
the \defterm{reachable-from-root transition relation}:
\begin{equation*}
\conf \overset{\stackact}{\underset{M}{\RPDTrans}} \conf' \text{ iff } 
(\qstate_0,\vect{})
\mathrel{\underset{M}{\PDTrans}^*} 
\conf 
\text{ and }
\conf 
\mathrel{\overset{\stackact}{\underset{M}{\PDTrans}}}
\conf'
\text.
\end{equation*}
In other words, the root-reachable transition relation also makes
sure that the root control state can actually reach the transition.

We overload the root-reachable transition relation to operate on
control states:
\begin{equation*}
\qstate 
\mathrel{\overset{\stackact}{\underset{M}{\RPDTrans}}}
\qstate' \text{ iff }
(\qstate,\vec{\stackchar}) 
\mathrel{\overset{\stackact}{\underset{M}{\RPDTrans}}}
(\qstate',\vecp{\stackchar}) 
\text{ for some stacks }
\vec{\stackchar},
\vecp{\stackchar}
\text.
\end{equation*}
For both root-reachable relations, if we elide the stack-action label,
then, as in the un-rooted case, the transition holds if \emph{there
  exists} some stack action that enables the transition:
  \begin{align*}
  \qstate 
  \mathrel{\underset{M}{\RPDTrans}}
  \qstate' \text{ iff }
  \qstate 
  \mathrel{\overset{\stackact}{\underset{M}{\RPDTrans}}}
  \qstate' 
  \text{ for some action } \stackact
  \text.
  \end{align*}

\subsection{Computing reachability in pushdown systems}

A pushdown flow analysis 
can be construed as
  computing the \emph{root-reachable}
  subset of control states in a rooted pushdown system, 
  $M = (\ControlStates,\StackAlpha,\transfunction,\qstate_0)$:
  \begin{align*}
  \setbuild{\qstate}{
    \qstate_0 
      \mathrel{\underset{M}{\RPDTrans}}
    \qstate
  }\text.
\end{align*}
Reps~\emph{et. al} and many others provide a straightforward ``summarization'' algorithm
to compute this set~\cite{mattmight:Bouajjani:1997:PDA-Reachability,dvanhorn:Kodumal2004Set,mattmight:Reps:1998:CFL,mattmight:Reps:2005:Weighted-PDA}.
We will develop a complete alternative to summarization, and then instrument this development for introspective pushdown systems.
Summarization builds two large tables: 
\begin{itemize}
\item{One maps ``calling contexts'' to ``return sites'' (AKA ``local
  continuations'') so that a returning function steps to all the
  places it must return to.}
\item{The other maps ``calling contexts'' to
  ``return states,'' so that any place performing a call with an
  already analyzed calling context can jump straight to the returns.}
\end{itemize}
This setup requires intimate knowledge of the language in question for where continuations should be segmented to be ``local'' and is strongly tied to function call and return.
Our algorithm is based on graph traversals of the transition relation for a generic pushdown system.
It requires no specialized knowledge of the analyzed language, and it avoids the memory footprint of summary tables.

\subsection{Pushdown automata}
A \defterm{pushdown automaton} is an input-accepting generalization
of a rooted pushdown system, a 7-tuple
$(\QStates,\Alphabet,\StackAlpha,\transition,\qstate_0,\FStates,\vec{\stackchar})$ in which:
\begin{enumerate}
\item $\Alphabet$ is an input alphabet; 

\item $\transfunction \subseteq \ControlStates \times \StackAlpha_\pm
  \times (\Alphabet \union \set{\epsilon}) \times \ControlStates$ is a
  transition relation; 

\item $\FStates \subseteq \QStates$ is a set of accepting states; and

\item $\vec{\stackchar} \in \StackAlpha^*$ is the initial stack.

\end{enumerate}
We use $\mathbb{PDA}$ to denote the class of all pushdown automata.

Pushdown automata recognize languages over their input alphabet.
To do so, their transition relation may optionally consume an input
character upon transition.
Formally, a PDA $M = (\QStates,\Alphabet,\StackAlpha,\transition,\qstate_0,\FStates,\vec{\stackchar})$
recognizes the language $\Lang(M) \subseteq \Alphabet^*$:
\begin{align*}
\epsilon \in \Lang(M) &\text{ if }
\qstate_0 \in \FStates
%
\\
aw \in \Lang(M) &\text{ if }
\transition(\qstate_0,\stackchar_+,a,\qstate')
\text{ and }
w \in \Lang(\QState,\Alphabet,\StackAlpha,\transition,\qstate', \FStates, \stackchar : \vec{\stackchar})
\\
aw \in \Lang(M) &\text{ if }
\transition(\qstate_0,\epsilon,a,\qstate')
\text{ and }
w \in \Lang(\QState,\Alphabet,\StackAlpha,\transition,\qstate', \FStates, \vec{\stackchar})
\\
aw \in \Lang(M) &\text{ if }
\transition(\qstate_0,\stackchar_-,a,\qstate')
\text{ and }
w \in \Lang(\QState,\Alphabet,\StackAlpha,\transition,\qstate', \FStates, \vec{\stackchar}')
\\
&\text{ where } \vec{\stackchar} = \vect{\stackchar,\stackchar_2,\ldots,\stackchar_n}
\text{ and } \vec{\stackchar}' = \vect{\stackchar_2,\ldots,\stackchar_n}
\text,
\end{align*}
where $a$ is either the empty string $\epsilon$ or a single character.

\subsection{Nondeterministic finite automata}
In this work, we will need a finite description of 
all possible stacks at a given control state within
a rooted pushdown system.
We will exploit the fact that the set of stacks
at a given control point is a regular language.
Specifically, we will extract a nondeterministic finite automaton
accepting that language from the structure
of a rooted pushdown system.
A \defterm{nondeterministic finite automaton} (NFA) is a quintuple
$M = (\QStates, \Alphabet, \transfunction, \qstate_0, \FStates)$:
\begin{itemize}
\item $\ControlStates$ is a finite set of control states;

\item $\Alphabet$ is an input alphabet; 

\item $\transfunction 
  \subseteq
\ControlStates \times (\Alphabet \union \set{\epsilon}) 
\times \ControlStates$ 
is a transition relation.

\item $\qstate_0$ is a distinguished start state.

\item $\FStates \subseteq \QStates$ is a set of accepting states.
\end{itemize}
We denote the class of all NFAs as $\mathbb{NFA}$.

\section{Setting: A-Normal Form 
$\lambda$-Calculus}
\label{sec:anf}

Since our goal is analysis of
\emph{higher-order languages}, we operate on the
\lc{}.
To simplify presentation of the concrete and abstract semantics, we choose 
A-Normal Form \lc{}.
(This is a strictly cosmetic
 choice: all of our results can be replayed \emph{mutatis mutandis} in
 the standard direct-style setting as well.
 This differs from CFA2's requirement of CPS, since ANF can be applied locally whereas CPS requires a global transformation.)
ANF enforces an order of evaluation and it requires
that all arguments to a function be atomic:
\begin{align*}
\expr \in \syn{Exp} &\produces \letiform{\vv}{\call}{\expr} && \text{[non-tail call]}
\\
&\;\;\opor\;\; \call && \text{[tail call]}
\\
&\;\;\opor\;\; \aexpr && \text{[return]}
\\
\fexpr,\aexpr \in \syn{Atom} &\produces \vv \opor \lam && \text{[atomic expressions]}
\\
\lam \in \syn{Lam} &\produces \lamform{\vv}{\expr} && \text{[lambda terms]}
\\
\call \in \syn{Call} &\produces \appform{\fexpr}{\aexpr} && \text{[applications]}
\\
\vv \in \syn{Var} &\text{ is a set of identifiers} && \text{[variables]}
\text.
\end{align*}

We use the CESK machine of \citet{mattmight:Felleisen:1987:CESK} to specify a small-step semantics
for ANF.
The CESK machine has an explicit stack, and under a structural abstraction, the stack
component of this machine directly becomes the stack component of a
pushdown system.
The set of configurations ($\s{Conf}$) for 
this machine has the four expected components (Figure~\ref{fig:cesk}).
\begin{figure}
\figrule
\begin{align*}
\conf \in \s{Conf} &= \syn{Exp} \times \s{Env} \times \s{Store} \times \s{Kont} && \text{[configurations]}
\\
\env \in \s{Env} &= \syn{Var} \parto \s{Addr} && \text{[environments]}
\\
\store \in \s{Store} &= \s{Addr} \to \s{Clo} && \text{[stores]}
\\
\clo \in \s{Clo} &= \syn{Lam} \times \s{Env} && \text{[closures]}
\\
\cont \in \s{Kont} &= \s{Frame}^* && \text{[continuations]}
\\
\phrame \in \s{Frame} &= \syn{Var} \times \syn{Exp} \times \s{Env}  && \text{[stack frames]}
\\
\addr \in \s{Addr} &\text{ is an infinite set of addresses} && \text{[addresses]}
\text.
\end{align*}
\captionsetup{justification=centering}
\caption{The concrete configuration-space.}
\label{fig:cesk}
\figrule
\end{figure}

\subsection{Semantics}

To define the semantics, we need five items:
\begin{enumerate}
\item $\Inject : \syn{Exp} \to \s{Conf}$ injects an expression into
a configuration:
\begin{equation*}
\conf_0 = \Inject(\expr) = (\expr, [], [], \vect{})
\text.
\end{equation*}

\item{%
$\ArgEval : \syn{Atom} \times \s{Env} \times \s{Store} \parto
\s{Clo}$ evaluates atomic expressions:
\begin{align*}
  \ArgEval(\lam,\env,\store) &= (\lam,\env) && \text{[closure creation]}
  \\
  \ArgEval(\vv,\env,\store) &= \store(\env(\vv)) && \text{[variable look-up]}
  \text.
\end{align*}}

\item{%
$(\To) \subseteq \s{Conf} \times \s{Conf}$ transitions between
configurations. (Defined below.)}

\item{%
 $\Eval : \syn{Exp} \to \Pow{\s{Conf}}$ computes the set of
 reachable machine configurations for a given program:
 \begin{equation*}
   \Eval(\expr) = \setbuild{ \conf }{ \Inject(\expr) \To^* \conf } 
   \text.
 \end{equation*}}

\item{%
$\alloc : \syn{Var} \times \s{Conf} \to \s{Addr}$ 
chooses fresh store addresses for newly bound variables.
The address-allocation function is an opaque parameter in this semantics,
so that the forthcoming abstract semantics may also parameterize allocation. 
The nondeterministic nature of the semantics makes any choice of $\alloc$ sound \citep{mattmight:Might:2009:APosteriori}.
This parameterization provides the knob to tune the polyvariance and context-sensitivity of the resulting analysis.
For the sake of defining the concrete semantics, letting addresses be natural numbers suffices.
The allocator can then choose the lowest unused address:
\begin{align*}
  \s{Addr} &= \mathbb{N}
  \\
  \alloc(v,(\expr,\env,\store,\cont)) &= 
  1 + \max(\dom(\store))
  \text.
\end{align*}}
\end{enumerate}




\paragraph{Transition relation}
To define the transition $\conf \To \conf'$, we need three rules.
The first rule handles tail calls by evaluating the function into a closure, evaluating the argument into a value and then moving to the body of the closure's \lamterm{}:
\begin{center}
  \minipagebreak{0.50}[-1cm]{
  \overbrace{(\sembr{\appform{\fexpr}{\aexpr}}, \env, \store, \cont)}^{\conf}
  &\To
  \overbrace{(\expr,\env'',\store',\cont)}^{\conf'}
  \text{, where }
  \\
  (\sembr{\lamform{\vv}{\expr}}, \env') &= \ArgEval(\fexpr,\env,\store)}
  {0.45}[1mm]{%
  \addr &= \alloc(\vv,\conf)
  \\
  \env'' &= \env'[\vv \mapsto \addr]
  \\
  \store' &= \store[\addr \mapsto \ArgEval(\aexpr,\env,\store)]
  \text.}
\end{center}

\noindent
Non-tail calls push a frame onto the stack and evaluate the call:
\begin{align*}
  \overbrace{(\sembr{\letiform{\vv}{\call}{\expr}}, \env, \store, \cont)}^{\conf}
  &\To
  \overbrace{(\call,\env,\store, (\vv,\expr,\env) : \cont)}^{\conf'}
  \text.
\end{align*}

\noindent
Function return pops a stack frame:
\begin{center}
  \minipagebreak{0.50}[-1.1cm]{%
    \overbrace{(\aexpr, \env, \store, (\vv,\expr,\env') :
      \cont)}^{\conf} &\To \overbrace{(\expr,\env'',\store',
      \cont)}^{\conf'} \text{, where }}
  {0.45}[5mm]{%
    \addr &= \alloc(\vv,\conf)
    \\
    \env'' &= \env'[\vv \mapsto \addr]
    \\
    \store' &= \store[\addr \mapsto \ArgEval(\aexpr,\env,\store)]
    \text.}
\end{center}
\section{An Infinite-State Abstract Interpretation}
\label{sec:abstraction}
Our first step toward a static analysis 
is an abstract interpretation
into an \emph{infinite} state-space.
To achieve a pushdown analysis,
we simply
abstract away less than we normally would.
Specifically, 
we leave the stack height unbounded.

Figure~\ref{fig:abs-conf-space}
details the 
abstract configuration-space.
To synthesize it, we force addresses to be a finite set, but
crucially, we leave the stack untouched.
When we compact the set of addresses into a finite set, the machine
may run out of addresses to allocate, and when it does, the
pigeon-hole principle will force multiple closures to reside at the
same address.
As a result, to remain sound we change the range of the store to
become a power set in the abstract configuration-space.
The abstract transition relation has  components
analogous to those from the concrete semantics:

\begin{figure}
\figrule
\begin{align*}
\aconf \in \sa{Conf} &= \syn{Exp} \times \sa{Env} \times \sa{Store} \times \sa{Kont} && \text{[configurations]}
\\
  \aenv \in \sa{Env} &= \syn{Var} \parto \sa{Addr} && \text{[environments]}
  \\
    \astore \in \sa{Store} &= \sa{Addr} \to \Pow{\sa{Clo}} && \text{[stores]}
 \\
   \aclo \in \sa{Clo} &= \syn{Lam} \times \sa{Env} && \text{[closures]}
 \\
 \acont \in \sa{Kont} &= \sa{Frame}^* && \text{[continuations]}
 \\
 \aphrame \in \sa{Frame} &= \syn{Var} \times \syn{Exp} \times \sa{Env}  && \text{[stack frames]}
 \\
\aaddr \in \sa{Addr} &\text{ is a \emph{finite} set of addresses} && \text{[addresses]}
\text.
\end{align*}
\captionsetup{justification=centering}
\caption{The abstract configuration-space.}
\label{fig:abs-conf-space}
\figrule
\end{figure}

              \paragraph{Program injection}
              The abstract injection function $\aInject : \syn{Exp} \to \sa{Conf}$
              pairs an expression with an empty environment, an empty store and an
              empty stack to create the initial abstract configuration:
              \begin{equation*}
\aconf_0 = \aInject(\expr) = (\expr, [], [], \vect{})
  \text.
  \end{equation*}

\paragraph{Atomic expression evaluation}
The abstract atomic expression evaluator, $\aArgEval : \syn{Atom}
\times \sa{Env} \times \sa{Store} \to \PowSm{\sa{Clo}}$, returns the value of
an atomic expression in the context of an environment and a store;
it returns a \emph{set} of abstract closures:
\begin{align*}
\aArgEval(\lam,\aenv,\astore) &= \set{(\lam,\aenv)} && \text{[closure creation]}
\\
\aArgEval(\vv,\aenv,\astore) &= \astore(\aenv(\vv)) && \text{[variable look-up]}
\text.
\end{align*}

\paragraph{Reachable configurations}
The abstract program evaluator $\aEval : \syn{Exp} \to
\PowSm{\sa{Conf}}$ returns all of the configurations reachable from
the initial configuration:
\begin{equation*}
\aEval(\expr) = \setbuild{ \aconf }{ \aInject(\expr) \aTo^* \aconf } 
\text.
\end{equation*}
Because there are an infinite number of abstract configurations, a
na\"ive implementation of this function may not terminate.
Pushdown analysis provides a way of precisely computing this set and both finitely and compactly representing the result.

\paragraph{Transition relation}
The abstract transition relation $(\aTo) \subseteq \sa{Conf} \times
\sa{Conf}$ has three rules, one of which has become nondeterministic.
A tail call may fork because there could be multiple abstract closures
that it is invoking:
\begin{center}
\minipagebreak{0.50}{%
  \overbrace{(\sembr{\appform{\fexpr}{\aexpr}}, \aenv, \astore,
      \acont)}^{\aconf} &\aTo
    \overbrace{(\expr,\aenv'',\astore',\acont)}^{\aconf'} \text{,
      where } \\
    (\sembr{\lamform{\vv}{\expr}}, \aenv') &\in
    \aArgEval(\fexpr,\aenv,\astore)}
  {0.45}[1cm]{%
    \aaddr &= \aalloc(\vv,\aconf)
    \\
    \aenv'' &= \aenv'[\vv \mapsto \aaddr]
    \\
    \astore' &= \astore \join [\aaddr \mapsto
    \aArgEval(\aexpr,\aenv,\astore)] \text.}
\end{center}
We define all of the partial orders shortly, but for stores:
\begin{equation*}
(\astore \join \astore')(\aaddr) = \astore(\aaddr) \union \astore'(\aaddr)
\text.
\end{equation*}

\noindent
A non-tail call pushes a frame onto the stack and evaluates the call:
\begin{align*}
\overbrace{(\sembr{\letiform{\vv}{\call}{\expr}}, \aenv, \astore, \acont)}^{\aconf}
&\aTo
\overbrace{(\call,\aenv,\astore, (\vv,\expr,\aenv) : \acont)}^{\aconf'} 
\text.
\end{align*}

\noindent
A function return pops a stack frame:
\begin{center}
  \minipagebreak{0.50}[-1.1cm]{%
    \overbrace{(\aexpr, \aenv, \astore, (\vv,\expr,\aenv') :
      \acont)}^{\aconf} &\aTo \overbrace{(\expr,\aenv'',\astore',
      \acont)}^{\aconf'} \text{, where }}
  {0.45}[5mm]{%
    \aaddr &= \aalloc(\vv,\aconf)
    \\
    \aenv'' &= \aenv'[\vv \mapsto \aaddr]
    \\
    \astore' &= \astore \join [\aaddr \mapsto
    \aArgEval(\aexpr,\aenv,\astore)] \text.}
\end{center}
\paragraph{Allocation: Polyvariance and context-sensitivity}
\label{sec:polyvariance}
In the abstract semantics, the abstract allocation function
$\aalloc : \syn{Var} \times \sa{Conf} \to \sa{Addr}$ determines the
polyvariance of the analysis.
In a control-flow analysis, \emph{polyvariance} literally refers to
the number of abstract addresses (variants) there are for each
variable.
An advantage of this framework over CFA2 is
that varying this abstract allocation function
instantiates pushdown versions of classical flow analyses.
All of the following allocation approaches can be used with the
abstract semantics. Note, though only a technical detail, that the concrete address space and allocation would change as well for the abstraction function to still work.
The abstract allocation function is a
parameter to the analysis.

\paragraph{Monovariance: Pushdown 0CFA}

Pushdown 0CFA uses variables themselves for abstract addresses:

\begin{align*}
\sa{Addr} &= \syn{Var}
\\
  \alloc(v,\aconf) &= v
  \text.
  \end{align*}

For better precision, a program would be transformed to have unique binders.

  \paragraph{Context-sensitive: Pushdown 1CFA}

  Pushdown 1CFA pairs the variable with the current expression to get
  an abstract address:

  \begin{align*}
  \sa{Addr} &= \syn{Var} \times \syn{Exp}
  \\
    \alloc(\vv,(\expr,\aenv,\astore,\acont)) &= (\vv,\expr)
    \text.
    \end{align*}

  For better precision, expressions are often uniquely labeled so that textually equal expressions at different points in the program are distinguished.

    \paragraph{Polymorphic splitting: Pushdown poly/CFA}

    Assuming we compiled the program from a programming language with
    let expressions and we marked which identifiers were let-bound, we
    can enable polymorphic splitting:

    \begin{align*}
    \sa{Addr} &= \syn{Var} + \syn{Var} \times \syn{Exp} 
    \\
      \alloc(\vv,(\sembr{\appform{\fexpr}{\aexpr}},\aenv,\astore,\acont)) &= 
      \begin{cases}
      (\vv,\sembr{\appform{\fexpr}{\aexpr}}) & \fexpr \text{ is let-bound}
      \\
        \vv & \text{otherwise}
        \text.
        \end{cases}
        \end{align*}

        \paragraph{Pushdown $k$-CFA}

        For pushdown $k$-CFA, we need to look beyond the current state
        and at the last $k$ states, necessarily changing the signature of $\aalloc$ to $\syn{Var} \times \sa{Conf}^* \to \sa{Addr}$.
        By concatenating the expressions in the last $k$ states together, and
        pairing this sequence with a variable we get pushdown $k$-CFA:
        \begin{align*}
        \sa{Addr} &= \syn{Var} \times \syn{Exp}^k
        \\
          \aalloc(\vv,\vect{(\expr_1,\aenv_1,\astore_1,\acont_1),\ldots}) &=
(\vv,\vect{\expr_1,\ldots,\expr_k})
  \text.
  \end{align*}

  \subsection{Partial orders}

  For each set $\hat X$ inside the abstract configuration-space, we use the
  natural partial order, $(\wt_{\hat X}) \subseteq \hat X \times \hat X$.
  Abstract addresses and syntactic sets have flat partial orders.
  For the other sets, the
  partial order lifts:
  \begin{itemize}
  \item point-wise over environments:
  \begin{equation*}
  \aenv \wt \aenv' \text{ iff }
  \aenv(\vv) = \aenv'(\vv) \text{ for all } \vv \in \dom(\aenv)
  \text;
  \end{equation*}

  \item component-wise over closures:
  \begin{equation*}
  (\lam,\aenv) \wt (\lam,\aenv') \text{ iff } \aenv \wt \aenv'
  \text;
  \end{equation*}

  \item point-wise over stores:
  \begin{equation*}
  \astore \wt \astore' 
  \text{ iff }
  \astore(\aaddr) \wt \astore'(\aaddr) \text{ for all } \aaddr
\in \dom(\astore)
  \text;
  \end{equation*}

  \item component-wise over frames:
  \begin{equation*}
  (\vv,\expr,\aenv) \wt (\vv,\expr,\aenv') \text{ iff } \aenv \wt \aenv'
  \text;
  \end{equation*}

  \item element-wise over continuations:
  \begin{equation*}
  \vect{\aphrame_1,\ldots,\aphrame_n} 
  \wt 
  \vect{\aphrame'_1,\ldots,\aphrame'_n} 
  \text{ iff }
  \aphrame_i \wt \aphrame'_i
  \text{; and }
  \end{equation*}

  \item component-wise across configurations:
  \begin{equation*}
(\expr,\aenv,\astore,\acont) 
  \wt
  (\expr,\aenv',\astore',\acont')
  \text{ iff }
  \aenv \wt \aenv'
  \text{ and }
  \astore \wt \astore'
  \text{ and }
  \acont \wt \acont'
  \text.
  \end{equation*}
  \end{itemize}


  \subsection{Soundness}

  To prove soundness,  an abstraction map $\absmap$ connects the
  concrete and abstract configuration-spaces:
  \begin{align*}
\absmap(\expr,\env,\store,\cont) &= (\expr,\absmap(\env),\absmap(\store),\absmap(\cont))
  \\
    \absmap(\env) &= \lt{\vv}{\absmap(\env(\vv))}
    \\
      \absmap(\store) &= \lt{\aaddr}{\!\!\! \bigjoin_{\absmap(\addr) = \aaddr} \!\!\! \set{\absmap(\store(\addr))}}
      \\
        \absmap\vect{\phrame_1,\ldots,\phrame_n} &= \vect{\absmap(\phrame_1), \ldots, \absmap(\phrame_n)}
        \\
          \absmap(\vv,\expr,\env) &= (\vv,\expr,\absmap(\env))
          \\
            \absmap(\addr) &\text{ is determined by the allocation functions}
            \text.
            \end{align*}
            It is then easy to prove that the abstract transition relation
            simulates the concrete transition relation:
            \begin{theorem}
If $\absmap(\conf) \wt \aconf \text{ and } \conf \To \conf'$, then
there exists $\aconf' \in \sa{Conf}$ such that $\absmap(\conf')
\wt \aconf' \text{ and } \aconf \aTo \aconf'$.
  \end{theorem}
  \begin{proof}
  The proof follows by case analysis on the
  expression in the configuration.
  It is a straightforward adaptation of similar proofs, such as that of
  \citet{mattmight:Might:2007:Dissertation} for $k$-CFA.
  \end{proof}

\section{From the Abstracted CESK Machine
  to a PDA}
\label{sec:pda}
In the previous section, we constructed an infinite-state abstract
interpretation of the CESK machine.
The infinite-state nature of the abstraction makes it difficult to see how
to answer static analysis questions.
Consider, for instance, a control flow-question:
\begin{center}
\emph{  At the call site $\appform{\fexpr}{\aexpr}$, may a closure over
  $\lam$ be called?}
\end{center}
If the abstracted CESK machine were a finite-state machine, an
algorithm could answer this question by enumerating all reachable
configurations and looking for an abstract configuration
$(\sembr{\appform{\fexpr}{\aexpr}},\aenv,\astore,\acont)$ in which
$(\lam,\_) \in \aArgEval(\fexpr,\aenv,\astore)$.
However, because the abstracted CESK machine may contain an infinite
number of reachable configurations, an algorithm cannot enumerate
them.

Fortunately, we can recast the abstracted CESK as a special kind of
infinite-state system: a pushdown automaton (PDA).
Pushdown automata occupy a sweet spot in the theory of computation:
they have an infinite configuration-space, yet many useful properties
(\eg, word membership, non-emptiness, control-state reachability)
remain decidable.
Once the abstracted CESK machine becomes a PDA, we can answer the
control-flow question by checking whether a specific regular language,
accounting for the states of interest, when intersected with the
language of the PDA, is nonempty.

The recasting as a PDA is a shift in perspective.
A configuration has an expression, an environment and a store. 
A stack character is a frame.
%
We choose to make the alphabet the set of control states, so that the
language accepted by the PDA will be sequences of control-states
visited by the abstracted CESK machine.
Thus, every transition will consume the control-state to which it
transitioned as an input character.
Figure~\ref{fig:acesk-to-pda} defines the program-to-PDA conversion
function $\afPDA : \syn{Exp} \to \mathbb{PDA}$.  (Note the implicit
use of the isomorphism $\QState \times \sa{Kont} \cong \sa{Conf}$.)

\begin{figure}
\figrule
  \minipagebreak{0.45}[-5mm]{
      \afPDA(\expr) &=
      (\QStates,\Alphabet,\StackAlpha,\transfunction,\qstate_0,\FStates,\vect{})
      \text{, where }
      \\
      \QStates &= \syn{Exp} \times \sa{Env} \times \sa{Store}
      \\
      \Alphabet &= \QStates
      \\
      \StackAlpha &= \sa{Frame}}
{0.50}
{(\qstate,\epsilon,\qstate',\qstate') \in \transfunction
  & \text{ iff }
  (\qstate, \acont)
  \aTo
  (\qstate', \acont)
  \text{ for all } \acont
  \\
  (\qstate,\aphrame_{-},\qstate',\qstate') \in \transfunction
  & \text{ iff }
  (\qstate, \aphrame : \acont)
  \aTo
  (\qstate',\acont)
  \text{ for all } \acont
  \\
  (\qstate,\aphrame'_{+},\qstate',\qstate') \in \transfunction
  & \text{ iff }
  (\qstate, \acont)
  \aTo
  (\qstate',\aphrame' : \acont)
  \text{ for all } \acont
  \\
  (\qstate_0,\vect{}) &= \aInject(\expr)
  \\
  \FStates &= \QStates
  \text.}
\captionsetup{justification=centering}
\caption{$\afPDA : \syn{Exp} \to \mathbb{PDA}$.
}
\label{fig:acesk-to-pda}
\figrule
\end{figure}

At this point, we can answer questions about whether a specified
control state is reachable by formulating a question about the
intersection of a regular language with a context-free language
described by the PDA.
That is, if we want to know whether the control state
$(\expr',\aenv,\astore)$ is reachable in a program $\expr$, we can reduce the problem to determining:
\begin{equation*}
  \Alphabet^*  \cdot \set{ (\expr',\aenv,\astore) } \cdot \Alphabet^* \mathrel{\cap} \Lang(\afPDA(\expr))
  \neq \emptyset
  \text,
\end{equation*}
where $L_1 \cdot L_2$ is the concatenation of formal languages $L_1$ and $L_2$.

\begin{theorem}
  Control-state reachability is decidable.
\end{theorem}
\begin{proof}
  The intersection of a regular language and a context-free language
  is context-free (simple machine product of PDA with DFA).
  The emptiness of a context-free language is decidable.
  The decision procedure is easiest for CFGs: mark terminals, mark non-terminals that reduce to marked (non)terminals until we reach a fixed point. If the start symbol is marked, then the language is nonempty.
  The PDA to CFG translation is a standard construction.
\end{proof}

Now, consider how to use control-state reachability to answer the
control-flow question from earlier.
There are a finite number of possible control states in which the
\lamterm{} $\lam$ may flow to the function $\fexpr$ in call site
$\appform{\fexpr}{\aexpr}$; let's call this set of states $\hat
S$:
\begin{equation*}
  \hat S = 
  \setbuild{ (\sembr{\appform{\fexpr}{\aexpr}}, \aenv, \astore) }{
    (\lam,\aenv') \in \aArgEval(\fexpr,\aenv,\astore)
    \text{ for some } \aenv'
  }
  \text.
\end{equation*}
What we want to know is whether any state in the set $\hat S$ is
reachable in the PDA.
In effect what we are asking is whether there exists a control state $\qstate \in
\hat S$ such that:
\begin{align*}
  \Alphabet^* \cdot \set{\qstate} \cdot \Alphabet^*
  \mathrel{\cap}
  \Lang(\afPDA(\expr))
  \neq 
  \emptyset
  \text.
\end{align*}
If this is true, then $\lam$ may flow to $\fexpr$; if false, then it
does not.

\paragraph{Problem: Doubly exponential complexity}
\label{sec:doubly-exponential}
The non-emptiness-of-intersection approach establishes decidability of
pushdown control-flow analysis.
But, two exponential complexity barriers make this technique
impractical.

First, there are an exponential number of both environments
($\abs{\sa{Addr}}^{\abs{\syn{Var}}}$) and stores ($(2^{\abs{\sa{Clo}}})^{\abs{\sa{Addr}}} =
2^{\abs{\sa{Clo}} \times {\abs{\sa{Addr}}}}$) to consider for the set $\hat S$.
On top of that, computing the intersection of a regular language with
a context-free language will require enumeration of the
(exponential) control-state-space of the PDA.
The size of the control-state-space of the PDA is clearly doubly exponential:
\begin{align*}
 \abs{Q} 
 &= 
    \abs{\mathsf{Exp} \times \sa{Env} \times \sa{Store}}
 \\
 &= 
    \abs{\mathsf{Exp}} \times \abs{\sa{Env}} \times \abs{\sa{Store}}
 \\
 &=
    \abs{\mathsf{Exp}} 
    \times \abs{\sa{Addr}}^{\abs{\syn{Var}}}
    \times 2^{\abs{\sa{Clo}} \times {\abs{\sa{Addr}}}}
 \\   
 &=
    \abs{\mathsf{Exp}} 
    \times \abs{\sa{Addr}}^{\abs{\syn{Var}}}
    \times 2^{\abs{\mathsf{Lam} \times \sa{Env}} \times {\abs{\sa{Addr}}}}
 \\
 &=
    \abs{\mathsf{Exp}} 
    \times \abs{\sa{Addr}}^{\abs{\syn{Var}}}
    \times 2^{\abs{\mathsf{Lam}} \times \abs{\sa{Addr}}^{\abs{\syn{Var}}} \times {\abs{\sa{Addr}}}}
\end{align*}
%
As a result, this approach is doubly exponential.
For the next few sections, our goal will be to lower the complexity of
pushdown control-flow analysis.

\section{Focusing on Reachability}
\label{sec:pdreachability}

In the previous section, we saw that control-flow analysis reduces to
the reachability of certain control states within a pushdown system.
We also determined reachability by converting the
abstracted CESK machine into a PDA, and using emptiness-testing on a
language derived from that PDA.
Unfortunately, we also found that this approach is deeply exponential.

Since control-flow analysis reduced to the reachability of
control-states in the PDA, we skip the language problems and go
directly to reachability algorithms of
\citet{mattmight:Bouajjani:1997:PDA-Reachability,dvanhorn:Kodumal2004Set,mattmight:Reps:1998:CFL}
and \citet{mattmight:Reps:2005:Weighted-PDA} that determine the
reachable \emph{configurations} within a pushdown system.
These algorithms are even polynomial-time.
Unfortunately, some of them are polynomial-time in the number of
control states, and in the abstracted CESK machine, there are an
exponential number of control states.
We don't want to \emph{enumerate} the entire control state-space, or
else the search becomes exponential in even the best case.

To avoid this worst-case behavior, we present a straightforward
pushdown-reachability algorithm that considers only the
\emph{reachable} control states.
We cast our reachability algorithm as a fixed-point iteration, in
which we incrementally construct the reachable subset of a pushdown system.
A rooted pushdown system $M = (\QStates, \StackAlpha, \transfunction, \qstate_0)$ is \emph{compact}
if for any $(\qstate,\stackact,\qstate') \in \transfunction$, it is the case that:
\begin{equation*}
  (\qstate_0,\vect{}) \mathrel{\underset{M}{\PDTrans}^*} (\qstate,\vec{\stackchar})
  \text{ for some stack } \vec{\stackchar}
  \text,
\end{equation*}
and the domain of states and stack characters are exactly those that appear in $\transfunction$:
\begin{itemize}
\item[]{$\QStates = \bigcup\setbuild{\set{\qstate,\qstate'}}{(\qstate,\stackact,\qstate') \in \transfunction}$}
\item[]{$\StackAlpha = \setbuild{\stackchar}{(\qstate,\stackchar_+,\qstate') \in \transfunction \text{ or } (\qstate,\stackchar_-,\qstate') \in \transfunction}$}
\end{itemize}
In other words, a rooted pushdown system is compact when its states, transitions and stack characters appear on legal paths from the initial control state.
We will refer to the class of compact rooted pushdown systems as $\mathbb{CRPDS}$.

We can compact a rooted pushdown system with a map:
\begin{align*}
  \fCRPDS &: \mathbb{RPDS} \to \mathbb{CRPDS} \\
  \fCRPDS\overbrace{(\QStates, \StackAlpha, \transfunction, \qstate_0)}^M &= (\QStates', \StackAlpha', \transfunction', \qstate_0) \\
  \text{where } \QStates' &= \setbuild{\qstate}{(\qstate_0,\vect{}) \mathrel{\underset{M}{\PDTrans}^*} (\qstate,\vec{\stackchar})} \\
  \StackAlpha' &= \setbuild{\stackchar_i}{(\qstate_0,\vect{}) \mathrel{\underset{M}{\PDTrans}^*} (\qstate,\vec{\stackchar})} \\
  \transfunction' &= \setbuild{(\qstate,\stackact,\qstate')}{(\qstate_0,\vect{}) \mathrel{\overset{\vec{\stackact}}{\underset{M}{\PDTrans}}} (\qstate,[\vec{\stackact}]) \mathrel{\overset{\stackact}{\underset{M}{\PDTrans}}} (\qstate',[\vec{\stackact}\stackact]) }\text.
\end{align*}

In practice, the real difference between a rooted pushdown system and
its compact form is that our original system will be defined
intensionally (having come from the components of an abstracted CESK
machine), whereas the compact system will be defined extensionally,
with the contents of each component explicitly enumerated during its construction.

Our near-term goals are (1) to convert our abstracted CESK machine into
a rooted pushdown system and (2) to find an \emph{efficient} method
to compact it.

To convert the abstracted CESK machine into a rooted pushdown system,
we use the function $\afRPDS : \syn{Exp} \to \mathbb{RPDS}$:

\begin{center}
  \minipagebreak{0.45}{%
      \afRPDS(\expr) &=
      (\QStates,\StackAlpha,\transfunction,\qstate_0)
      \\
      \QStates &= \syn{Exp} \times \sa{Env} \times \sa{Store}
      \\
      \StackAlpha &= \sa{Frame} 
      \\
      (\qstate_0,\vect{}) &= \aInject(\expr)}
  {0.50}{%
     \qstate \pdedge^\epsilon \qstate' \in \transfunction & \text{ iff
     } (\qstate, \acont) \aTo (\qstate', \acont) \text{ for all }
     \acont
     \\
     \qstate \pdedge^{\aphrame_{-}} \qstate'
     \in \transfunction & \text{ iff } (\qstate, \aphrame : \acont)
     \aTo (\qstate',\acont) \text{ for all } \acont
     \\
     \qstate \pdedge^{\aphrame_{+}} \qstate' \in \transfunction &
     \text{ iff } (\qstate, \acont) \aTo (\qstate',\aphrame : \acont)
     \text{ for all } \acont
     \text.}
\end{center}

\section{Compacting a Rooted Pushdown System}
\label{sec:rpds-to-crpds}
We now turn our attention to compacting a rooted pushdown
system (defined intensionally) into its compact form (defined
extensionally).
That is, we want to find an implementation of the function $\fCRPDS$.
To do so, we first phrase the construction as the least fixed point of a monotonic function.
This will provide a method (albeit an inefficient one) for computing
the function $\fCRPDS$.
In the next section, we look at an optimized work-set driven
algorithm that avoids the inefficiencies of this section's algorithm.

The function $\mkCRPDS : \mathbb{RPDS} \to (\mathbb{CRPDS} \to
\mathbb{CRPDS})$ generates the monotonic iteration function we need:
\begin{center}
  \minipagebreak{0.40}[-1cm]{%
    \mkCRPDS(M) &= f\text{, where }
    \\
    M &= (\QStates,\StackAlpha,\transfunction,\qstate_0)}
    {0.55}{%
    f(\DSStates,\DSFrames,\DSEdges,\dsstate_0) &=
    (\DSStates',\DSFrames,\DSEdges',\dsstate_0) \text{, where }
    \\
    \DSStates' &= \DSStates \union \setbuild{ \dsstate' }{ \dsstate
      \in \DSStates \text{ and } \dsstate \mathrel{\underset{M}{\RPDTrans}}
      \dsstate' } \union \set{\dsstate_0}
    \\
    \DSEdges' &= \DSEdges \union \setbuild{ \dsstate \pdedge^\stackact
      \dsstate' }{ \dsstate \in \DSStates \text{ and } \dsstate
      \mathrel{\overset{\stackact}{\underset{M}{\RPDTrans}}} \dsstate' } \text.
}
\end{center}
Given a rooted pushdown system $M$, each application of the function
$\mkCRPDS(M)$ accretes new edges at the frontier of the system.
Once the algorithm reaches a fixed point, the system is complete:
\begin{theorem}\label{thm:mkCRPDS-correct}
  $\fCRPDS(M) = \lfp(\mkCRPDS(M))$.
\end{theorem}
\begin{proof}
  Let $M = (\QStates,\StackAlpha,\transfunction,\qstate_0)$.
  Let $f = \mkCRPDS(M)$.
  Observe that $\lfp(f) = f^n(\emptyset,\StackAlpha,\emptyset,\qstate_0)$ for some $n$.
  When $N \subseteq \fCRPDS(M)$, then it easy to show that $f(N) \subseteq \fCRPDS(M)$.
  Hence, $\fCRPDS(M) \supseteq \lfp(\mkCRPDS(M))$.

  To show $\fCRPDS(M) \subseteq \lfp(\mkCRPDS(M))$, suppose this is not
  the case.
  Then, there must be at least one edge in $\fCRPDS(M)$ that is not in
  $\lfp(\mkCRPDS(M))$.
  Since these edges must be root reachable, let $(\dsstate,\stackact,\dsstate')$ be the first such edge in some path from the root.
  This means that the state $\dsstate$ is in $\lfp(\mkCRPDS(M))$.
  Let $m$ be the lowest natural number such that $\dsstate$ appears in $f^m(M)$.
  By the definition of $f$, this edge must appear in $f^{m+1}(M)$, which means it must also appear in 
  $\lfp(\mkCRPDS(M))$, which is a contradiction.
  Hence, $\fCRPDS(M) \subseteq \lfp(\mkCRPDS(M))$.
\end{proof}

\subsection{Complexity: Polynomial and exponential}
\label{sec:compl-pol-exp}
 
To determine the complexity of this algorithm, we ask two questions:
how many times would the algorithm invoke the iteration function in
the worst case, and how much does each invocation cost in the
worst case?
The size of the final system bounds the run-time of the algorithm.
Suppose the final system has $m$ states.
In the worst case, the iteration function adds only a single edge each
time.
Since there are at most $2\abs{\StackAlpha}m^2 + m^2$ edges in the final graph, the maximum
number of iterations is $2\abs{\StackAlpha}m^2 + m^2$.

The cost of computing each iteration is harder to bound.
The cost of determining whether to add a push edge is proportional to
the size of the stack alphabet, while the cost of determining whether
to add an $\epsilon$-edge is constant, so the cost of determining all
new push and $\epsilon$ edges to add is proportional to $\abs{\StackAlpha}m +
m$.
Determining whether or not to add a pop edge is expensive.
To add the pop edge
$\triedge{\dsstate}{\stackchar_-}{\dsstate'}$, we must prove that
there exists a configuration-path to the control state $\dsstate$, in which
the character $\stackchar$ is on the top of the stack.
This reduces to a CFL-reachability query~\cite{mattmight:Melski:2000:CFL}
at each node, the cost of which is $O(\abs{\StackAlpha_\pm}^3
m^3)$~\cite{dvanhorn:Kodumal2004Set}.

To summarize, in terms of the number of reachable control states, the
complexity of the most recent algorithm is:
\[
O((2\abs{\StackAlpha}m^2 + m^2) 
  \times
  (\abs{\StackAlpha}m + m
  + \abs{\StackAlpha_\pm}^3m^3)) 
  =  O(\abs{\StackAlpha}^4m^5)\text.
\]
While this approach is polynomial in the number of reachable
control states, it is far from efficient.
In the next section, we provide an optimized version of this
fixed-point algorithm that maintains a work-set and an
\ecg{} to avoid spurious recomputation.

Moreover, we have carefully phrased the complexity in terms of
``reachable'' control states because, in practice, compact rooted
pushdown systems will be extremely sparse, and because the maximum
number of control states is exponential in the size of the input
program.
After the subsequent refinement, we will be able to develop a hierarchy of
pushdown control-flow analyses that employs widening to achieve a
polynomial-time algorithm at its foundation.

\section{An Efficient Algorithm: Work-sets and 
    $\boldsymbol\epsilon$-Closure Graphs}
\label{sec:ecg-worklist}
We have developed a fixed-point formulation of the rooted pushdown
system compaction algorithm, but found that, in each iteration, it
wasted effort by passing over all discovered states and edges, even
though most will not contribute new states or edges.
Taking a cue from graph search, we can adapt the fixed-point algorithm
with a work-set.
That is, our next algorithm will keep a work-set of new states and
edges to consider, instead of reconsidering all of them.
We will refer to the compact rooted pushdown system we are constructing as a graph, since that is how we represent it ($\QStates$ is the set of nodes, and $\transfunction$ is a set of labeled edges).

In each iteration, it will pull new states and edges from the work
list, insert them into the graph and then populate the
work-set with new states and edges that have to be added as a
consequence of the recent additions.

\subsection{$\boldsymbol\epsilon$-closure graphs}
Figuring out what edges to add as a consequence of another edge
requires care, for adding an edge can have ramifications on distant
control states.
Consider, for example, adding the $\epsilon$-edge $\qstate
\pdedge^{\epsilon} \qstate'$ into the following graph:
\begin{equation*}
  \xymatrix{ 
    \qstate_0 \ar[r]^{\stackchar_+} & \qstate & \qstate' \ar[r]^{\stackchar_-} & \qstate_1
  } 
\end{equation*}
As soon this edge drops in, an $\epsilon$-edge ``implicitly'' appears
 between $\qstate_0$ and $\qstate_1$ because the net
stack change between them is empty; the resulting graph looks like:
\begin{equation*}
  \xymatrix{ 
    \qstate_0 \ar[r]^{\stackchar_+} \ar @{..>} @(u,u) [rrr]^\epsilon &
    \qstate \ar[r]^\epsilon &
    \qstate' \ar[r]^{\stackchar_-} &
    \qstate_1
  }  
\end{equation*}
where we have illustrated the implicit $\epsilon$-edge as a dotted line.

To keep track of these implicit edges, we will construct a second
graph in conjunction with the graph: an $\epsilon$-closure
graph.
In the $\epsilon$-closure graph, every edge indicates the existence of
a no-net-stack-change path between control states.
The $\epsilon$-closure graph simplifies the task of figuring out
which states and edges are impacted by the addition of a new edge.

Formally, an \textbf{$\boldsymbol \epsilon$-closure graph}, $H \subseteq N \times N$, is
a set of edges.
Of course, all \ecg s are reflexive: every node has a self loop.
We use the symbol $\mathbb{ECG}$ to denote the class of all \ecg s.

We have two notations for finding ancestors and descendants of a state
in an \ecg{}:
\begin{align*}
  \eancestor{\dsstate} &= \setbuild{ \dsstate' }{ (\dsstate',\dsstate) \in H } \cup \set{\dsstate}
  && \text{[ancestors]}
  \\
  \edescendent{\dsstate} &= \setbuild{ \dsstate' }{ (\dsstate,\dsstate') \in H } \cup \set{\dsstate}
  && \text{[descendants]}
  \text.
\end{align*}

\subsection{Integrating a work-set}
Since we only want to consider new states and edges in each iteration,
we need a work-set, or in this case, three work-sets:
\begin{itemize}
\item{$\Delta S$ contains states to add,}
\item{$\Delta E$ contains edges to add,}
\item{$\Delta H$ contains new $\epsilon$-edges.}
\end{itemize}
Let $\mathbb{WS} ::= (\Delta S, \Delta E, \Delta H)$ be the space of work-sets.

\subsection{A new fixed-point iteration-space}
Instead of consuming a graph and producing a graph, the new fixed-point iteration function will consume and produce a graph,
an \ecg{}, and the work-sets.
Hence, the iteration space of the new algorithm is:
\begin{equation*}
  \s{ICRPDS} = (\wp(\ControlStates) \times \wp(\ControlStates \times \StackAlpha_\pm \times \ControlStates)) \times \mathbb{ECG}
  \times \mathbb{WS}
  \text.
\end{equation*}
The \emph{I} in $\s{ICRPDS}$ stands for \emph{intermediate}.

\subsection{The $\boldsymbol\epsilon$-closure 
  graph work-list algorithm}
The function $\mkCRPDS' : \mathbb{RPDS} \to (\s{ICRPDS} \to
\s{ICRPDS})$ generates the required iteration function (Figure~\ref{fig:mkcompact-ecg}).
\begin{figure}
\figrule
\begin{align*}
  \mkCRPDS'(M)
  &= 
  f\text{, where }
  \\
  M &= (\QStates,\StackAlpha,\transfunction,\qstate_0)
  \\
  f(G,H,(\Delta S, \Delta E,\Delta H)) &=
   (G',H',(\Delta S' - S',\Delta E' - E',\Delta H' -  H))\text{, where }
  \\
  (\DSStates,\DSFrames,\DSEdges,\dsstate_0)  &= G 
  \\
  (\Delta E_0, \Delta H_0) &= \!\! \Union_{\dsstate \in \Delta S} \!\!  \fsprout_M(\dsstate)
  \\
  (\Delta E_1, \Delta H_1) &= 
  \!\!\!\!\!\!\!\!\!\!
  \Union_{(\dsstate,\stackchar_+,\dsstate') \in \Delta E} 
  \!\!\!\!\!\!\!\!\!\!
  \faddpush_M(G,H) (\dsstate,\stackchar_+,\dsstate') 
  \\
  (\Delta E_2, \Delta H_2) &=
  \!\!\!\!\!\!\!\!\!\!
  \Union_{(\dsstate,\stackchar_-,\dsstate') \in \Delta E} 
  \!\!\!\!\!\!\!\!\!\!
  \faddpop_M(G,H) (\dsstate,\stackchar_-,\dsstate')  
  \\
  (\Delta E_3, \Delta H_3) &=
  \!\!\!\!\!\!\!\!
  \Union_{(\dsstate,\epsilon,\dsstate') \in \Delta E} 
  \!\!\!\!\!\!\!\!
  \faddempty_M(G,H) (\dsstate,\dsstate')   
  \\
  (\Delta E_4, \Delta H_4) &=
  \!\!\!\!\!\!
  \Union_{(\dsstate,\dsstate') \in \Delta H} 
  \!\!\!\!\!\!\!
  \faddempty_M(G,H) (\dsstate,\dsstate')   
  \\
  S' &= \DSStates \union \Delta S
  \\
  E' &= \DSEdges \union \Delta E
  \\
  H' &= H \union \Delta H
  \\
  \Delta E' &= \Delta E_0 \union \Delta E_1 \union \Delta E_2 \union \Delta E_3 \union \Delta E_4
  \\
  \Delta S' &= \setbuild{ \dsstate' }{ (\dsstate,\stackact,\dsstate') \in \Delta E' } \cup \set{\dsstate_0}
  \\
  \Delta H' &= \Delta H_0 \union \Delta H_1 \union \Delta H_2 \union \Delta H_3 \union \Delta H_4
  \\
  \Delta \DSFrames &= \setbuild{\stackchar}{(\dsstate,\stackchar_+,\dsstate') \in \Delta E'}
  \\
  G' &= (\DSStates \union \Delta S,\DSFrames \union \Delta \DSFrames, E',\qstate_0)
  \text.
\end{align*}
\caption{The fixed point of the function $\mkCRPDS'(M)$ contains
  the compact form of the rooted pushdown system $M$.}
\label{fig:mkcompact-ecg}
\figrule
\end{figure}
Please note that we implicitly distribute union across tuples:
\begin{equation*}
  (X,Y) \union (X',Y') = 
  (X \union X, Y \union Y')
  \text.
\end{equation*}
The functions $\fsprout$, $\faddpush$, $\faddpop$, $\faddempty$ (defined shortly)
calculate the additional the graph edges and \ecg{} edges
(potentially) introduced by a new state or edge.

\paragraph{Sprouting}
Whenever a new state gets added to the graph, 
the algorithm must check whether that state has any new edges to contribute.
Both push edges and $\epsilon$-edges do not depend on the current
stack, so any such edges for a state in the pushdown system's
transition function belong in the graph. 
The sprout function:
\begin{equation*}
  \fsprout_{(\QStates,\StackAlpha,\transfunction,\dsstate_0)} :
\QStates \to (\Pow{\transfunction} \times \Pow{\QStates \times \QStates})
\text,
\end{equation*}
checks whether a new state could produce any new push edges or no-change edges.
We can represent its behavior diagrammatically (as previously, the
dotted arrows correspond to the corresponding additions to the
work-graph and $\epsilon$-closure work graph):
\begin{equation*}
  \xymatrix{
    &
    *+[F-:<1pt>]{\dsstate} \ar @{..>} [dl]_{\epsilon}^\transfunction \ar @{..>} [dr]^{\stackchar_+}_\transfunction  & 
    \\
     \qstate' && \qstate''
   }
\end{equation*}
which means if adding control state $\dsstate$:
\begin{itemize}
\item[] add edge $\triedge{\dsstate}{\epsilon}{\qstate'}$ if it exists
  in $\delta$ (hence the arrow subscript $\transfunction$), and
\item[] add edge $\triedge{\dsstate}{\stackchar_+}{\qstate''}$ if it exists in  $\delta$.
\end{itemize}
Formally:
\begin{align*}
  \fsprout_{(\QStates,\StackAlpha,\transfunction,\dsstate_0)} (\dsstate) &= (\Delta E, \Delta H)\text{, where }
\end{align*}
\begin{align*}
  \Delta E &= \setbuild{ \triedge{\dsstate}{\epsilon}{\qstate} }{ \triedge{\dsstate}{\epsilon}{\qstate} \in \transfunction }
  \mathrel{\union}
   \setbuild{ \triedge{\dsstate}{\stackchar_+}{\qstate} }{ \triedge{\dsstate}{\stackchar_+}{\qstate} \in \transfunction }
  \\
  \Delta H &= \setbuild{ \biedge{\dsstate}{\qstate} }{ \triedge{\dsstate}{\epsilon}{\qstate} \in \transfunction } 
  \text.
\end{align*}

\paragraph{Considering the consequences of a new push edge}
Once our algorithm adds a new push edge to a graph, there
is a chance that it will enable new pop edges for the same stack frame
somewhere downstream.
If and when it does enable pops, it will also add new edges to the
\ecg{}.
The $\faddpush$ function:
\begin{equation*}
  \faddpush_{(\QStates,\StackAlpha,\transfunction,\dsstate_0)} :
\mathbb{RPDS} \times \mathbb{ECG} \to \transfunction \to (\Pow{\transfunction} \times \Pow{\QStates \times \QStates})
\text,
\end{equation*}
checks for $\epsilon$-reachable states that could produce a pop.
We can represent this action by the following diagram (the arrow
subscript $\epsilon$ indicates edges in the $\epsilon$-closure graph):
\begin{equation*}
  \xymatrix{
    *+[F-:<1pt>]{\dsstate} \ar [r]^{\stackchar_+} \ar @(d,d) @{..>} [rrr]^\epsilon_\epsilon
    & 
    *+[F-:<1pt>]{\qstate} \ar [r]^{\epsilon}_{\epsilon}
    & 
    {\qstate'} \ar @{..>} [r]^{\stackchar_-}_\transfunction
    & \qstate''
    \\
    \\
 }
\end{equation*}
which means if adding push-edge $\triedge{\dsstate}{\stackchar_+}{\qstate}$:
\begin{itemize}
  \item[] if pop-edge $\triedge{\qstate'}{\stackchar_-}{\qstate''}$ is in $\delta$, then
  \item[] \hspace{1.5em} add edge $\triedge{\qstate'}{\stackchar_-}{\qstate''}$, and
  \item[] \hspace{1.5em} add $\epsilon$-edge $\biedge{\dsstate}{\qstate''}$.
\end{itemize}
Formally:
\begin{align*}
  &\faddpush_{(\QStates,\StackAlpha,\transfunction,\dsstate_0)} (G,H) (\triedge{\dsstate}{\stackchar_+}{\qstate}) =
  (\Delta E, \Delta H)\text{, where }
  \\
  &\;\;\;\Delta E = \setbuild{ \triedge{\qstate'}{\stackchar_-}{\qstate''} }{
    \qstate' \in \edescendent{\qstate}
    \text{ and }
    \triedge{\qstate'}{\stackchar_-}{\qstate''} \in \transfunction }
  \\
  &\;\;\; \Delta H = \setbuild{ \biedge{\dsstate}{\qstate''} }{ 
    \qstate' \in \edescendent{\qstate}
    \text{ and }
    \triedge{\qstate'}{\stackchar_-}{\qstate''} \in \transfunction }
  \text.
\end{align*}

\paragraph{Considering the consequences of a new pop edge}
%
%
Once the algorithm adds a new pop edge to a graph, it will create 
at least one new \ecg{} edge and possibly more by matching up with
upstream pushes.
The $\faddpop$ function:
\begin{equation*}
  \faddpop_{(\QStates,\StackAlpha,\transfunction,\dsstate_0)} :
\mathbb{RPDS} \times \mathbb{ECG} \to \transfunction \to (\Pow{\transfunction} \times \Pow{\QStates \times \QStates})
\text,
\end{equation*}
checks for $\epsilon$-reachable push-edges that could match this pop-edge.
This action is illustrated by the following diagram:
\begin{equation*}
  \xymatrix{
    {\dsstate} \ar [r]^{\stackchar_+} \ar @(d,d) @{..>} [rrr]^\epsilon_\epsilon
    & 
    {\dsstate'} \ar [r]^{\epsilon}_{\epsilon}
    & 
    *+[F-:<1pt>]{\dsstate''} \ar [r]^{\stackchar_-}_\transfunction
    & 
    *+[F-:<1pt>]{\qstate}
    \\
    \\
 }
\end{equation*}
which means if adding pop-edge $\triedge{\dsstate''}{\stackchar_-}{\qstate}$:
\begin{itemize}
 \item[] if push-edge $\triedge{\dsstate}{\stackchar_+}{\dsstate'}$ is already in the graph, then
 \item[] \hspace{1.5em} add $\epsilon$-edge $\biedge{\dsstate}{\qstate}$.
\end{itemize}
Formally:
\begin{align*}
  &\faddpop_{(\QStates,\StackAlpha,\transfunction,\dsstate_0)} (G,H) (\triedge{\dsstate''}{\stackchar_-}{\qstate}) =
  (\Delta E, \Delta H)\text{, where }
\\
&  \Delta E = \emptyset\mbox{ and }
  \Delta H = \setbuild{ \biedge{\dsstate}{\qstate} }{ 
    \dsstate' \in \eancestor{\dsstate''}
    \text{ and }
    \triedge{\dsstate}{\stackchar_+}{\dsstate'} \in G }
  \text.
\end{align*}
 
\paragraph{Considering the consequences of a new $\boldsymbol\epsilon$-edge}
%
%
Once the algorithm adds a new \ecg{} edge, it may transitively
have to add more \ecg{} edges, and it may connect an old push to
(perhaps newly enabled) pop edges.
The $\faddempty$ function:
\[
\begin{array}{l}
  \faddempty_{(\QStates,\StackAlpha,\transfunction,\dsstate_0)} :
\mathbb{RPDS} \times \mathbb{ECG} \to (\QState \times \QState) \to (\Pow{\transfunction} \times \Pow{\QStates \times \QStates})
\text,
\end{array}
\]
checks for newly enabled pops and \ecg{} edges:
Once again, we can represent this action diagrammatically:
\begin{equation*}
  \xymatrix{
    \\
   {\dsstate} \ar [r]^{\stackchar_+} 
   \ar @{..>} `d[ddr] `r[rrrrr]^\epsilon_\epsilon [rrrrr]
   &
   {\dsstate'} \ar [r]^\epsilon_\epsilon \ar @(u,u) @{..>} [rr]^\epsilon_\epsilon
   \ar @(d,d) @{..>} [rrr]^\epsilon_\epsilon
   & 
   *+[F-:<1pt>]{\dsstate''} \ar [r]^{\epsilon} \ar @(u,u) @{..>} [rr]^\epsilon_\epsilon
   &
   *+[F-:<1pt>]{\dsstate'''} \ar [r]^\epsilon_\epsilon
   & 
   {\dsstate''''} \ar @{..>} [r]^{\stackchar_-}_\transfunction
   & 
   {\qstate}
   \\
   & & & & \\
   & & & & \\
}
\end{equation*}
which means if adding $\epsilon$-edge $\biedge{\dsstate''}{\dsstate'''}$:
\begin{itemize}
  \item[] if pop-edge $\triedge{\dsstate''''}{\stackchar_-}{\qstate}$ is in $\transition$, then
  \item[] \hspace{2em} add $\epsilon$-edge $\biedge{\dsstate}{\qstate}$; and
  \item[] \hspace{2em} add edge $\triedge{\dsstate''''}{\stackchar_-}{\qstate}$; 
  \item[] add $\epsilon$-edges $\biedge{\dsstate'}{\dsstate'''}$,
    $\biedge{\dsstate''}{\dsstate''''}$, and
    $\biedge{\dsstate'}{\dsstate''''}$.
\end{itemize}
Formally:
\begin{align*}
  &\faddempty_{(\QStates,\StackAlpha,\transfunction,\dsstate_0)} (G,H) (\biedge{\dsstate''}{\dsstate'''})
  = (\Delta E, \Delta H)\text{, where }
  \\
  &\;\;\;\;\Delta E = \big\{ \triedge{\dsstate''''}{\stackchar_-}{\qstate} :
    \dsstate' \in \eancestor{\dsstate''}
    \text{ and }
    \dsstate'''' \in \edescendent{\dsstate'''}
    \text{ and } \\
    & \hspace{9.5em} \triedge{\dsstate}{\stackchar_+}{\dsstate'} \in G
    \text{ and }
    \triedge{\dsstate''''}{\stackchar_-}{\qstate} \in \transfunction \big\}
  \\
  &\;\;\;\;\Delta H = \big\{ \biedge{\dsstate}{\qstate} :
    \dsstate' \in \eancestor{\dsstate''}
    \text{ and }
    \dsstate'''' \in \edescendent{\dsstate'''}
    \text{ and } \\
    & \hspace{8.5em} \triedge{\dsstate}{\stackchar_+}{\dsstate'} \in G
    \text{ and }
    \triedge{\dsstate''''}{\stackchar_-}{\qstate} \in \transfunction \big\}
  \\
  & \hspace{3em} \union \setbuild{ \biedge{\dsstate'}{\dsstate'''}  }{
    \dsstate' \in \eancestor{\dsstate''}
  }
  \\
  & \hspace{3em} \union \setbuild{ \biedge{\dsstate''}{\dsstate''''}  }{
    \dsstate'''' \in \edescendent{\dsstate'''}
   }
  \\
  & \hspace{3em} \union \setbuild{ \biedge{\dsstate'}{\dsstate''''}  }{
    \dsstate' \in \eancestor{\dsstate''}
    \text{ and }
    \dsstate'''' \in \edescendent{\dsstate'''}
  }
  \text.
\end{align*}

\subsection{Termination and correctness}
To prove that a fixed point exists, we show the iteration function is monotonic.
The key observation is that $\Delta G$ and $\Delta H$ drive all additions to, and are disjoint from, $G$ and $H$.
Since $G$ and $H$ monotonically increase in a finite space, $\Delta G$ and $\Delta H$ run out of room (full details in \autoref{lem:termination}).
Once the graph reaches a fixed point, all work-sets will be empty, and the \ecg{} will also be saturated.
We can also show that this algorithm is correct by defining first 
$\fECG : \mathbb{RPDS} \to \mathbb{ECG}$ as
\begin{equation*}
  \fECG(M) = \setbuild{\biedge{\dsstate}{\dsstate'}}{\dsstate \RPDTransOU{\vec{\stackact}}{M} \dsstate' \text{ and } [\vec{\stackact}] = \epsilon}
\end{equation*}

and stating the following theorem:

\begin{theorem}\label{thm:eps-closure-correct}
  For all $M \in \mathbb{RPDS}$, $\fCRPDS(M) = G$ and $\fECG(M) = H$,
  where $(G,H,(\emptyset,\emptyset,\emptyset)) = \lfp(\mkCRPDS'(M))$.
\end{theorem}
In the proof of Theorem~\ref{thm:eps-closure-correct}, the $\subseteq$
case comes from an invariant lemma we have on $\mkCRPDS'$:

\begin{lemma}
\begin{align*}
  \invcrpdsp((S, E), H, (\Delta S, \Delta E, \Delta H)) &=
    \forall \triedge{\dsstate}{\stackact}{\dsstate'} \in E \cup \Delta E. \dsstate \RPDTransOU{\stackact}{M} \dsstate' \\
   &\wedge \forall \biedge{\dsstate}{\dsstate'} \in H \cup \Delta H.
             \exists \vec{\stackact}. [\vec{\stackact}] = \epsilon
                     \wedge \forall \acont. (\dsstate, \acont) \overset{\mkern-12mu\vec{\stackact}}{\underset{\mkern-12mu{}M}{\longmapsto^*}} (\dsstate, \acont)
\end{align*}  
\end{lemma}
The $\supseteq$ case follows from
\begin{lemma}
  For all traces $\mtrace \equiv \dsstate_0 \RPDTransOU{\vec{\stackact}}{M} \dsstate$,
  there is both a corresponding path $\dsstate_0 \RPDTransOU{\vec{\stackact}}{G} \dsstate$ and
  for all non-empty subtraces of $\mtrace$, $\dsstate_b \RPDTransOU{\vec{\stackact'}}{M} \dsstate_a$, if $[\vec{\stackact'}] = \epsilon$ then $\biedge{\dsstate_b}{\dsstate_a} \in H$.
\end{lemma}

Since all edges in a compact rooted pushdown system must be in a path from the initial state, we can extract the edges from said paths using this lemma.

\subsection{Complexity: Still exponential, but more efficient}

As in the previous case (Section~\ref{sec:compl-pol-exp}), to
determine the complexity of this algorithm, we ask two questions: how
many times would the algorithm invoke the iteration function in the
worst case, and how much does each invocation cost in the worst case?
The run-time of the algorithm is bounded by the size of the final
graph plus the size of the \ecg.
Suppose the final graph has $m$ states.
In the worst case, the iteration function adds only a single edge each
time.
There are at most $2\abs{\StackAlpha}m^2 + m^2$ edges in the graph ($\abs{\StackAlpha}m^2$ push edges, just as many pop edges, and $m^2$ no-change edges) and
at most $m^2$ edges in the \ecg{}, which bounds the number of
iterations.
Recall that $m$ can be exponential in the size of the program, since $m \leq \abs{Q}$ (and Section~\ref{sec:doubly-exponential} derived the exponential size of $\abs{Q}$).

Next, we must reason about the worst-case cost of adding an edge: how
many edges might an individual iteration consider?
In the worst case, the algorithm will consider every edge in every
iteration, leading to an asymptotic time-complexity of:
\begin{equation*}
  O((2\abs{\StackAlpha}m^2 + 2m^2)^2) = 
  O(\abs{\StackAlpha}^2m^4)
  \text.
\end{equation*}
While still high, this is a an improvement upon the previous
algorithm.  
%
%
For sparse graphs, this is a reasonable algorithm.

\section{Polynomial-Time Complexity from Widening}
\label{sec:widening}

In the previous section, we developed a more efficient fixed-point
algorithm for computing a compact rooted pushdown system.
Even with the core improvements we made, the algorithm remained
exponential in the worst case, owing to the fact that there could be
an exponential number of reachable control states.
When an abstract interpretation is intolerably complex, the standard
approach for reducing complexity and accelerating convergence is
widening~\cite{mattmight:Cousot:1977:AI}.
Of course, widening techniques trade away some precision to gain this
speed.
It turns out that the small-step variants of finite-state CFAs are
exponential without some sort of widening as
well~\cite{dvanhorn:VanHorn-Mairson:ICFP08}.

To achieve polynomial time complexity for pushdown control-flow
analysis requires the same two steps as the classical case: (1)
widening the abstract interpretation to use a global,
``single-threaded'' store and (2) selecting a monovariant allocation
function to collapse the abstract configuration-space.
Widening eliminates a source of exponentiality in the size of the
store; monovariance eliminates a source of exponentiality from
environments.
In this section, we redevelop the pushdown control-flow analysis
framework with a single-threaded store and calculate
its complexity.

\subsection{Step 1: Refactor the concrete semantics}
First, consider defining the reachable states of the concrete
semantics using fixed points.
That is, let the system-space of the evaluation function be
sets of configurations:
\begin{equation*}
  \system \in \s{System} = 
  \Pow{\s{Conf}} =
  \Pow{\syn{Exp} \times \s{Env} \times \s{Store} \times \s{Kont}}
  \text.
\end{equation*}
We can redefine the concrete evaluation function:
\begin{align*}
  \Eval(\expr) &= \lfp(\tf_\expr)\text{, where }
  \tf_\expr: \s{System} \to \s{System}
  \text{ and }
  \\
  \tf_\expr(\system) &= \set{\Inject(\expr)} \union \setbuild{ \conf' }{ \conf \in \system \text{ and } \conf \To \conf' }
  \text.
\end{align*}

\subsection{Step 2: Refactor the abstract semantics}
We can take the same approach with the abstract evaluation function,
first redefining the abstract system-space:
\begin{align*}
  \asystem  \in \sa{System} &= 
  \Pow{\sa{Conf}} 
\\
&=
  \Pow{\syn{Exp} \times \sa{Env} \times \sa{Store} \times \sa{Kont}}
  \text,
\end{align*}
and then the abstract evaluation function:
\begin{align*}
  \aEval(\expr) &= \lfp(\atf_\expr)\text{, where } \atf_\expr : \sa{System} \to \sa{System} \text{ and }
  \\
  \atf_\expr(\asystem) &= \set{\aInject(\expr)} \union \setbuild{ \aconf' }
  { \aconf \in \asystem \text{ and } \aconf \aTo \aconf' }
  \text.
\end{align*}
What we'd like to do is shrink the abstract system-space with a
refactoring that corresponds to a widening.

\subsection{Step 3: Single-thread the abstract store}
We can approximate a set of abstract stores
$\set{\astore_1,\ldots,\astore_n}$ with 
the least-upper-bound of those stores: $\astore_1 \join \cdots
\join \astore_n$.
We can exploit this by creating a new abstract system space in which
the store is factored out of every configuration.
Thus, the system-space contains a set of \emph{partial configurations}
and a single global store:
\begin{align*}
  \sa{System}' &= \Pow{\sa{PConf}} \times \sa{Store}
  \\
  \apconf \in \sa{PConf} &= \syn{Exp} \times \sa{Env} \times \sa{Kont}
  \text.
\end{align*}
We can factor the store out of the abstract transition relation as well, so that
$(\afTo^\astore) \subseteq \sa{PConf} \times (\sa{PConf} \times \sa{Store})$:
\begin{align*}
  (\expr, \aenv, \acont) \mathrel{\afTo^{\astore}} ((\expr',\aenv',\acont'), \astore') 
  \text{ iff }
(\expr, \aenv, \astore, \acont) \aTo (\expr',\aenv', \astore', \acont') 
\text,
\end{align*}
which gives us a new iteration function,
$\atf'_\expr : \sa{System}' \to \sa{System}'$,
\begin{align*}
  \atf'_\expr(\hat P,\astore) &= (\hat P',\astore')\text{, where } 
  \\
  \hat P' &= \setbuild{ \apconf' }{ \apconf \mathrel{\afTo^\astore} (\apconf',\astore'') } \union \set{ \apconf_0 } 
  \\
  \astore' &= \bigjoin \setbuild{ \astore'' }{ \apconf \mathrel{\afTo^\astore} (\apconf',\astore'') }
  \\
  (\apconf_0,\vect{}) &= \aInject(\expr)
  \text.
\end{align*}

\subsection{Step 4: Pushdown control-flow graphs}
Following the earlier graph reformulation of the compact rooted
pushdown system, we can reformulate the set of partial
configurations as a \emph{pushdown control-flow graph}.
A \defterm{pushdown control-flow graph} is a frame-action-labeled
graph over partial control states, and a \defterm{partial control
  state} is an expression paired with an environment:
\begin{align*}
  \sa{System}'' &= \sa{PDCFG} \times \sa{Store}
  \\
  \sa{PDCFG} &= \PowSm{\sa{PState}} \times \PowSm{\sa{PState} \times \sa{Frame}_\pm \times \sa{PState}}
  \\
  \apstate \in \sa{PState} &= \syn{Exp} \times \sa{Env}
  \text.
\end{align*}
In a pushdown control-flow graph, the partial control states are
partial configurations which have dropped the continuation component;
the continuations are encoded as paths through the graph.

\paragraph{A preliminary analysis of complexity}
Even without defining the system-space iteration function, we can ask,
\emph{How many iterations will it take to reach a fixed point in the worst
case?}
This question is really asking, \emph{How many edges can we add?}
And, \emph{How many entries are there in the store?}
Summing these together, we arrive at the worst-case number of
iterations:
\begin{align*}
  \overbrace{\abs{\sa{PState}}
  \times
  \abs{\sa{Frame}_\pm}
  \times
  \abs{\sa{PState}}}^{\text{PDCFG edges}}
  +
  \overbrace{
    \abs{\sa{Addr}}
    \times
    \abs{\sa{Clo}}
  }^{\text{store entries}}
  \text.
\end{align*}
With a monovariant allocation scheme that eliminates abstract environments, the number of iterations
ultimately reduces to:
\begin{align*}
  \abs{\syn{Exp}}
  \times (2 \abs{\sa{\syn{Var}}} + 1)
  \times \abs{\syn{Exp}}
  +
  \abs{\syn{Var}}
  \times \abs{\syn{Lam}}
  \text,
\end{align*}
which means that, in the worst case, the algorithm makes a cubic
number of iterations with respect to the size of the input
program.\footnote{In computing the number of frames, we note that in
  every continuation, the variable and the expression uniquely
  determine each other based on the let-expression from which they
  both came.
  As a result, the number of abstract frames available in a
  monovariant analysis is bounded by both the number of variables and
  the number of expressions, \ie, $\abs{\sa{Frame}} =
  \abs{\syn{Var}}$.}

The worst-case cost of the each iteration would be dominated by a
CFL-reachability calculation, which, in the worst case, must consider
every state and every edge:
\begin{equation*}
  O(\abs{\syn{Var}}^3 \times \abs{\syn{Exp}}^3)
  \text.
\end{equation*}
Thus, each iteration takes $O(n^{6})$ and there are a maximum of $O(n^{3})$ iterations, where $n$ is the size of the program.
So, total complexity would be $O(n^{9})$ for a monovariant
pushdown control-flow analysis with this scheme, where $n$ is again the size of the program.
Although this algorithm is polynomial-time, we can do better.

\subsection{Step 5: Reintroduce 
  $\boldsymbol\epsilon$-closure graphs}\label{sec:pdcfa-eps}
Replicating the evolution from Section~\ref{sec:ecg-worklist} for this
store-widened analysis, we arrive at a more efficient polynomial-time
analysis.
An \ecg{} in this setting is a set of pairs of store-less,
continuation-less partial states:
\begin{align*}
  \sa{ECG} &= \Pow{\sa{PState} \times \sa{PState}}
  \text.
\end{align*}
Then, we can set the system space to include \ecg s:
\begin{align*}
  \sa{System}''' &= \sa{CRPDS} \times \sa{ECG} \times \sa{Store}\text.
\end{align*}

Before we redefine the iteration function, we need another factored
transition relation.
The stack- and action-factored transition relation
$(\overset{\astore}{\underset{\stackact}{\apTo}}) \subseteq \sa{PState} \times
\sa{PState} \times \s{Store}$ determines if a transition is possible
under the specified store and stack-action:
\begin{align*}
  \pdcfato{(\expr,\aenv)}{\astore}{\aphrame_+}{((\expr',\aenv'),\astore')}
  & \text{ iff }
  (\expr,\aenv,\astore,\acont) 
  \aTo
  (\expr',\aenv',\astore',\aphrame : \acont)
  \\
  \pdcfato{(\expr,\aenv)}{\astore}{\aphrame_-}{((\expr',\aenv'),\astore')}
  & \text{ iff }
  (\expr,\aenv,\astore,\aphrame : \acont) 
  \aTo
  (\expr',\aenv',\astore', \acont)
  \\
  \pdcfato{(\expr,\aenv)}{\astore}{\epsilon}{((\expr',\aenv'),\astore')}
  & \text{ iff }
  (\expr,\aenv,\astore,\acont) 
  \aTo
  (\expr',\aenv',\astore', \acont)
  \text.
\end{align*}

Now, we can redefine the iteration function (Figure~\ref{fig:widen-trans}).

\begin{figure}
\figrule
\begin{align*}
  \atf_\expr((\hat P,\hat E), \hat H, \astore) &= ((\hat P',\hat E'), \hat H', \astore'')
  \text{, where }
\end{align*}
\begin{align*}
  \hat T_+ &= \setbuild{ (\triedge{\apstate}{\aphrame_+}{\apstate'},\astore') }{
    \pdcfato{\apstate}{\astore}{\aphrame_+}{(\apstate',\astore')}
  }
  \\
  \hat T_\epsilon &= \setbuild{ (\triedge{\apstate}{\epsilon}{\apstate'},\astore') }{
    \pdcfato{\apstate}{\astore}{\epsilon}{(\apstate',\astore')} }
  \\
 \hat T_- &= 
  \big\{
    (\triedge{\apstate''}{\aphrame_-}{\apstate'''},\astore') 
  :
    \pdcfato{\apstate''}{\astore}{\aphrame_-}{(\apstate''',\astore')} \text{ and }  
    \\
    &\hspace{9.5em}
    \triedge{\apstate}{\aphrame_+}{\apstate'} \in \hat E \text{ and }
    \\
    &\hspace{9.5em}
    \biedge{\apstate'}{\apstate''} \in \hat H
    \big\}
  \\
  \hat T' &= {\hat T_+} \union {\hat T_\epsilon} \union {\hat T_-}
  \\
  \hat E' &= \setbuild{ \hat e }{ (\hat e,\_) \in \hat T '}
  \\
  \astore'' &= \bigjoin \setbuild{ \astore' }{ (\_, \astore') \in \hat T'}
  \\
  \hat H_{\epsilon} &= \setbuild{ \biedge{\apstate}{\apstate''} }{
    \biedge{\apstate}{\apstate'} \in \hat H \text{ and } 
      \biedge{\apstate'}{\apstate''} \in \hat H
  }
  \\
  \hat H_{{+}{-}} &= \big\{
  \biedge{\apstate}{\apstate'''}
    : 
    \triedge{\apstate}{\aphrame_+}{\apstate'} \in \hat E
    \text{ and }
    \biedge{\apstate'}{\apstate''} \in \hat H 
    \\
    & \hspace{6.1em} \text{ and } 
    \triedge{\apstate''}{\aphrame_-}{\apstate'''} \in \hat E
    \big\}
  \\
  \hat H' &=   \hat H_{\epsilon} \union   \hat H_{{+}{-}} 
  \\
  \hat P' &= \hat P \union \setbuild{ \apstate' }{ \triedge{\apstate}{\stackact}{\apstate'} } \cup \set{(\expr,\bot)}  
  \text.
\end{align*}
\caption{An \ecg{}-powered iteration function for pushdown control-flow analysis with a single-threaded store.}
\label{fig:widen-trans}
\figrule
\end{figure}

\begin{theorem}
Pushdown 0CFA with single-threaded store (PDCFA) can be computed in $O(n^6)$-time, where $n$ is the
size of the program.
\end{theorem}
\begin{proof}
%
As before, the maximum number of iterations is cubic in the size of
the program for a monovariant analysis.
Fortunately, the cost of each iteration is also now bounded by the number
of edges in the graph, which is also cubic.
\end{proof}

\section{Introspection for Abstract Garbage Collection}
Abstract garbage collection~\cite{mattmight:Might:2006:GammaCFA} yields large
improvements in precision by using the abstract interpretation of garbage
collection to make more efficient use of the finite address space available
during analysis.
Because of the way abstract garbage collection operates, it grants exact
precision to the flow analysis of variables whose bindings die
between invocations of the same abstract context.
Because pushdown analysis grants exact precision in tracking return-flow, it is
clearly advantageous to combine these techniques.
Unfortunately, as we shall demonstrate, abstract garbage collection
breaks the pushdown model by requiring a full traversal of the stack to discover the
root set.

Abstract garbage collection modifies the transition relation
to conduct a ``stop-and-copy'' garbage collection before each
transition.
To do this, we define a garbage collection function 
$\aCollect : \sa{Conf} \to \sa{Conf}$
on
configurations:
\begin{align*}
\aCollect(\overbrace{\expr,\aenv,\astore,\acont}^{\aconf})
&= (\expr,\aenv,\astore|\mathit{Reachable}(\aconf),\acont)
  \text,
  \end{align*}
  where the pipe operation $f|S$ yields the function $f$, but with
  inputs not in the set $S$ mapped to bottom---the empty set.
  The reachability function $\mathit{Reachable} : \sa{Conf} \to \PowSm{\sa{Addr}}$
  first computes the root set, and then the transitive closure of an
  address-to-address adjacency relation: 
  \begin{align*}
  \mathit{Reachable}(\overbrace{\expr,\aenv,\astore,\acont}^\aconf) &=
\setbuild{ \aaddr }{ \aaddr_0 \in \mathit{Root}(\aconf)
  \text{ and }
  \aaddr_0  
    \mathrel{\areaches_\astore^*}
  \aaddr
}
\text,
  \end{align*}
  where the function $\mathit{Root} : \sa{Conf} \to 
  \PowSm{\sa{Addr}}$ 
  finds the root addresses:
  \begin{align*} 
  \mathit{Root}(\expr,\aenv,\astore,\acont) &=
  \mathit{range}(\aenv) \union
\StackRoot(\acont)
  \text,
  \end{align*}
  and the $\StackRoot : \sa{Kont} \to \PowSm{\sa{Addr}}$ function
  finds roots down the stack:
  \begin{align*} 
  \StackRoot
  \vect{\phrame_1,\ldots,\phrame_n}
  &= 
\Union_i \touches(\phrame_i)
  \text,
  \end{align*}
  using a ``touches'' function, $\touches : \sa{Frame} \to \PowSm{\sa{Addr}}$:
  \begin{align*}
    \touches(\vv,\expr,\aenv) &= \range(\aenv)
    \text,
  \end{align*}
  and the relation
  $(\areaches) \subseteq \sa{Addr} \times \sa{Store} \times \sa{Addr}$
  connects adjacent addresses:
  \begin{align*}
  \aaddr 
  \mathrel{\areaches_\astore} 
  \aaddr'
  \text{ iff there exists }
(\lam,\aenv) \in \astore(\aaddr)
  \text{ such that }
  \aaddr' \in \range(\aenv)
  \text.
  \end{align*}

The new abstract transition relation is thus the composition of abstract garbage collection with the old transition relation:
\begin{align*}
(\aTo_{\mathrm{GC}}) = 
(\aTo) \compose \aCollect\text.
\end{align*}

\paragraph{Problem: Stack traversal violates pushdown constraint}

In the formulation of pushdown systems, the transition relation is restricted
to looking at the top frame, and in less restricted formulations that may read the stack,
the reachability decision procedures need the entire system up-front.
Thus, the relation $(\aTo_{\mathrm{GC}})$ cannot be computed as
a straightforward pushdown analysis using summarization.

\paragraph{Solution: Introspective pushdown systems}
To accommodate the richer structure of the relation $(\aTo_{\mathrm{GC}})$, we
now define \emph{introspective} pushdown systems.
Once defined, we can embed the garbage-collecting abstract interpretation
within this framework, and then focus on developing a control-state
reachability algorithm for these systems.

An \defterm{introspective pushdown system} is a quadruple
$M = (\ControlStates,\StackAlpha,\transfunction,\qstate_0)$:
\begin{enumerate}

\item $\ControlStates$ is a finite set of control states;

\item $\StackAlpha$ is a stack alphabet; 

\item $\transfunction \subseteq \ControlStates \times \StackAlpha^*
\times \StackAlpha_\pm \times \ControlStates$ is a transition relation (where $(\qstate, \cont, \phrame_-, \qstate') \in \transfunction$ implies $\cont \equiv \phrame:\cont'$); and

\item $\qstate_0$ is a distinguished root control state. 
\end{enumerate}
The second component in the transition relation is 
a realizable stack at the given control-state.
This realizable stack distinguishes an introspective pushdown system
from a general pushdown system.
$\mathbb{IPDS}$  denotes the class of all introspective pushdown
systems.

Determining how (or if) a control state $\qstate$
transitions to a control state $\qstate'$, requires knowing a
path taken to the state $\qstate$.
We concern ourselves with root-reachable states.
%
When $M = (\ControlStates,\StackAlpha,\transfunction,\qstate_0)$,
if there is a $\acont$ such that $(\qstate_0,\vect{}) \mathrel{\underset{M}{\PDTrans}^*} (\qstate, \acont)$ we say $\qstate$ is reachable via $\acont$, where
\begin{align*}
\infer{ }{(\qstate, \acont) \mathrel{\underset{M}{\PDTrans}^*} (\qstate, \acont)}
\qquad
\infer{(\qstate, \acont) \mathrel{{\underset{M}{\PDTrans}}^*} (\qstate', \acont') \quad
 (\qstate', \acont', \stackact', \qstate'') \in \transfunction}
{(\qstate, \acont) \mathrel{\underset{M}{\PDTrans}^*} (\qstate'', [\acont'_+\stackact'])}
\end{align*}

\subsection{Garbage collection in introspective pushdown systems}

To convert the garbage-collecting,
abstracted CESK machine into an introspective pushdown system,
we use the function $\afIPDS : \syn{Exp} \to \mathbb{IPDS}$:
\begin{center}
  \minipagebreak{0.45}{%
    \afIPDS(\expr) &= (\QStates,\StackAlpha,\transfunction,\qstate_0)
    \\
    \QStates &= \syn{Exp} \times \sa{Env} \times \sa{Store}
    \\
    \StackAlpha &= \sa{Frame}
    \\
    (\qstate_0,\vect{}) &= \aInject(\expr)} {0.50}{%
    (\qstate,\acont,\epsilon,\qstate')
    \in \transfunction & \text{ iff } \aCollect(\qstate, \acont) \aTo
    (\qstate', \acont)
    \\
    (\qstate,\aphrame : \acont,\aphrame_{-},\qstate') \in
    \transfunction & \text{ iff } \aCollect(\qstate, \aphrame :
    \acont) \aTo (\qstate',\acont)
    \\
    (\qstate,\acont,\aphrame_{+},\qstate')
    \in \transfunction & \text{ iff } \aCollect(\qstate, \acont) \aTo
    (\qstate',\aphrame : \acont)
    \text.}
\end{center}

\section{Problem: Reachability for Introspective Pushdown Systems is Uncomputable}
\label{sec:ipds-incomputable}

As currently formulated, computing control-state reachability
for introspective pushdown systems is uncomputable.
The problem is that the transition relation expects to enumerate every possible
stack for every control point at every transition, without restriction.
\begin{theorem}
  Reachability in introspective pushdown systems is uncomputable.
\end{theorem}
\begin{proof}
  Consider an IPDS with two states --- {\tt searching} (start state) and
  {\tt valid} --- and a singleton stack alphabet of unit
  ($\top$).
  For any first-order logic proposition, $\phi$, we can define a
  reduction relation that interprets the length of the stack as an
  encoding of a proof of $\phi$.
  If the length encodes an ill-formed proof object, or is not a proof
  of $\phi$, {\tt searching} pushes $\top$ on the stack
  and transitions to itself.
  If the length encodes a proof of $\phi$, transition
  to {\tt valid}.
  By the completeness of first-order logic, if $\phi$ is valid, there
  is a finite proof, making the pushdown system terminate in {\tt valid}.
  If it is not valid, then there is no proof and {\tt valid} is unreachable.
  Due to the undecidability of first-order logic, we definitely cannot have a decision procedure for
  reachability of IPDSs.
\end{proof}

To make introspective pushdown systems computable,
we must first refine our definition of introspective pushdown
systems to operate on \emph{sets} of stacks
and insist these sets be regular.

\newpage

A \defterm{conditional pushdown system} (CPDS) is a quadruple
$M = (\ControlStates,\StackAlpha,\transfunction,\qstate_0)$:
\begin{enumerate}

\item $\ControlStates$ is a finite set of control states;

\item $\StackAlpha$ is a stack alphabet; 

\item $\transfunction \subseteq_\text{fin} \ControlStates \times \mathcal{REG}(\StackAlpha^*)
\times \StackAlpha_\pm \times \ControlStates$ is a transition relation (same restriction on stacks); and

\item $\qstate_0$ is a distinguished root control state,
\end{enumerate}
where $\mathcal{REG}(S)$ is the set of all regular languages formable with strings in $S$.

The regularity constraint on the transition relations guarantees that
we can decide applicability of transition rules at each state, since (as we will see)
the set of all stacks that reach a state in a CPDS has decidable overlap with regular languages.
Let $\mathbb{CPDS}$ denote the set of all conditional pushdown systems.

The rules for reachability with respect to sets of stacks are similar to those for IPDSs.
\begin{align*}
\infer{ }{(\qstate,\acont) \mathrel{\underset{M}{\PDTrans}^*} (\qstate,\acont)}
\qquad
\infer{(\qstate,\acont) \mathrel{\underset{M}{\PDTrans}^*} (\qstate',\acont')
       \quad \acont' \in \hat K'
       \quad (\qstate', \hat K', \stackact', \qstate'') \in \transfunction
       }
      {(\qstate,\acont) \mathrel{\underset{M}{\PDTrans}^*} (\qstate'',[\acont'_+\stackact'])}
\end{align*}

We will write $\qstate \mathrel{\overset{\hat K,\stackact}{\underset{M}{\RPDTrans}}} \qstate'$ to mean there are $\acont, \hat K$ such that $\qstate$ is reachable via $\acont$, $\acont \in \hat K$ and $(\qstate,\hat K,\stackact,\qstate') \in \transfunction$. We will omit the labels above if they merely exist.

\subsection{Garbage collection in conditional pushdown systems}

Of course, we must adapt abstract garbage collection to this refined framework.
To convert the garbage-collecting,
abstracted CESK machine into a conditional pushdown system,
we use the function $\afIPDS' : \syn{Exp} \to \mathbb{CPDS}$:
\begin{align*}
\afIPDS'(\expr) &= (\QStates,\StackAlpha,\transfunction,\qstate_0)
  \\
    \QStates &= \syn{Exp} \times \sa{Env} \times \sa{Store}
    \\
      \StackAlpha &= \sa{Frame}
      \\
\text{For all sets of addresses } A \subseteq \sa{Addr} \text{ let }& \hat K = \setbuild{\acont}{\StackRoot(\acont) = A}
\\
        (\qstate,\hat K,\epsilon,\qstate') 
        \in \transfunction
        & \text{ iff }
\aCollect(\qstate, \acont)
  \aTo
  (\qstate', \acont)
  \text{ for any }
  \acont \in \hat K
  \\
    (\qstate,\hat K,\aphrame_{-},\qstate')
    \in \transfunction
    & \text{ iff }
\aCollect(\qstate, \aphrame : \acont) 
  \aTo
  (\qstate',\acont)
  \text{ for any }
  \aphrame : \acont \in \hat K
  \\
    (\qstate,\hat K,\aphrame_{+},\qstate') 
    \in \transfunction
    & \text{ iff }
\aCollect(\qstate, \acont)
  \aTo
  (\qstate',\aphrame : \acont)
  \text{ for any }
   \acont \in \hat K
  \\
(\qstate_0,\vect{}) &= \aInject(\expr)
\text.
\end{align*}

Assuming we can overcome the difficulty of computing with some representation of a set of stacks, the intuition for the decidability of control-state reachability with garbage collection stems from two observations:
garbage collection operates on sets of addresses, and for any given control point there is a finite number of sets of sets of addresses.
The finiteness makes the definition of $\transfunction$ fit the finiteness restriction of CPDSs.
The regularity of $\hat K$ (for any given $A$, which we recall are finite sets) is apparent from a simple construction: let the DFA control states represent the subsets of $A$, with $\emptyset$ the start state and $A$ the accepting state.
Transition from $A' \subseteq A$ to $A \cup \touches(\aphrame)$ for each $\aphrame$ (no transition if the result is not a subset of $A$).
Thus any string of frames that has a stack root of $A$ (and only $A$) gets accepted.

The last challenge to consider before we 
can delve into the mechanics of computing reachable control states
is \emph{how} to represent the sets of stacks 
that may be paired with each control state.
Fortunately, a regular language can describe the stacks that share the same root addresses,
the set of stacks at a control point are recognized by a one-way non-deterministic stack automaton (1NSA), \emph{and}, fortuitously,
non-empty overlap of these two is decidable (but NP-hard~\citep{ianjohnson:rounds:complexity:1973}).
The 1NSA describing the set of stacks at a control point is already encoded in the structure of the (augmented) CRPDS that we will accumulate while computing reachable control states.
As we develop an algorithm for control-state reachability, we will
exploit this insight (Section~\ref{sec:implementation}).

\section{Reachability in Conditional Pushdown Systems}
We will show a progression of constructions that take us along the following line:
\begin{equation*}
  \mathbb{CPDS} \xrightarrow[{\mathsection\ref{sec:icrpds}}]{} \mathbb{CCPDS}
  \xrightarrow[{\mathsection\ref{sec:gc-pdcfa}}]{\text{specialize}} \mathit{PDCFA\ with\ GC} \to \mathit{approx.\ PDCFA\ with\ GC}
\end{equation*}

In the first construction, we show that a CCPDS is finitely constructible in a similar fashion as in Section~\ref{sec:rpds-to-crpds}.
The key is to take the current introspective CRPDS and ``read off'' an automaton that describes the stacks accepted at each state.
For traditional pushdown systems, this is always an NFA, but introspection adds another feature: transition if the string accepted so far is accepted by a given NFA.
Such power falls outside of standard NFAs and into one-way non-deterministic stack automata (1NSA)\footnote{The reachable states of a 1NSA is known to be regular, but the paths are not.}.
These automata enjoy closure under finite intersection with regular languages and decidable emptiness checking~\citep{ianjohnson:one-way-sa:ginsburg:1967}, which we use to decide applicability of transition rules.
If the stacks realizable at $\qstate$ have a non-empty intersection with a set of stacks $\hat K$ in a rule $(\qstate, \hat K, \stackact, \qstate') \in \transfunction$, then there are paths from the start state to $\qstate$ that further reach $\qstate'$.

The structure of the GC problem allows us to sidestep the 1NSA constructions and more directly compute state reachability.
We specialize to garbage collection in \autoref{sec:gc-pdcfa}.
We finally show a space-saving approximation that our implementation uses.

\subsection{One-way non-deterministic stack automata}

The machinery we use for describing the realizable stacks at a state is a generalized pushdown automaton itself.
A stack automaton is permitted to move a cursor up and down the stack and read frames (left and right on the input if two-way, only right if one-way), but only push and pop when the stack cursor is at the top.
Formally, a \defterm{one-way stack automaton} is a 6-tuple $A = (Q, \Sigma, \StackAlpha, \transfunction, \qstate_0, F)$ where
\begin{enumerate}
\item $Q$ is a finite nonempty set of states,
\item $\Sigma$ is a finite nonempty input alphabet,
\item $\StackAlpha$ is a finite nonempty stack alphabet,
\item $\transfunction \subseteq Q \times (\StackAlpha \cup \set{\epsilon}) \times (\Sigma \cup \set{\epsilon}) \times \set{\uparrow, \cdot, \downarrow} \times \StackAlpha_\pm \times Q$ is the transition relation,
\item $\qstate_0 \in Q$ is the start state, and
\item $F \subseteq Q$ the set of final states
\end{enumerate}

An element of the transition relation, $(\qstate, \phrame_\epsilon, a, d, \phrame_\pm, \qstate')$, should be read as, ``if at $\qstate$ the right of the stack cursor is prefixed by $\phrame_\epsilon$ and the input is prefixed by $a$, then consume $a$ of the input, transition to state $\qstate'$, move the stack cursor in direction $d$, and if at the top of the stack, perform stack action $\phrame_\pm$.''
This reading translates into a run relation on \defterm{instantaneous descriptions}, $Q \times (\StackAlpha^* \times \StackAlpha^*) \times \Sigma^*$.
These descriptions are essentially machine states that hold the current control state, the stack split around the cursor, and the rest of the input.
\begin{equation*}
  \infer{(\qstate, \phrame_\epsilon, a, d, \phrame_\pm, \qstate') \in \transfunction
         \quad \phrame_\epsilon \sqsubseteq \Gamma_T
         \quad w \equiv aw'
         \quad [\Gamma_B', \Gamma_T'] = P(\phrame_\pm, D(d, [\Gamma_B, \Gamma_T]))}
        {(\qstate, [\Gamma_B, \Gamma_T], w) \longmapsto (\qstate', [\Gamma_B', \Gamma_T'], w')}
\end{equation*}
where
\begin{center}
  \begin{minipage}{0.55\linewidth}
    \begin{align*}
      P(\phrame_+, [\Gamma_B, \phrame'_\epsilon]) &= [\Gamma_B\phrame'_\epsilon, \phrame] \\
      P(\phrame_-, [\Gamma_B\phrame', \phrame]) &= [\Gamma_B, \phrame'] \\
      P(\phrame_-, [\epsilon, \phrame]) &= [\epsilon, \epsilon] \\
      P(\phrame_\pm, \Gamma_{B,T}) &= \Gamma_{B,T} \quad \text{otherwise}
    \end{align*}
  \end{minipage}
  \begin{minipage}{0.40\linewidth}
  \begin{align*}
    D(\uparrow, [\Gamma_B, \phrame\Gamma_T]) &= [\Gamma_B\phrame, \Gamma_T] \\
    D(\downarrow, [\Gamma_B\phrame, \Gamma_T]) &= [\Gamma_B, \phrame\Gamma_T] \\
    D(d, \Gamma_{B,T}) &= \Gamma_{B,T} \quad \text{otherwise}
  \end{align*}
\end{minipage}
\end{center}

The meta-functions $P$ and $D$ perform the stack actions and direct the stack cursor, respectively.
A string $w$ is thus accepted by a 1NSA $A$ iff there are $\qstate \in F, \Gamma_B, \Gamma_T \in \StackAlpha^*$ such that
\begin{equation*}
  (\qstate_0, [\epsilon, \epsilon], w) \longmapsto^* (\qstate, [\Gamma_B,\Gamma_T], \epsilon)
\end{equation*}

Next we develop an introspective form of compact rooted pushdown systems that use 1NSAs for realizable stacks, and prove a correspondence with conditional pushdown systems.

\subsection{Compact conditional pushdown systems}\label{sec:icrpds}

Similar to rooted pushdown systems, we say a conditional pushdown system $G = (\DSStates, \DSFrames, \DSIEdges, \dsstate_0)$ is compact if all states, frames and edges are on some path from the root.
We will refer to this class of conditional pushdown systems as $\mathbb{CCPDS}$.
Assuming we have a way to decide overlap between the set of realizable stacks at a state and a regular language of stacks, we can compute the CCPDS in much the same way as in \autoref{sec:rpds-to-crpds}.

\begin{align*}
  \mkCCPDS(M) &= f\text{, where }
  \\
  M &= (\ControlStates,\StackAlpha,\transfunction,\qstate_0)
  \\
  f(
   \overbrace{\DSStates,\DSFrames,\DSIEdges,\dsstate_0}^G
   ) &= (\DSStates',\DSFrames,\DSIEdges',\dsstate_0) \text{, where }
  \\
  \DSStates' &= \DSStates \union \setbuild{ \dsstate' }{ 
    \dsstate \in \DSStates 
    \text{ and }  
    \dsstate \mathrel{\underset{M}{\RPDTrans}} \dsstate'
  }
  \union \set{\dsstate_0}
  \\
  \DSEdges' &= \DSEdges \union \setbuild{ \dsstate \pdedge^{\hat K, \stackact} \dsstate' }{ 
    \dsstate \in \DSStates 
    \text{ and }  
    \dsstate \mathrel{\overset{\hat K,\stackact}{\underset{M}{\RPDTrans}}} \dsstate'
    \text{ and }
    \Stacks(G)(\dsstate) \cap \hat K \neq \emptyset
  }
  \text.
\end{align*}

\newcommand*{\gadget}{\mathit{gadget}}
The function $\Stacks : \mathbb{CCPDS} \to \DSStates \to 1\mathbb{NSA}$ performs the stack extraction with a construction that inserts the stack-checking NFA for each reduction rule after it has run the cursor to the bottom of the stack, and continues from the final states to the state dictated by the rule (added by meta-function $\gadget$).
All the stack manipulations from $\dsstate_0$ to $\dsstate$ are $\epsilon$-transitions in terms of reading input; only once control reaches $\dsstate$ do we check if the stack is the same as the input, which captures the notion of a stack realizable at $\dsstate$.
Once control reaches $\dsstate$, we run down to the bottom of the stack again, and then match the stack against the input; complete matches are accepted.
To determine the bottom and top of the stack, we add distinct sentinel symbols to the stack alphabet, \textcent{} and \$.

\begin{align*}
\Stacks
(
 \overbrace{\DSStates, \StackAlpha, \DSEdges, \dsstate_0}^{G}
)
(\dsstate)
  &= 
  (\DSStates \cup \DSStates', \StackAlpha, \StackAlpha \cup \set{\text\textcent, \$}, \transfunction, \dsstate_{\text{start}}, \set{\dsstate_{\text{final}}} ) \text{, where }
  \\
  \dsstate_{\text{start}}, \dsstate_{\text{down}}, \dsstate_{\text{check}}, \dsstate_{\text{final}} \text{ fresh, and } &\DSStates', \transfunction \text{ the smallest sets such that} \\
  \set{\dsstate_{\text{start}}, \dsstate_{\text{down}}, \dsstate_{\text{check}}, \dsstate_{\text{final}}} &\subseteq \DSStates' \\
  (\dsstate_{\text{start}}, \epsilon, \epsilon, \cdot, \text\textcent_+, \dsstate_0) &\in \transfunction
  \\
  \gadget(\dsstate', \hat K, \stackchar_\pm, \dsstate'') \sqsubseteq (\transfunction, \DSStates')
  & \text{ if } (\dsstate', \hat K, \stackchar_\pm, \dsstate'') \in \DSEdges
 \\ 
 (\dsstate, \epsilon, \epsilon, \cdot, \$_+, \dsstate_{\text{down}}) &\in \transfunction
 \\ 
 (\dsstate_{\text{down}}, \epsilon, \epsilon, \downarrow, \epsilon, \dsstate_{\text{down}}) &\in \transfunction
 \\ 
 (\dsstate_{\text{down}}, \text\textcent, \epsilon, \uparrow, \epsilon, \dsstate_{\text{check}}) &\in \transfunction
 \\ 
 (\dsstate_{\text{check}}, a, a, \uparrow, \epsilon, \dsstate_{\text{check}}) &\in \transfunction \text{, } a \in \StackAlpha\cup \set{\epsilon}
 \\ 
 (\dsstate_{\text{check}}, \$, \epsilon, \uparrow, \epsilon, \dsstate_{\text{final}}) &\in \transfunction 
\end{align*}

The first rule changes the initial state to initialize the stack with the ``bottom'' sentinel.
Every reduction of the CPDS is given the gadget discussed above and explained below.
The last five rules are what implement the final checking of stack against input.
When at the state we are recognizing realizable stacks for, the machine will have the cursor at the top of the stack, so we push the ``top'' sentinel before moving the cursor all the way down to the bottom.
When $\dsstate_{\text{down}}$ finds the bottom sentinel at the cursor, it moves the cursor past it to start the exact matching in $\dsstate_{\text{check}}$.
If the cursor matches the input exactly, we consume the input and move the cursor past the matched character to start again.
When $\dsstate_{\text{check}}$ finds the top sentinel, it transitions to the final state; if the input is not completely exhausted, the machine will get stuck and not accept.

\begin{align*}
\gadget(\dsstate, \hat K, \stackchar_\pm, \dsstate') &= (\transfunction', Q \cup \set{\qstate_{\text{down}}, \qstate_{\text{out}}})\text{ where} \\
 \text{Let } N = (Q, \Sigma, \transfunction, \qstate_0, F) &\text{ be a fresh NFA recognizing } \hat K
 \text{, } \qstate_{\text{down}}, \qstate_{\text{out}} \text{ fresh states}
 \\ 
 (\qstate, a, \epsilon, \uparrow, \epsilon, \qstate') \in \transfunction' &\text{ if } (\qstate, a, \qstate') \in \transfunction \text{, } a \in \Sigma
 \\ 
 (\qstate, \epsilon, \epsilon, \cdot, \epsilon, \qstate') \in \transfunction' &\text{ if } (\qstate, \epsilon, \qstate') \in \transfunction
 \\ 
 (\qstate, \$, \epsilon, \cdot, \$_-, \qstate_{\text{out}}) \in \transfunction' &\text{ if } \qstate \in F
 \\ 
 (\qstate_{\text{out}}, \epsilon, \epsilon, \cdot, \stackchar_\pm, \dsstate') \in \transfunction'
 \\ 
 (\dsstate, \epsilon, \epsilon, \cdot, \$_+, \qstate_{\text{down}}) \in \transfunction'
 \\ 
 (\qstate_{\text{down}}, \epsilon, \epsilon, \downarrow, \epsilon, \qstate_{\text{down}}) \in \transfunction'
 \\ 
 (\qstate_{\text{down}}, \text\textcent, \epsilon, \uparrow, \epsilon, \qstate_0) \in \transfunction'
\end{align*}

We explain each rule in order.
When the NFA that recognizes $\hat K$ consumes a character, the stack automaton should similarly read the character on the stack and move the cursor along.
If the NFA makes an $\epsilon$-transition, the stack automaton should also, without moving the stack cursor.
When this sub-machine $N$ is in a final state, the cursor should be at the top of the stack (if indeed it matched), so we pop off the top sentinel and proceed to do the stack action the IPDS does when transitioning to the next state.
The last three rules implement the same ``run down to the bottom'' gadget used before, when matching the stack against the input.

Finally, we can show that states are reachable in a conditional pushdown system iff they are reached in their corresponding CCPDS.
Consider a map
\begin{equation*}
  \fCCPDS : \mathbb{CPDS} \to \mathbb{CCPDS}
\end{equation*}
such that given a conditional pushdown system $M = (\QStates, \StackAlpha, \transfunction, \qstate_0)$, its equivalent CCPDS is $\fCCPDS(M) = (\DSStates, \StackAlpha, \DSIEdges, \qstate_0)$ where $\DSStates$ contains reachable nodes:
\begin{equation*}
  \DSStates = \setbuild{\qstate}{(\qstate_0,\vect{}) \mathrel{\underset{M}{\PDTrans}^*} (\qstate,\acont)}
\end{equation*}
and the set $\DSIEdges$ contains reachable edges:
\begin{equation*}
  \DSIEdges = \setbuild{\qstate \pdedge^{\hat K, \stackact} \qstate'}{\qstate \mathrel{\overset{\hat K, \stackact}{\underset{M}{\RPDTrans}}} \qstate'}
\end{equation*}
\begin{theorem}[Computable reachability]\label{thm:ripds-to-icrpds}
  For all $M \in \mathbb{CPDS}$, $\fCCPDS(M) = \lfp(\mkCCPDS(M))$
\end{theorem}
Proof in Appendix~\ref{sec:ripds-reach}.

\begin{corollary}[Realizable stacks of CPDSs are recognized by 1NSAs]
  For all $M = (\QStates, \StackAlpha, \transfunction, \qstate_0) \in \mathbb{CPDS}$,
  and $(\DSStates, \StackAlpha, \DSIEdges, \qstate_0) = \lfp(\mkCCPDS(M))$,
  $(\qstate_0,\vect{}) \mathrel{{\underset{M}{\PDTrans}}^*} (\qstate,\acont)$
  iff $\qstate \in \DSStates$ and $\Stacks(G)(\qstate)$ accepts $\acont$.
\end{corollary}

\subsection{Simplifying garbage collection in conditional pushdown systems}\label{sec:gc-pdcfa}

The decision problems on 1NSAs are computationally intractable in general, but luckily GC is a special problem where we do not need the full power of 1NSAs.
There are equally precise techniques at much lower cost,
and less precise techniques that can shrink the explored state space.\footnote{The added precision of GC with tighter working sets makes the state space comparison between the two approaches non-binary. Neither approach is clearly better in terms of performance.}
The transition relation we build does not enumerate all sets of addresses, but instead queries the graph for the sets of addresses it should consider in order to apply GC.
A fully precise method to manage the stack root addresses is to add the root addresses to the representation of each state, and update it incrementally.
The root addresses can be seen as the representation of $\hat K$ in edge labels, but to maintain the precision, the set must also distinguish control states.
This addition to the state space is an effective \emph{reification} of the stack filtering that conditional performs.

\begin{figure}
\figrule
  \centering
  \begin{align*}
  \atf_\expr((\hat P,\hat E), \hat H) &= ((\hat P',\hat E'), \hat H')
  \text{, where }
\end{align*}
\begin{align*}
  \hat E_+ &= \setbuild{ \triedge{(\apstate,A)}{\aphrame_+}{(\apstate',A\cup\touches(\aphrame))} }{
    \ipdcfato{\apstate}{A}{\aphrame_+}{\apstate'}
  }
  \\
  \hat E_\epsilon &= \setbuild{ \triedge{(\apstate,A)}{\epsilon}{(\apstate',A)} }{
    \ipdcfato{\apstate}{A}{\epsilon}{\apstate'} }
  \\
 \hat E_- &= 
  \big\{
    \triedge{(\apstate'',A)}{\aphrame_-}{(\apstate''',A')} 
  :
    \ipdcfato{\apstate''}{A}{\aphrame_-}{\apstate'''} \text{ and }  
    \triedge{(\apstate,A')}{\aphrame_+}{(\apstate',A)} \in \hat E \text{ and }
    \biedge{(\apstate',A)}{(\apstate'',A)} \in \hat H
    \big\}
  \\
  \hat E' &= {\hat E_+} \union {\hat E_\epsilon} \union {\hat E_-}
  \\
  \hat H_{\epsilon} &= \setbuild{ \biedge{\aopstate}{\aopstate''} }{
    \biedge{\aopstate}{\aopstate'} \in \hat H \text{ and } 
      \biedge{\aopstate'}{\aopstate''} \in \hat H
  }
  \\
  \hat H_{{+}{-}} &= \big\{
  \biedge{\aopstate}{\aopstate'''}
    : 
    \triedge{\aopstate}{\aphrame_+}{\aopstate'} \in \hat E
    \text{ and }
    \biedge{\aopstate'}{\aopstate''} \in \hat H 
    \\
    & \hspace{6.1em} \text{ and } 
    \triedge{\aopstate''}{\aphrame_-}{\aopstate'''} \in \hat E
    \big\}
  \\
  \hat H' &=   \hat H_{\epsilon} \union   \hat H_{{+}{-}} 
  \\
  \hat P' &= \hat P \union \setbuild{ \aopstate' }{ \triedge{\aopstate}{\stackact}{\aopstate'} } \cup \set{((\expr,\bot,\bot),\emptyset)}
  \text.
\end{align*}
  \caption{An \ecg{}-powered iteration function for pushdown garbage-collecting control-flow analysis}
  \label{fig:full-gc}
\figrule
\end{figure}

An approximative method is to not distinguish control states, but rather to traverse the graph backward through $\epsilon$-closure edges and push edges, and collect the root addresses through the pushed frames.
As more paths are discovered to control states, more stacks will be realizable there, which add more to the stack root addresses to consider as the relation steps forward.
For soundness, edges still must be labeled with the language of stacks they are valid for, since they can become invalid as more stacks reach control states.
Notice that the root sets of addresses are isomorphic to languages of stacks that have the given root set, so we can use sets of addresses as the language representation.

We consider both methods in turn, augmenting the compaction algorithm from \autoref{sec:pdcfa-eps}.
Each have program states that consist of the expression, environment, \emph{and} store; $\apstate \in \sa{PState} = \syn{Exp} \times \sa{Env} \times \sa{Store}$.
Since GC is a non-monotonic operation, stores cannot be shared globally without sacrificing the precision benefits of GC.
For the first method, program states additionally include the stack root set of addresses; we will call these ornamented program states, $\aopstate \in \sa{OPState} = \sa{PState} \times \PowSm{\sa{Addr}}$.
We show the non-worklist solution to computing reachability by
employing the function $\atf_{\expr}$ defined in
\autoref{fig:full-gc}.
In order to define the iteration function, we need a refactored
transition relation $\overset{A}{\underset{\stackact}{\apTo}}
\subseteq \sa{PState} \times \sa{PState}$, defined as follows:

{\small
\begin{align*}
  \ipdcfato{(\expr,\aenv,\astore)}{A}{\aphrame_+}{(\expr',\aenv',\astore')}
  &\text{ iff }
  \StackRoot(\acont) = A \text{ and } 
  \hat G(\expr,\aenv,\astore,\acont)
  \aTo
  (\expr',\aenv',\astore',\aphrame:\acont)
  \\
  \ipdcfato{(\expr,\aenv,\astore)}{A}{\aphrame_-}{(\expr',\aenv',\astore')}
  &\text{ iff }
  \StackRoot(\aphrame:\acont) = A \text{ and }
  \hat G(\expr,\aenv,\astore,\aphrame:\acont)
  \aTo
  (\expr',\aenv',\astore',\acont)
  \\
  \ipdcfato{(\expr,\aenv,\astore)}{A}{\epsilon}{(\expr',\aenv',\astore')}
  &\text{ iff }
  \StackRoot(\acont) = A \text{ and }
  \hat G(\expr,\aenv,\astore,\acont)
  \aTo
  (\expr',\aenv',\astore',\acont)
\end{align*}
}

\begin{theorem}[Correctness of GC specialization]\label{thm:gc-specialization}
  $\lfp(\atf_\expr)$ completely abstracts $\fCCPDS(\afIPDS'(\expr))$
\end{theorem}

The approximative method changes the representation of edges in the graph to contain sets of addresses, $\hat E \in \sa{Edge} = \Pow{\sa{PState} \times \PowSm{\sa{Addr}} \times \sa{Frame}_\pm \times \sa{PState}}$.
We also add in a sub-fixed-point computation for $\att : (\sa{PState} \to \PowSm{\sa{Addr}}) \to (\sa{PState} \to \PowSm{\sa{Addr}})$, to traverse the graph and collect the union of all stack roots for stacks realizable at a state.
Although we show a non-worklist solution here (in \autoref{fig:approx-gc}) to not be distracting, this solution will not compute the same reachable states as a worklist solution due to the ever-growing stack roots at each state.
Only states in the worklist would need to be analyzed at the larger stack root sets.
In other words, the non-worklist solution potentially throws in more live addresses at states that would otherwise not need to be re-examined.

\begin{figure}
\figrule
  \centering
  \begin{align*}
  \atf_\expr'((\hat P,\hat E), \hat H) &= ((\hat P',\hat E'), \hat H')
  \text{, where }
\end{align*}
\begin{align*}
  \hat E_+ &= \setbuild{ \quadedge{\apstate}{A}{\aphrame_+}{\apstate'} }{
    A = {\mathcal R}(\apstate), 
    \ipdcfato{\apstate}{A}{\aphrame_+}{\apstate'}
  }
  \\
  \hat E_\epsilon &= \setbuild{ \quadedge{\apstate}{A}{\epsilon}{\apstate'} }{
    A = {\mathcal R}(\apstate), 
    \ipdcfato{\apstate}{A}{\epsilon}{\apstate'} }
  \\
 \hat E_- &= 
  \big\{
    \quadedge{\apstate''}{A}{\aphrame_-}{\apstate'''} 
  :
    A = {\mathcal R}(\apstate''), 
    \ipdcfato{\apstate''}{A}{\aphrame_-}{\apstate'''} \text{ and }  
    \\
    &\hspace{3.775cm}\quadedge{\apstate}{A'}{\aphrame_+}{\apstate'} \in \hat E \text{ and }
    \\
    &\hspace{3.775cm}\biedge{\apstate'}{\apstate''} \in \hat H
    \big\}
  \\
  \hat E' &= {\hat E_+} \union {\hat E_\epsilon} \union {\hat E_-}
  \\
  \hat H_{\epsilon} &= \setbuild{ \biedge{\apstate}{\apstate''} }{
    \biedge{\apstate}{\apstate'} \in \hat H \text{ and } 
      \biedge{\apstate'}{\apstate''} \in \hat H
  }
  \\
  \hat H_{{+}{-}} &= \big\{
  \biedge{\apstate}{\apstate'''}
    : 
    \quadedge{\apstate}{A}{\aphrame_+}{\apstate'} \in \hat E
    \text{ and }
    \biedge{\apstate'}{\apstate''} \in \hat H 
    \\
    & \hspace{6.1em} \text{ and } 
    \quadedge{\apstate''}{A'}{\aphrame_-}{\apstate'''} \in \hat E
    \big\}
  \\
  \hat H' &=   \hat H_{\epsilon} \union   \hat H_{{+}{-}} 
  \\
  \hat P' &= \hat P \union \setbuild{ \apstate' }{ \quadedge{\apstate}{A}{\stackact}{\apstate'} } \cup \set{(\expr,\bot,\bot)}
  \\
  \att({\mathcal R}) &= \lambda \apstate. \bigcup (\setbuild{\touches(\aphrame) \cup {\mathcal R}(\apstate')}{\quadedge{\apstate'}{A}{\aphrame_+}{\apstate} \in \hat E} \cup \setbuild{{\mathcal R}(\apstate')}{\biedge{\apstate'}{\apstate} \in \hat H})
  \\
  {\mathcal R} &= \lfp(\att)
  \text.
\end{align*}
  \caption{Approximate pushdown garbage-collecting control-flow analysis.}
  \label{fig:approx-gc}
\figrule
\end{figure}

This approximation is not an easily described introspective pushdown system since the root sets it uses depend on the iteration state --- particularly what frames have reached a state \emph{so far}, regardless of the stack filtering the original CPDS performs.
The regular sets of stacks acceptable at some state can be extracted
\emph{a posteriori} from the fixed point of the function $\atf_\expr'$
defined in \autoref{fig:approx-gc}, if so desired.
The next theorem follows from the fact that ${\mathcal R}(\apstate)
\supseteq A$ for any represented~$(\apstate, A)$.

\begin{theorem}[Approximate GC is sound]\label{thm:approx-gc-sound}
  $\lfp(\atf_\expr')$ approximates $\lfp(\atf_\expr)$.
\end{theorem}

The last thing to notice is that by disregarding the filtering, the stack root set can get larger and render previous GCs unsound, since more addresses can end up live than were previously considered.
Thus we label edges with the root set for which the GC was considered,
in order to not make false predictions.

\section{Implementing Introspective Pushdown Analysis with Garbage
  Collection}
\label{sec:implementation}

The reachability-based analysis for a pushdown system described in
the previous section requires two mutually-dependent pieces
of information in order to add another edge:

\begin{enumerate}
\item The topmost frame on a stack for a given control state $q$. This is
essential for \emph{return} transitions, as this frame should be
popped from the stack and the store and the environment of a caller
should be updated respectively.

\item Whether a given control state $q$ is reachable or not from the
initial state $q_0$ along realizable sequences of stack actions.
For example, a path from $q_0$ to $q$ along edges labeled ``push, pop, pop, push''
is not realizable: the stack is empty after the first pop, so the second pop
cannot happen---let alone the subsequent push.

%

\end{enumerate}

Knowing about a possible topmost frame on a stack and initial-state
reachability is enough for a classic pushdown reachability
summarization to proceed one step further, and we presented an
efficient algorithm to compute those in
Section~\ref{sec:ecg-worklist}.
However, to deal with the presence of an abstract GC in a
\emph{conditional} PDS, we add:

\begin{enumerate}
\item[3.] For a given control state $q$, what are the touched addresses of \emph{all} possible frames
that could happen to be \emph{on} the stack  at the moment the CPDS
is in the state $q$?
\end{enumerate}




The crucial addition to the algorithm is maintaining for each node
$\qstate'$ in the CRPDS a set of $\epsilon$-\emph{predecessors}, i.e.,
nodes $\qstate$, such that $\qstate \mathrel{\overset{\vec{\stackact}}{\underset{M}{\RPDTrans}}}
\qstate'$ and $[\vec{\stackact}] = \epsilon$.
In fact, only two out of three kinds
of transitions can cause a change to the set of
$\epsilon$-predecessors for a particular node $\qstate$: an addition of
an $\epsilon$-edge or a pop edge to the CRPDS.


One can notice a subtle mutual dependency between computation of
$\epsilon$-predecessors and top frames during the construction of a
CRPDS.
Informally:

\begin{itemize}
\item A \emph{top frame} for a state $q$ can be pushed as a direct
  predecessor (\eg, $q$ follows a nested let-binding), or as a direct
  predecessor to an $\epsilon$-predecessor (\eg, $q$ is in tail
  position and will return to a waiting let-binding).

\item When a new $\epsilon$-edge $\qstate \xrightarrow{\epsilon}
  \qstate'$ is added, all $\epsilon$-predecessors of $\qstate$ become
  also $\epsilon$-predecessors of $\qstate'$. 
  That is, $\epsilon$-summary edges are transitive.

\item When a $\stackchar_-$-pop-edge $\qstate
  \xrightarrow{\stackchar_-} \qstate'$ is added, new
  $\epsilon$-predecessors of a state $\qstate_1$ can be obtained by
  checking if $\qstate'$ is an $\epsilon$-predecessor of $\qstate_1$
  and examining all existing $\epsilon$-predecessors of $\qstate$,
  such that $\stackchar_+$ is their possible top frame: this situation
  is similar to the one depicted in the example above. 

\end{itemize}

The third component---the touched addresses of \emph{all} possible frames on the stack for a state
$\qstate$---is straightforward to compute with $\epsilon$-predecessors:
starting from $\qstate$, trace out only the edges which are labeled $\epsilon$
(summary or otherwise) or $\stackchar_+$.
The frame for any action $\stackchar_+$ in this trace is a possible stack
action.
Since these sets grow monotonically, it is easy to cache the results
of the trace, and in fact, propagate incremental changes to these
caches when new $\epsilon$-summary or $\stackchar_+$ nodes are
introduced.
This implementation strategy captures the approximative approach to performing GC, as discussed in the previous section.
%
%
Our implementation directly reflects the optimizations discussed above.

\section{Experimental Evaluation}
\label{sec:experiments}

A fair comparison between different families of analyses
should compare both precision and speed.
We have implemented a version $k$-CFA for a subset of R5RS Scheme and
instrumented it with a possibility to optionally enable pushdown
analysis, abstract garbage collection or both.
Our implementation source (in Scala) and benchmarks are available:
\begin{center}
\url{http://github.com/ilyasergey/reachability}
\end{center}

In the experiments, we have focused on the version of $k$-CFA with a
per-state store (\ie, \emph{without} widening), as in the presence of
single-threaded store, the effect of abstract GC is neutralized due to
merging. For non-widened versions of $k$-CFA, as expected, the fused
analysis does at least as well as the best of either analysis alone in
terms of singleton flow sets (a good metric for program
optimizability) and better than both in some cases.
Also worthy of note is the dramatic reduction in the size of the
abstract transition graph for the fused analysis---even on top of the
already large reductions achieved by abstract garbage collection and
pushdown flow analysis individually.
The size of the abstract transition graph is a good heuristic measure of the temporal
reasoning ability of the analysis, \eg, its ability to support model-checking
of safety and liveness
properties~\cite{mattmight:Might:2007:ModelChecking}.

\subsection{Plain $k$-CFA vs. pushdown $k$-CFA }
\label{sec:plain-k-cfa}

\begin{table}
\centering
\scriptsize

\newcommand{\spc}[1]{\phantom{a}{#1}\phantom{a}}
\newcommand{\spcc}[1]{\spc{\spc{#1}}}


\begin{tabular}{|l|c|c|c||c|c|c||c|c|c||c|c|c||c|c|c|}

\hline
\multirow{2}{*}{\spc{Program}} & 
\multirow{2}{*}{\spc{$\# e$}} & 
\multirow{2}{*}{\spc{$\#v$}} & 
\multirow{2}{*}{\spc{$k$}} &
\multicolumn{3}{c||}{\multirow{2}{*}{$k$-CFA}} &
\multicolumn{3}{c||}{\multirow{2}{*}{$k$-PDCFA}} &
\multicolumn{3}{c||}{\multirow{2}{*}{\spcc{$k$-CFA + GC}}} &
\multicolumn{3}{c|}{\multirow{2}{*}{\spc{$k$-PDCFA + GC}}} \\

&&&& \multicolumn{3}{c||}{} & \multicolumn{3}{c||}{} & \multicolumn{3}{c||}{}
& \multicolumn{3}{c|}{} \\

\hline \hline


\multirow{2}{*}{\spc{\texttt{mj09}}} & 
\multirow{2}{*}{19} & 
\multirow{2}{*}{8} &
0 &
83 & 107 & \spcc{4}  &
38 & 38 & \spc{4} &
36 & 39 & 4  &
33 & 32 & 4 \\
&&&
1 &
454 & 812 & 1 &
44 & 48 & 1 &
34 & 35 & 1 &
32 & 31 & 1 \\
\hline


\multirow{2}{*}{\spc{\texttt{eta}}} & 
\multirow{2}{*}{21} & 
\multirow{2}{*}{13} &
0 &
63 & 74 & 4 &
32 & 32 & 6 &
28 & 27 & 8 &
28 & 27 & 8 \\
&&&
1 &
33 & 33 & 8 &
32 & 31 & 8 &
28 & 27 & 8 &
28 & 27 & 8 \\
\hline


\multirow{2}{*}{\spc{\texttt{kcfa2}}} & 
\multirow{2}{*}{20} & 
\multirow{2}{*}{10} &
0 &
194 & 236  & 3 &
36 & 35 & 4 &
35 & 34 & 4 &
35 & 34 & 4 \\

&&&
1 &
970 & 1935 & 1 &
87 & 144 & 2 &
35 & 34  & 2 &
35 & 34  & 2  \\
\hline


\multirow{2}{*}{\spc{\texttt{kcfa3}}} & 
\multirow{2}{*}{25} & 
\multirow{2}{*}{13} &
0 &
272 & 327 & 4 &
58 & 63 & 5 &
53 & 52 & 5 &
53 & 52 & 5 \\

&&&
1 &
\spc{$>$ 32662} & \spc{$>$ 88548} & -- &
1761 & 4046 & 2  &
53 & 52 & 2 &
53 & 52 & 2 \\
\hline


\multirow{2}{*}{\spc{\texttt{blur}}} & 
\multirow{2}{*}{40} & 
\multirow{2}{*}{20} &
0 &
4686 & 7606 & 4  &
115 & 146 & 4 &
90 & 95 & 10 &
68 & 76 & 10 \\

&&&
1 &
123 & 149 & 10 &
94 & 101 & 10 &
76 & 82 & 10 &
75 & 81 & 10 \\
\hline


\multirow{2}{*}{\spc{\texttt{loop2}}} & 
\multirow{2}{*}{41} & 
\multirow{2}{*}{14} &
0 &
149 & 163 & 7 &
69 & 73 & 7 &
43 & 46 & 7 &
34 & 35  & 7  \\

&&&
1 &
$>$ 10867 & $>$ 16040 & -- &
411 & 525 & 3  &
151 & 163 & 3  &
145 & 156 & 3 \\
\hline


\multirow{2}{*}{\spc{\texttt{sat}}} & 
\multirow{2}{*}{51} & 
\multirow{2}{*}{23} &
0 &
3844 & 5547 & 4 &
545 & 773 & 4 &
\spc{1137} & \spc{1543} & 4 &
254 & 317 & 4 \\

&&&
1 &
$>$ 28432 & $>$ 37391 & -- &
\spc{12828} & \spc{16846} & 4 &
958 & 1314 & \spc{5} &
\spc{71} & \spc{73} & 10 \\
\hline


\end{tabular}

\vspace{0.75cm}

\caption{Benchmark results for toy programs. The first three columns
  provide the name of a benchmark, the number of expressions and
  variables in the program in the ANF, respectively. For each of eight
  combinations of pushdown analysis, $k \in \set{0, 1}$ and garbage
  collection on or off, the first two columns in a group show the
  number of \emph{control states} and transitions/CRPDS edges computed
  during the analysis (for both less is better). The third column
  presents the amount of \emph{singleton} variables, i.e, how many
  variables have a single lambda flow to them (more is
  better). Inequalities for some results of the plain $k$-CFA denote
  the case when the analysis explored more than $10^5$
  \emph{configurations} (\ie, control states coupled with
  continuations) or did not finish within 30 minutes. For such cases
  we do not report on singleton variables.  }
\label{fig:table-results}
\figrule
\end{table}

In order to exercise both well-known and newly-presented instances of
CESK-based CFAs, we took a series of \emph{small} benchmarks
exhibiting archetypal control-flow patterns (see
Table~\ref{fig:table-results}).
Most benchmarks are taken from the CFA literature: \texttt{mj09} is a
running example from the work of Midtgaard and Jensen designed to
exhibit a non-trivial return-flow
behavior~\cite{mattmight:Midtgaard:2007:Dissertation}, \texttt{eta}
and \texttt{blur} test common functional idioms, mixing closures and
eta-expansion, \texttt{kcfa2} and \texttt{kcfa3} are two worst-case
examples extracted from the proof of $k$-CFA's EXPTIME
hardness~\cite{dvanhorn:VanHorn-Mairson:ICFP08}, \texttt{loop2} is
an example from Might's dissertation that was used to demonstrate
the impact of abstract GC~\cite[Section
13.3]{mattmight:Might:2007:Dissertation}, \texttt{sat} is a
brute-force SAT-solver with backtracking.

\subsubsection{Comparing precision}
\label{sec:comparing-precision}

In terms of precision, the fusion of pushdown analysis and abstract
garbage collection substantially cuts abstract transition graph sizes
over one technique alone.

We also measure singleton flow sets as a heuristic metric for precision.
Singleton flow sets are a necessary precursor to optimizations such as
flow-driven inlining, type-check elimination and constant propagation.
It is essential to notice that for the experiments in
Table~\ref{fig:table-results} our implementation computed the sets of
\emph{values} (\ie, closures) assigned to each variable (as opposed to
mere \emph{syntactic lambdas}). This is why in some cases the results
computed by 0CFA appear to be better than those by 1CFA: the later one
may examine \emph{more} states with different environments, which
results in exploring more different values, whereas the former one
will just collapse all these values to a single
lambda~\cite{mattmight:Might:2010:mCFA}.
What is important is that for a fixed $k$ the fused analysis prevails
as the best-of- or better-than-both-worlds.

Running on the benchmarks, we have re-validated hypotheses about
the improvements to precision granted by both pushdown
analysis~\cite{mattmight:Vardoulakis:2010:CFA2} and abstract
garbage collection~\cite{mattmight:Might:2007:Dissertation}.
Table~\ref{fig:table-results} contains our detailed
results on the precision of the analysis. In order to make the
comparison fair, in the table we report on the numbers of
\emph{control states}, which do not contain a stack component and are
the nodes of the constructed CRPDS. In the case of plain $k$-CFA,
control states are coupled with stack pointers to obtain
\emph{configurations}, whose resulting number is significantly bigger.

The SAT-solving benchmark showed a dramatic improvement with the
addition of context-sensitivity.
Evaluation of the results showed that context-sensitivity provided
enough fuel to eliminate most of the non-determinism from the
analysis.

\subsubsection{Comparing speed}
\label{sec:comparing-speed}

\begin{figure*}
\centering
\includegraphics[scale=0.8]{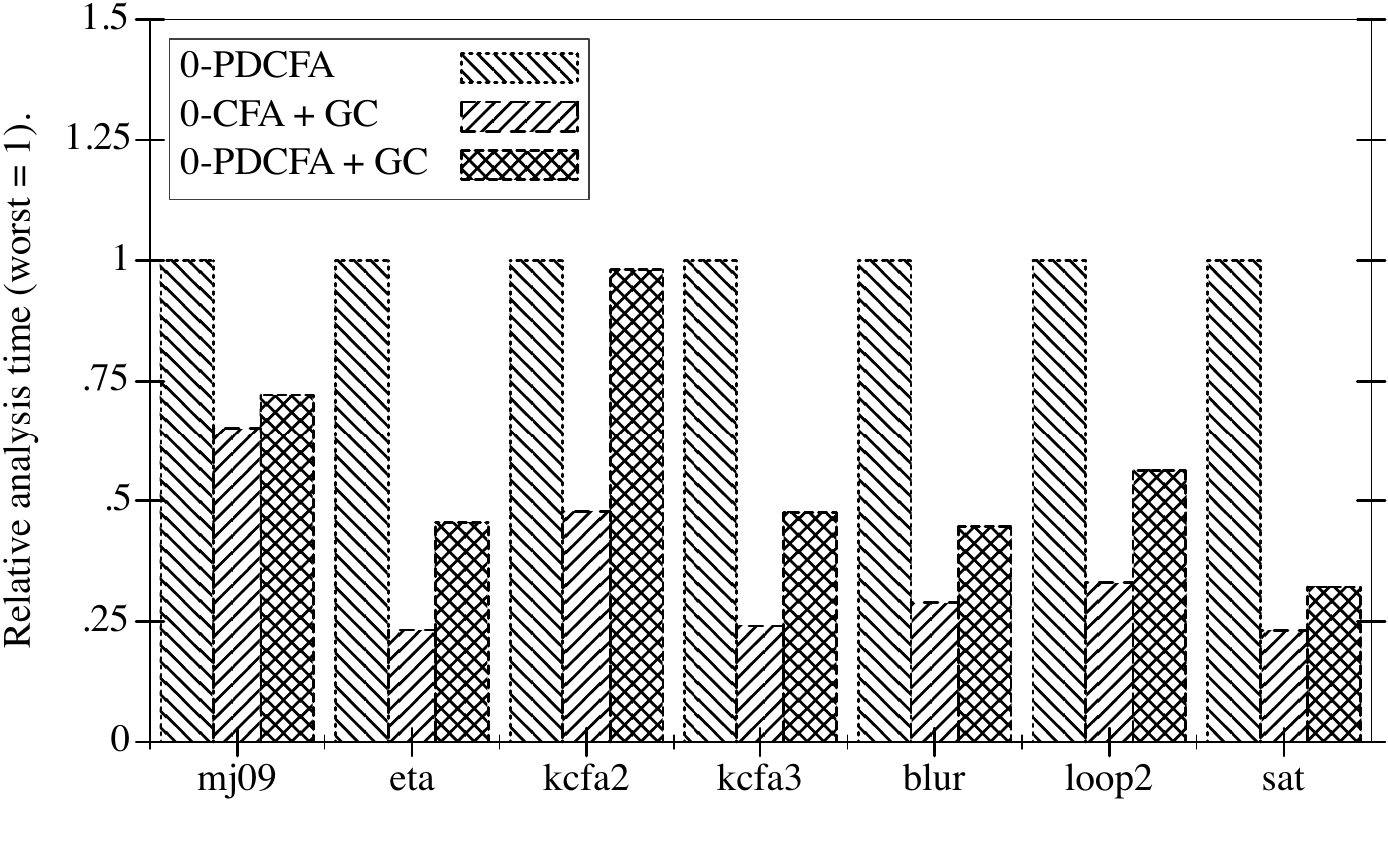}
\hfill
\includegraphics[scale=0.8]{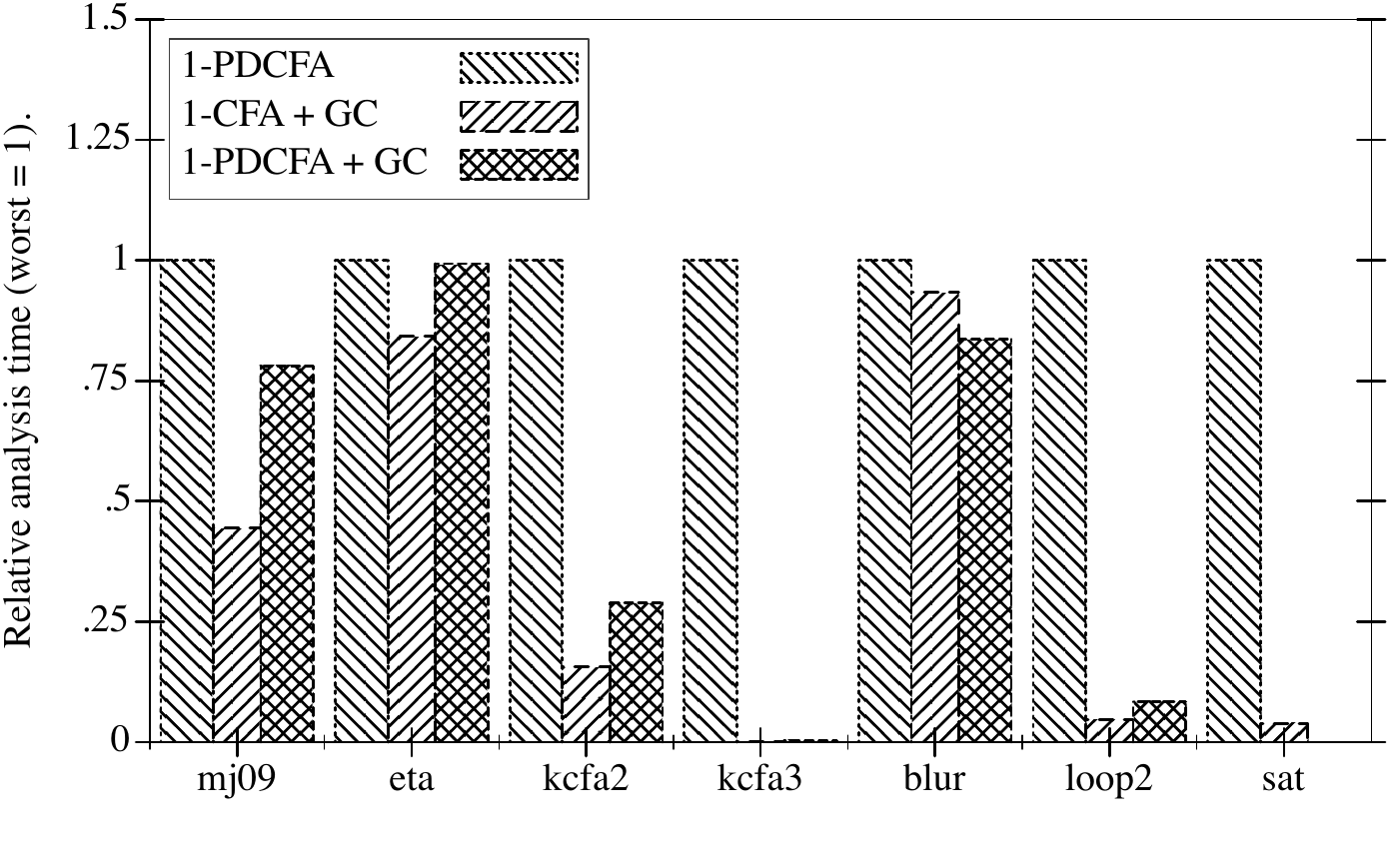}
\caption{Analysis times relative to worst (= 1) in class; smaller is
  better.  At the top is the monovariant 0CFA class of analyses, at
  the bottom is the polyvariant 1CFA class of analyses. (Non-GC
  $k$-CFA omitted.)}
\label{fig:execution-times}
\figrule
\end{figure*}

In the original work on CFA2, Vardoulakis and Shivers present
experimental results with a remark that the running time of the
analysis is proportional to the size of the reachable
states~\cite[Section 6]{mattmight:Vardoulakis:2010:CFA2}. There
is a similar correlation in our fused analysis,
but with higher variance due to the live address set computation GC performs.
Since most of the programs from our toy suite run for less than a
second, we do not report on the absolute time. Instead, the histogram
on Figure~\ref{fig:execution-times} presents normalized relative times
of analyses' executions. To our observation the pure machine-style
$k$-CFA is always significantly worse in terms of execution time than
either with GC or push-down system, so we excluded the plain,
non-optimized $k$-CFA from the comparison.

Our earlier implementation of a garbage-collecting pushdown
analysis~\cite{mattmight:Earl:2012:Introspective} did not fully
exploit the opportunities for caching $\epsilon$-predecessors, as
described in Section~\ref{sec:implementation}. This led to significant
inefficiencies of the garbage-collecting analyzer with respect to the
regular $k$-CFA, even though the former one observed a smaller amount of
states and in some cases found larger amounts of singleton
variables. After this issue has been fixed, it became clearly visible
that in all cases the GC-optimized analyzer is strictly faster than
its non-optimized pushdown counterpart.

Although caching of $\epsilon$-predecessors and $\epsilon$-summary
edges is relatively cheap, it is not free, since maintaining the
caches requires some routine machinery at each iteration of the
analyzer. This explains the loss in performance of the
garbage-collecting pushdown analysis with respect to the GC-optimized
$k$-CFA.

As it follows from the plot, fused analysis is always faster than
the non-garbage-collecting pushdown analysis, and about a fifth of the time, it beats
$k$-CFA with garbage collection in terms of performance.
When the fused analysis is slower than just a GC-optimized one, it is
generally not much worse than twice as slow as the next slowest
analysis.
Given the already substantial reductions in analysis times provided by
collection and pushdown analysis, the amortized penalty is a small and
acceptable price to pay for improvements to precision.

\subsection{Analyzing real-world programs with garbage-collecting pushdown CFA}
\label{sec:analysing-real-life}

\begin{table}
\centering
\scriptsize

\newcommand{\spc}[1]{\phantom{a}{#1}\phantom{a}}
\newcommand{\lspc}[1]{{#1}\phantom{a}}
\newcommand{\spcc}[1]{\spc{\spc{#1}}}
\newcommand{\mins}{$'$}
\newcommand{\secs}{$''$}
\newcommand{\linf}{\large{$\infty$}}


\begin{tabular}{|l|ccc|ccl|ccc|ccc|ccl|}

\hline
\multirow{2}{*}{\spc{Program}} & 
\multirow{2}{*}{\spc{$\# e$}} & 
\multirow{2}{*}{\spc{$\#v$}} & 
\multirow{2}{*}{\spc{!$\#v$}} &
\multicolumn{3}{c|}{\multirow{2}{*}{$k=0$, GC off}} &
\multicolumn{3}{c|}{\multirow{2}{*}{\spcc{$k=0$, GC on}}} &
\multicolumn{3}{c|}{\multirow{2}{*}{\spcc{$k=1$, GC off}}} &
\multicolumn{3}{c|}{\multirow{2}{*}{\spc{$k=1$, GC on}}} \\ 

&&&&&&&&&&&&&&&\\

\hline \hline


\spc{\texttt{primtest}} & 155 & 44 & 16 &
\spc{790} & \spc{955} & 14\secs &
\spc{113} & \spc{127} &  1\secs &
\spc{$>$43146} & \spc{$>$54679} & \spc{\linf} &
\spc{442} & \spc{562} & 13\secs \\

\spc{\texttt{rsa}} & 211 & 93 & 36 &
\spc{1267} & \spc{1507} & 23\secs &
\spc{355} & \spc{407} &  6\secs &
\spc{20746} & \spc{28895} & \spc{21\mins} &
\spc{926} & \spc{1166} & 28\secs \\

\spc{\texttt{regex}} & 344 & 150 & 44 &
\spc{943} & \spc{956} & \lspc{54\secs} &
\spc{578} & \spc{589} &  45\secs &
\spc{1153} & \spc{1179} & \spc{88\secs} &
\spc{578} & \spc{589} & 50\secs \\

\spc{\texttt{scm2java}} & \spc{1135} & 460 & 63 &
\spc{376} & \spc{375} & 13\secs &
\spc{376} & \spc{375} &  13\secs &
\spc{376} & \spc{375} & 14\secs &
\spc{376} & \spc{375} & \lspc{13\secs} \\

\hline

\end{tabular}

\vspace{0.75cm}

\caption{Benchmark results of PDCFA on real-world programs. The first four columns provide the name
  of a program, the number of expressions and variables in the
  program in the ANF, and the number of singleton variables revealed by the analysis (same in all cases). For each of four combinations of
  $k \in \set{0, 1}$ and garbage collection on or
  off, the first two columns in a group show the number of visited 
  control states and edges, respectively, and the third one shows absolute time of running the analysis (for both less is better). The results of the analyses are presented in minutes ($'$) or
  seconds ($''$), where {\linf} stands for an analysis, which has been interrupted due to the
  an execution time greater than 30 minutes.}
\label{fig:table-real}
\figrule
\end{table}

Even though our prototype implementation is just a proof of concept,
we evaluated it not on a suite of toy programs, tailored for
particular functional programming patterns, but on a set of real-world
programs. In order to set this experiment, we have chosen four
programs, dealing with numeric and symbolic computations:

\begin{itemize}

\item \texttt{primtest} -- an implementation of the
  probabilistic Fermat and Solovay-Strassen primality testing in
  Scheme for the purpose of large prime generation;

\item \texttt{rsa} -- an implementation of the RSA public-key
  cryptosystem;

\item \texttt{regex} -- an implementation of a regular expression
  matcher in Scheme using Brzozovski
  derivatives~\cite{Might-al:ICFP11,mattmight:Owens:2009:Derivative};

\item \texttt{scm2java} -- scm2java is a Scheme to Java compiler;

\end{itemize}

We ran our benchmark suite on a 2.7~GHz Intel Core~i7 OS~X machine
with 8 Gb RAM. Unfortunately, $k$-CFA without global stores timed out
on most of these examples (\ie, did not finish within 30 minutes), so
we had to exclude it from the comparison and focus on the effect of a
pushdown analyzer only. Table~\ref{fig:table-real} presents the
results of running the benchmarks for $k \in \set{0,1}$ with a garbage
collector on and off. Surprisingly, for each of the six programs,
those cases, which terminated within 30 minutes, found the same number
of global singleton variables.\footnote{Of course, the numbers of
  states explored for each program by different analyses were
  different, and there were variations in \emph{function parameters}
  cardinalities, which we do not report on here.} However, the numbers
of observed states and runtimes are indeed different in most of the
cases except \texttt{scm2java}, for which all the four versions of the
analysis were precise enough to actually evaluate the program:
happily, there was no reuse of abstract addresses, which resulted in
the absence of forking during the CRPDS construction. In other words,
the program \texttt{scm2java}, which used no scalar data but strings
being concatenated and was given a simple input, has been evaluated
\emph{precisely}, which is confirmed by the number of visited control
states and edges.

Time-wise, the results of the experiment demonstrate the general
positive effect of the abstract garbage collection in a pushdown
setting, which might improve the analysis performance by more than two
orders of magnitude.

\section{Discussion: Applications}
\label{sec:applications}

Pushdown control-flow analysis offers more precise control-flow
analysis results than the classical finite-state CFAs.
Consequently, introspective pushdown control-flow analysis improves flow-driven
optimizations (\eg, constant propagation, global register allocation,
inlining~\cite{mattmight:Shivers:1991:CFA}) by eliminating more of
the false positives that block their application.

The more compelling applications of pushdown control-flow analysis
are those which are difficult to drive with classical control-flow
analysis.
Perhaps not surprisingly, the best examples of such analyses are
escape analysis and interprocedural dependence analysis.
Both of these analyses are limited by a static analyzer's ability to
reason about the stack, the core competency of introspective pushdown control-flow
analysis.
(We leave an in-depth formulation and study of these
analyses to future work.)

\subsection{Escape analysis}
In escape analysis, the objective is to determine whether a
heap-allocated object is safely convertible into a stack-allocated
object.
In other words, the compiler is trying to figure out whether the frame
in which an object is allocated outlasts the object itself.
In higher-order languages, closures are candidates for escape
analysis.

Determining whether all closures over a particular \lamterm{} $\lam$
may be heap-allocated is straightforward: find the control states in
the compact rooted pushdown system in which closures over $\lam$ are being created,
then find all control states reachable from these states
over only $\epsilon$-edge and push-edge transitions.
Call this set of control states the ``safe'' set.
Now find all control states which are invoking a closure over $\lam$.
If any of these control states lies outside of the safe set, then
stack-allocation may not be safe; if, however, all invocations lie
within the safe set, then stack-allocation of the closure is safe.

\subsection{Interprocedural dependence analysis}
In interprocedural dependence analysis, the goal is to determine, for
each \lamterm{}, the set of resources which it may read or write when
it is called.
\citet{mattmight:Might:2009:Dependence} showed that if one has knowledge of the program
stack, then one can uncover interprocedural
dependencies.
We can adapt that technique to work with compact rooted pushdown systems.
For each control state, find the set of reachable control states along
only $\epsilon$-edges and pop-edges.
The frames on the pop-edges determine the frames which could have been
on the stack when in the control state.
The frames that are live on the stack determine the procedures that are live on the stack.
Every procedure that is live on the stack has a read-dependence on any
resource being read in the control state, while every procedure that is live
on the stack also has a write-dependence on any resource being written in
the control state.
In control-flow terms, this translates to ``if $f$ calls $g$ and $g$
accesses $a$, then $f$ also accesses $a$.''

\section{Related Work}
\label{sec:related}

The Scheme workshop presentation of PDCFA~\citep{dvanhorn:Earl2010Pushdown} is not archival, nor were there rigorous proofs of correctness.
The complete development of pushdown analysis from first principles stands as a
new contribution, and it constitutes an alternative to CFA2.
It goes beyond work on CFA2 by specifying specific mechanisms for reducing the complexity to polynomial time ($\mathcal{O}(n^6)$) as well.
\citet{local:vardoulakis-diss12} sketches an approach to regain polynomial time in his dissertation, but does not give a precise bound.
An immediate advantage of the complete development is its exposure of
parameters for controlling polyvariance and context-sensitivity.
An earlier version of this work appeared in ICFP 2012~\cite{mattmight:Earl:2012:Introspective}.
We also provide a reference implementation of control-state reachability in Haskell.
We felt this was necessary to shine a light on the ``dark corners'' in the
formalism, and in fact, it helped expose both bugs and implicit design
decisions that were reflected in the revamped text of this work.
The development of introspective pushdown systems is also more complete and
more rigorous.
We expose the critical regularity constraint absent from the ICFP 2012 work,
and we specify the implementation of control-state reachability and feasible paths for
conditional pushdown systems in greater detail.
More importantly, this work uses additional techniques to improve the
performance of the implementation and discusses those changes.

Garbage-collecting pushdown control-flow analysis draws on work in higher-order
control-flow analysis~\cite{mattmight:Shivers:1991:CFA}, abstract
machines~\cite{mattmight:Felleisen:1987:CESK} and abstract
interpretation~\cite{mattmight:Cousot:1977:AI}.

%
%
%
%
%
%

\paragraph{Context-free analysis of higher-order programs}
The motivating work
for our own is Vardoulakis and Shivers recent discovery of
  CFA2.
CFA2 is a table-driven summarization algorithm that exploits the balanced
nature of calls and returns to improve return-flow precision in a control-flow
analysis.
Though CFA2 exploits context-free languages, context-free languages are not
explicit in its formulation in the same way that pushdown systems are explicit in
our presentation of pushdown flow analysis.
With respect to CFA2, our pushdown flow analysis is also polyvariant/context-sensitive (whereas CFA2 is monovariant/context-insensitive), and it covers direct-style.

On the other hand, CFA2 distinguishes stack-allocated and
store-allocated variable bindings, whereas our formulation of pushdown
control-flow analysis does not: it allocates all bindings in the
store.
If CFA2 determines a binding can be allocated on the stack, that
binding will enjoy added precision during the analysis and is not
subject to merging like store-allocated bindings.
While we could incorporate such a feature in our formulation,
it is not necessary for achieving ``pushdownness,''
and in fact, it could be added to classical finite-state CFAs as well.
CFA2 has a follow-up that sacrifices its complete abstraction with the machine that only abstracts bindings in order to handle first-class control~\citep{dvanhorn:Vardoulakis2011Pushdown}.
We do not have an analogous construction since loss of complete abstraction was an anti-goal of this work.
We leave an in-depth study of generalizations of CFA2's method to introspection, polyvariance and other control operators to future work.

\paragraph{Calculation approach to abstract interpretation}

\citet{dvanhorn:Midtgaard2009Controlflow} 
systematically calculate
0CFA using the Cousot-style calculational approach to abstract
interpretation~\cite{dvanhorn:Cousot98-5} applied to an ANF
$\lambda$-calculus.
Like the present work, Midtgaard and Jensen start with the CESK
machine of \citet{mattmight:Flanagan:1993:ANF} and employ a
reachable-states model. 

The analysis is then constructed by composing well-known
Galois connections to reveal a 0CFA incorporating reachability.
The abstract semantics approximate the control stack component of the
machine by its top element.
The authors remark monomorphism materializes in two mappings: one
``mapping all bindings to the same variable,'' the other ``merging all
calling contexts of the same function.''
Essentially, the pushdown 0CFA of Section~\ref{sec:abstraction}
corresponds to Midtgaard and Jensen's analysis when the latter
mapping is omitted and the stack component of the machine is not
abstracted.
However, not abstracting the stack requires non-trivial mechanisms to compute the compaction of the pushdown system.

\paragraph{CFL- and pushdown-reachability techniques}
This work also draws on CFL- and pushdown-reachability
analysis~\cite{mattmight:Bouajjani:1997:PDA-Reachability,dvanhorn:Kodumal2004Set,mattmight:Reps:1998:CFL,mattmight:Reps:2005:Weighted-PDA}.
For instance, \ecg s, or equivalent variants thereof, appear in many
context-free-language and pushdown reachability algorithms.
For our analysis, we implicitly invoked these methods as subroutines.
When we found these algorithms lacking (as with their enumeration of
control states), we developed rooted pushdown system compaction.

CFL-reachability techniques have also been used to compute classical
finite-state abstraction CFAs~\cite{mattmight:Melski:2000:CFL} and
type-based polymorphic control-flow
analysis~\cite{mattmight:Rehof:2001:TypeBased}.
These analyses should not be confused with pushdown control-flow
analysis, which is computing a fundamentally more precise kind of CFA.
Moreover, Rehof and F\"ahndrich's method is cubic in the size of the
\emph{typed} program, but the types may be exponential in the size of
the program.
Finally, our technique is not restricted to typed programs.

\paragraph{Model-checking pushdown systems with checkpoints}
A pushdown system with checkpoints has designated finite automata for state/frame pairs.
If in a given state/frame configuration, and the automaton accepts the current stack, then execution continues.
This model was first created in \citet{EsparzaKS03} and describes its applications to model-checking programs that use Java's \texttt{AccessController} class, and performing better data-flow analysis of Lisp programs with dynamic scope, though the specific applications are not fully explored.
The algorithm described in the paper is similar to ours, but not ``on-the-fly,'' however, so such applications would be difficult to realize with their methods.
The algorithm discussed has multiple loops that enumerate all transitions within the pushdown system considered.
Again this is a non-starter for higher-order languages, since up-front enumeration would conservatively suggest that any binding called would resolve to any possible function.
This strategy is a sure-fire way to destroy precision and performance.

\paragraph{Meet-over-all-paths for conditional weighted pushdown systems}
A conditional pushdown system is essentially a pushdown system in which every state/frame pair is a checkpoint.
The two are easily interchangeable, but weighted conditional pushdown systems assign weights to reduction rules from a bounded idempotent semiring in the same manner as \citet{mattmight:Reps:2005:Weighted-PDA}.
The work that introduces CWPDSs uses them for points-to analysis for Java.
They solve the meet-over-all-paths problem by an incrementally translating a skeleton CFG into a WPDS and using WPDS++~\citep{ianjohnson:DBLP:conf/cav/LalR06} to discover more points-to information to fill in call/return edges.
The translation involves a heavy encoding and is not obviously correct.
The killer for its use for GC is that it involves building the product automaton of all the (minimized) condition automata for the system, and interleaving the system states with the automaton's states --- there are exponentially many such machines in our case, and even though the overall solution is incremental, this large automaton is pre-built.
It is not obvious how to incrementalize the whole construction, nor is it obvious that the precision and performance are not negatively impacted by the repeated invocation of the WPDS solver (as opposed to a work-set solution that only considers recently changed states).
The approach to incremental solving using first-order tools is an interesting approach that we had not considered.
Perhaps first-order and higher-order methods are not too far removed.
It is possible that these frameworks could be extended to request transitions --- or even further, checkpoint machines --- on demand in order to better support higher-order languages.
As we saw in this article, however, we needed access to internal data structures to compute root sets of addresses, and the ability to update a cache of such sets in these structures.
The marriage could be rocky, but worth exploring in order to unite the two communities and share technologies.

\paragraph{Model-checking higher-order recursion schemes}
There is terminology overlap with work by
\citet{mattmight:Kobayashi:2009:HORS} on model-checking higher-order
programs with higher-order recursion schemes, which are a
generalization of context-free grammars in which productions can take
higher-order arguments, so that an order-0 scheme is a context-free
grammar.
Kobyashi exploits a result by \citet{dvanhorn:Ong2006ModelChecking} which
shows that model-checking these recursion schemes is decidable (but
ELEMENTARY-complete) by transforming higher-order programs into
higher-order recursion schemes.

Given the generality of model-checking, Kobayashi's technique may be
considered an alternate paradigm for the analysis of
higher-order programs.
For the case of order-0, both Kobayashi's technique and our own
involve context-free languages, though ours is for control-flow
analysis and his is for model-checking with respect to a temporal
logic.
After these surface similarities, the techniques diverge.
In particular, higher-order recursions schemes are limited
to model-checking programs in the simply-typed 
lambda-calculus with recursion.







\section{Conclusion}

Our motivation was to further probe the limits of decidability
for pushdown flow analysis of higher-order programs
by enriching it with abstract garbage collection.
We found that abstract garbage collection broke
the pushdown model, but not irreparably so.
By casting abstract garbage collection in terms of
an introspective pushdown system and synthesizing
a new control-state reachability algorithm, we have 
demonstrated the decidability of fusing two
powerful analytic techniques.

As a byproduct of our formulation, it was also easy
to demonstrate how polyvariant/context-sensitive 
flow analyses generalize to a pushdown formulation,
and we lifted the need to transform to continuation-passing style
in order to perform pushdown analysis.

Our empirical evaluation is highly encouraging: it shows that the fused
analysis provides further large reductions in the size of the abstract
transition graph---a key metric for interprocedural control-flow precision.
And, in terms of singleton flow sets---a heuristic metric for optimizability---the fused 
analysis proves to be a ``better-than-both-worlds'' combination.

Thus, we provide a sound, precise and polyvariant introspective
pushdown analysis for higher-order programs.

%

\section*{Acknowledgments}
We thank the anonymous reviewers of ICFP 2012 and JFP for their
detailed reviews, which helped to improve the presentation and
technical content of the paper.
Tim Smith was especially helpful with his knowledge of stack automata.
This material is based on research sponsored by DARPA under the programs 
Automated Program Analysis for Cybersecurity (FA8750-12-2-0106) and 
Clean-Slate Resilient Adaptive Hosts (CRASH). 
The U.S. Government is authorized to reproduce and
distribute reprints for Governmental purposes notwithstanding any copyright
notation thereon.





\bibliographystyle{chicago}


\bibliography{mattmight,dvanhorn,ianj,local}

\appendix

\section{Full Proofs}
\label{proofs}

\subsection{Pushdown reachability}

\noindent{}Proof of \autoref{thm:eps-closure-correct}.
The space $\mathit{ICRPDS}$ is further constrained than stated in the main article:
\begin{equation*}
  \mathit{ICRPDS} = \setbuild{
((S,E),H,(\Delta S, \Delta E, \Delta H))
}{
\begin{array}{l}
\bigcup\setbuild{\set{\dsstate,\dsstate'}}{\biedge{\dsstate}{\dsstate'}\in H} \subseteq S\text,\\
\Delta S \cap S = \emptyset \text{, } \Delta E \cap E = \emptyset\text{, and } \Delta H \cap H = \emptyset
\end{array}}
\end{equation*}

For this section we assume
\begin{align*}
  M &= (\QStates, \StackAlpha, \transfunction, \qstate_0) \in \mathbb{RPDS} \\
  G &= ((S, E), H, (\Delta S, \Delta E), \Delta H) \in \mathbb{CRPDS} \text{ where } (S,E) \subseteq (\QStates, \transfunction) \\
  \text{ and } &\qstate_0 = \dsstate_0
\end{align*}

Let $\mathbb{ORD}$ be the class of ordinal numbers.
We define a termination measure on the fixed-point computation of $\mkCRPDS'((\QStates,\StackAlpha,\_,\_))$, $d : \mathbb{CRPDS} \to \mathbb{ORD}$.
\begin{equation*}
  d((S,E),H, (\Delta S, \Delta E, \Delta H)) = (2^{|\QStates|^2\cdot|\StackAlpha|} - |E|)\omega + (2^{|\QStates|^2} - |H|)
\end{equation*}
\begin{lemma}[Termination]\label{lem:termination}
 Either $G = \mkCRPDS'(M)(G)$ or $d(\mkCRPDS'(M)(G)) \prec d(G)$
\end{lemma}
\begin{proof}
  If both $\Delta E$ and $\Delta H$ are empty, there are no additions made to $S$, $E$ or $H$, meaning $G$ is a fixed point.
  Otherwise, due to the non-overlap condition, one or both of $E$ and $H$ grow, meaning the ordinal is smaller.
\end{proof}
A corollary is that the fixed-point has empty $\Delta E$ and $\Delta H$.

\begin{lemma}[Key lemma for PDS reachability]\label{lem:mkcrpdsp-inv}
  If $\invcrpdsp(G)$ then $\invcrpdsp(\mkCRPDS'(M)(G))$
\end{lemma}
\begin{proof}
  All additional states and edges come from $\Delta E_i$ and $\Delta H_i$ for $i \in [0..4]$,
  so by cases on the sources of edges:
  \begin{byCases}
    \case{\triedge{\dsstate}{\stackact}{\dsstate'} \in \Delta E_0, \biedge{\dsstate''}{\dsstate'''} \in \Delta H_0}{
      By definition of $\fsprout$ and path extension.}
    \case{\triedge{\dsstate}{\stackact}{\dsstate'} \in \Delta E_1, \biedge{\dsstate''}{\dsstate'''} \in \Delta H_1}{
      If $\stackact \equiv \aphrame_-$, then by definition of $\faddpush$ there are
      $\triedge{\qstate}{\aphrame_+}{\qstate'} \in \Delta E$,
      $\biedge{\qstate'}{\dsstate} \in H$,
      such that $(\dsstate, \aphrame_+, \dsstate') \in \transfunction$.

      Let $\vec{\stackact}$ be the witness of the invariant on $\biedge{\qstate'}{\dsstate}$ given from definition of $\invcrpdsp$.
      Let $\acont$ be arbitrary.
      We have $[\aphrame_+\vec{\stackact}\aphrame_-] = \epsilon$.
      We also have $(\qstate,\acont) \overset{\aphrame_+\vec{\stackact}\aphrame_-}{\underset{M}{\longmapsto^*}} (\dsstate', \acont)$.
      Root reachability follows from path concatenation with the root path from $(\qstate,\acont) \PDTransOU{\aphrame_+}{M} (\qstate',\aphrame\acont)$ from $\invcrpdsp$.

      The balanced path for $\biedge{\dsstate''}{\dsstate'''}$ comes from a similar push edge from $\Delta E$ and concatenation with the path from the invariant on $H$.}
    \case{\biedge{\dsstate''}{\dsstate'''} \in \Delta H_2}{
      By definition of $\faddpop$, $\Delta E_2 = \emptyset$ and there are
      $\triedge{\qstate}{\aphrame_-}{\dsstate'''} \in \Delta E$,
      $\biedge{\qstate'}{\qstate} \in H$
      such that $\triedge{\dsstate''}{\aphrame_+}{\qstate'} \in E$.
      Let $\vec{\stackact}$ be the witness of the invariant on $\biedge{\qstate'}{\qstate}$.
      Let $\acont$ be arbitrary.
      We know by the invariant on $E$, $(\dsstate'',\acont) \overset{\aphrame_+\vec{\stackact}\aphrame_-}{\underset{M}{\longmapsto^*}} (\dsstate''', \acont)$
      and $[\aphrame_+\vec{\stackact}\aphrame_-] = \epsilon$.}
    \case{\triedge{\dsstate}{\stackact}{\dsstate'} \in \Delta E_3\cup\Delta E_4, \biedge{\dsstate''}{\dsstate'''} \in \Delta H_3 \cup \Delta H_4}{
      Follows from definition of $\invcrpdsp$ and path concatenation, following similar reasoning as above cases.}
  \end{byCases}
\end{proof}

We define ``$\mtrace$ is a subtrace of $\mtrace'$,'' $\mtrace \sqsubseteq \mtrace'$
\begin{align*}
  \infer{ }{(\dsstate',\acont') \PDTranssOU{\vect{}}{M} (\dsstate',\acont') \sqsubseteq (\dsstate,\acont) \PDTranssOU{\vec{\stackact}}{M} (\dsstate',\acont')}
  \quad
  \infer{\mtrace \sqsubseteq (\dsstate,\acont) \PDTranssOU{\vec{\stackact}}{M} (\dsstate',\acont')
         \quad (\dsstate',\stackact',\dsstate'') \in \transfunction}
        {\mtrace \sqsubseteq
         (\dsstate,\acont) \PDTranssOU{\vec{\stackact}}{M} (\dsstate',\acont') \PDTransOU{\stackact'}{M} (\dsstate'',[\acont'_+\stackact'])}
  \\
  \infer{(\dsstate,\acont) \PDTranssOU{\vec{\stackact}}{M} (\dsstate',\acont') \sqsubseteq
         (\dsstate''',\acont'') \PDTranssOU{\vec{\stackact''}}{M} (\dsstate',\acont')
         \quad (\dsstate',\stackact',\dsstate'') \in \transfunction}
        {(\dsstate,\acont) \PDTranssOU{\vec{\stackact}}{M} (\dsstate',\acont') \PDTransOU{\stackact'}{M} (\dsstate'',[\acont'_+\stackact'])
          \sqsubseteq
         (\dsstate''',\acont'') \PDTranssOU{\vec{\stackact''}}{M} (\dsstate',\acont') \PDTransOU{\stackact'}{M} (\dsstate'',[\acont'_+\stackact'])}
\end{align*}

\autoref{thm:eps-closure-correct} is a corollary of the following theorem.
\begin{theorem}
  $\lfp(\mkCRPDS'(M)) = (\fCRPDS(M),\fECG(M),(\emptyset,\emptyset),\emptyset)$
\end{theorem}
\begin{proof}
($\subseteq$): Directly from \autoref{lem:mkcrpdsp-inv}.

\noindent($\supseteq$): Let $\mtrace \equiv (\dsstate_0,\vect{}) \PDTranssOU{\vec{\stackact}}{M} (\dsstate,\acont)$ be an arbitrary path in $\fCRPDS(M)$ (the inclusion of root is not a restriction due to the definition of CRPDSs).
Let $n \in Nats$ be such that $\lfp(\mkCRPDS'(M)) = \mkCRPDS'(M)^n$.
We show
\begin{itemize}
\item{the same path through $G$,}
\item{for each $\dsstate \in S$, $\triedge{\dsstate}{\stackact}{\dsstate'} \in E$, $\biedge{\dsstate}{\dsstate'} \in H$,
      there is an $m < n$ such that
      $\dsstate \in \Delta S_m$
      $\triedge{\dsstate}{\stackact}{\dsstate'} \in \Delta E_m$
      $\biedge{\dsstate}{\dsstate'} \in \Delta H_m$ respectively,
      where $\mkCRPDS'(M)^m = (G_m, H_m, (\Delta S_m, \Delta E_m, \Delta H_m))$, and}
\item{all non-empty balanced subtraces have edges in $H$:
      $\forall (\dsstate_b,\acont) \PDTranssOU{\vec{\stackact_\epsilon}}{M} (\dsstate_a,\acont)
                 \sqsubseteq \mtrace.
               \vec{\stackact_\epsilon} \neq \vect{} \wedge [\vec{\stackact_\epsilon}] = \epsilon \implies \biedge{\dsstate_b}{\dsstate_a} \in H$.}
\end{itemize}

By induction on $\mtrace$,
\begin{byCases}
  \case{\text{Base: } \dsstate_0}{Follows by definition of $\mkCRPDS'$. No non-empty balanced subtrace.}
  \case{\text{Induction step: }
         (\dsstate_0,\vect{}) \PDTranssOU{\vec{\stackact'}}{M} (\dsstate',\acont)
         \PDTransOU{\stackact''}{M} (\dsstate,[\acont_+\stackact''])}{
    By IH, $(\dsstate_0,\vect{}) \PDTranssOU{\vec{\stackact'}}{G} (\dsstate', \acont)$.
    By cases on $\stackact''$:
    \begin{byCases}
      \case{\stackchar_+}{
        Let $m$ be the witness for $\dsstate'$ by the IH.
        By definition of $\mkCRPDS'$, $(\dsstate',\acont) \PDTransOU{\stackact''}{M} (\dsstate,[\acont_+\stackact''])$ is in $\Delta E_{m+1}$ and $E_{m+2}$ (and thus $\dsstate \in \Delta S_{m+1}$ and $S_{m+2}$).
        Thus the path is constructible through $G_n$.
        All balanced subtraces carry over from IH, since the last push edge cannot end a balanced path.}
      \case{\epsilon}{
        The path is constructible the same as $\stackchar_+$. Let $m$ be the witness used in the path construction.
        Let $\mtrace' \equiv (\dsstate_b,\acont) \PDTranssOU{\vec{\stackact_\epsilon}}{M} (\dsstate_e,\acont)$ be an arbitrary non-empty balanced subtrace. If $\dsstate_e \neq \dsstate$, then the IH handles it.
        Otherwise, $\vec{\stackact_\epsilon} = \vec{\stackact'_\epsilon}\epsilon$.
        If $\dsstate_b = \dsstate'$, then the $\epsilon$-edge is added by $\fsprout$ (so the witness number is $m + 1$).
        If not, then there is a balanced subtrace $(\dsstate_b,\acont) \PDTranssOU{\vec{\stackact'_\epsilon}}{M} (\dsstate',\acont)$, thus $\biedge{\dsstate_b}{\dsstate'} \in H$.
        Let $m'$ be the witness for $\biedge{\dsstate_b}{\dsstate'} \in \Delta H_{m'}$. Then $\biedge{\dsstate_b}{\dsstate} \in \Delta_{\max\set{m,m'}+1}$ by definition of $\faddempty$.}
      \case{\stackchar_-}{
        Since $[\vec{\stackact}]$ is defined, there is a push edge in the trace (call it $\dsstate_u \PDTransOU{\stackchar_+}{M} \dsstate_v$) with a (possibly empty) balanced subtrace following to $\dsstate'$.
        Thus by IH, there are some $m,m'$ such that $\triedge{\dsstate_u}{\stackchar_+}{\dsstate_v} \in E_m$, (if the subtrace is non-empty) $\biedge{\dsstate_v}{\dsstate'} \in H_{m'}$
        If $m \ge m'$ by definition of $\faddpush$, $\triedge{\dsstate'}{\stackchar_-}{\dsstate} \in \Delta E_{m+1}$.
        Otherwise, the edge is in $E_{m'}$ and by definition of $\faddempty$, $\triedge{\dsstate'}{\stackchar_-}{\dsstate} \in \Delta E_{m'+1}$.

        Let $\mtrace' \equiv (\dsstate_b,\acont) \PDTranssOU{\vec{\stackact_\epsilon}}{M} (\dsstate_e,\acont)$ be an arbitrary non-empty balanced subtrace.
        If $\dsstate_e \neq \dsstate$, the IH handles it.
        Otherwise, $\vec{\stackact_\epsilon} = \vec{\stackact'_\epsilon}\stackchar_+\vec{\stackact''_\epsilon}\stackchar_-$ and
        $\mtrace' \equiv (\dsstate_b,\acont) \PDTranssOU{\vec{\stackact'_\epsilon}}{M}
                         (\dsstate_u,\acont) \PDTransOU{\stackchar_+}{M}
                         (\dsstate_v,\stackchar\acont) \PDTranssOU{\vec{\stackact'_\epsilon}}{M}
                         (\dsstate',\stackchar\acont) \PDTransOU{\stackchar_-}{M} (\dsstate,\acont)$.
        $\biedge{\dsstate_u}{\dsstate}$ is added to $\Delta H_{\max\set{m,m'} +1}$ and thus $\biedge{\dsstate_b}{\dsstate_u}$ is in $H_{\max\set{m,m'}+3}$.}
    \end{byCases}}
\end{byCases}
\end{proof}
\subsection{RIPDS reachability}
\label{sec:ripds-reach}

\newcommand{\hastail}{\mathit{hastail}}
\newcommand{\Trace}{\mathit{Trace}}
\newcommand{\Cont}{\mathit{Cont}}
\newcommand{\contappend}[2]{#1 {\tt ++} #2}
We use metafunction $\contappend{\bullet}{\bullet} : \Cont \times \Cont \to \Cont$ to aid proofs:
\begin{align*}
  \contappend{\epsilon}{\acont} &= \acont \\
  \contappend{\phrame:\acont}{\acont'} &= \phrame:(\contappend{\acont}{\acont'})
\end{align*}


\newcommand{\ksplit}{\mathit{split}}
\begin{align*}
  \ksplit(\epsilon) &= [\text\textcent,\epsilon] \\
  \ksplit(\phrame:\acont) &= [\text\textcent\acont,\phrame]
\end{align*}

\begin{lemma}[Down spin]\label{lem:downspin}
  For $(\qstate, \epsilon, \epsilon, \downarrow, \epsilon, \qstate) \in \transfunction$,
  $(\qstate, [\contappend{\acont_B}{\acont_{B'}},\acont_T], w) \longmapsto^* (\qstate, [\acont_B,\contappend{\acont_{B'}}{\acont_T}], w)$
\end{lemma}
\begin{proof}
  By induction on $\acont_{B'}$.
  \begin{byCases}
    \case{\text{Base: } \acont_B = \epsilon}{Reflexivity.}
    \case{\text{Induction step: } \acont_{B'} = \acont\phrame}{
      By $\transfunction$, $(\qstate, [\contappend{\acont_B}{\acont_{B'}}, \acont_T], w) \longmapsto (\qstate, [\contappend{\acont_B}{\acont},\phrame\acont_T], w)$.
      By IH, $(\qstate, [\contappend{\acont_B}{\acont},
      \phrame\acont_T], w) $
      
      $\longmapsto^* (\qstate, [\acont_B, \contappend{\acont}{\phrame\acont_T}], w)$.
      This final configuration is the same as $(\qstate,[\acont_B,\contappend{\acont_{B'}}{\acont_T}],w)$.}
  \end{byCases}
\end{proof}

\begin{lemma}[$\gadget$ correctness]\label{lem:gadget}
  For $(\transfunction, \DSStates) = \gadget(\dsstate, \hat K, \stackact, \dsstate')$, $(\dsstate, \ksplit(\acont), w) \underset{\transfunction}{\longmapsto^*} (\dsstate', \ksplit([\acont_+\stackact]), w)$ iff $\acont \in \hat K$ and $[\acont_+\stackact]$ defined.
\end{lemma}
\begin{proof}
  ($\Rightarrow$):
  By inversion on the rules for $\transfunction$, the path must go through three stages: the down-spin, the middle path, and the pop-off.
  By \autoref{lem:downspin}, $(\dsstate, \ksplit(\acont), w) \longmapsto (\qstate_{\text{down}}, [\text\textcent\acont,\$], w) \longmapsto^* (\qstate_{\text{down}}, [\epsilon,\text\textcent\acont\$], w)$.
  Then the $(\qstate_{\text{down}}, \text\textcent, \epsilon, \uparrow,\epsilon,\qstate_0)$ rule must apply.
  We can construct an accepting path in the machine recognizing $\hat K$ from the middle path via the following lemma:
  $(\qstate_0, [\text\textcent, \acont\$], w) \overset{\acont'}{\underset{\transfunction}{\longmapsto^*}} (\qstate, [\text\textcent\acont',\acont''\$], w)$ implies $(\qstate_0, \acont) \overset{\acont'}{\underset{N}{\longmapsto^*}} (\qstate, \acont'')$.
  Proof by induction.
  
  Then $(\qstate,\$,\epsilon,\cdot,\$_-,\qstate_{\text{out}})$ must apply, and then $(\qstate_{\text{out}}, \epsilon,\epsilon,\stackact,\dsstate')$ must apply, meaning that $[\acont_+\stackact]$ is defined.
  \\

  ($\Leftarrow$):
  Since $\hat K$ is regular, there must be a path in the chosen NFA $N = (Q, \Sigma, \transfunction_N, \qstate_0, F)$ from $\qstate_0$ to a final state $\qstate \in F$, $(\qstate_0,\acont) \underset{N}{\longmapsto^*} (\qstate,\epsilon)$.
  \\

\noindent  In the first stage, $(\dsstate, \ksplit(\acont), w) \longmapsto^* (\qstate_0, [\text\textcent,\acont\$], w)$.

\noindent  The follows first by the $(\dsstate,\epsilon,\epsilon,\cdot,\$_+,\qstate_{\text{down}})$ transition, then by \autoref{lem:downspin} $(\qstate_{\text{down}}, \ksplit(\acont\$), w)$ $\longmapsto^* (\qstate_{\text{down}}, [\epsilon, \text\textcent\acont\$], w)$, finally by the $(\qstate_{\text{down}}, \text\textcent,\epsilon,\uparrow,\epsilon,\qstate_0)$ rule.
  \\

\noindent  In the second stage we construct a path $(\qstate_0,[\text\textcent,\acont\$],w) \longmapsto^* (\qstate, [\text\textcent\acont,\$], w)$, from an accepting path in $N$: $(\qstate_0,\acont) \longmapsto^* (\qstate,\epsilon)$ where $\qstate \in F$.
  The statement we can induct on to get this is $(\qstate_0, \acont) \overset{\acont'}{\underset{N}\longmapsto^*} (\qstate,\acont'')$ implies $(\qstate_0, [\text\textcent,\acont\$], w) \overset{\acont'}{\underset{\transfunction}{\longmapsto^*}} (\qstate, [\text\textcent\acont',\acont''\$], w)$.
  \begin{byCases}
    \case{\text{Base: } \acont' = \epsilon, \qstate = \qstate_0, \acont'' = \acont}{Reflexivity.}
    \case{\text{Induction step: } \acont' = \acont'''\phrame_\epsilon,
          (\qstate_0, \acont) \overset{\acont'''}{\underset{N}{\longmapsto^*}} (\qstate',\acont'''')
                             \overset{\phrame_\epsilon}{\underset{N}{\longmapsto}} (\qstate, \acont'')}{
      By IH, $(\qstate_0, [\text\textcent,\acont\$], w) \longmapsto^* (\qstate',[\text\textcent\acont''',\acont''''\$],w)$.
      If $\phrame_\epsilon = \epsilon$, then $\acont''' = \acont'$, $\acont'''' = \acont''$ and we apply the $(\qstate', \epsilon, \epsilon,\cdot,\epsilon,\qstate)$ rule to get to $(\qstate, [\text\textcent\acont',\acont''\$],w)$.
      Otherwise, $\acont' = \acont'''\phrame$ and we apply the $(\qstate',\phrame,\epsilon,\uparrow,\epsilon,\qstate)$ rule to get to $(\qstate, [\text\textcent\acont',\acont''\$], w)$.}
  \end{byCases}

\noindent  In the third and final stage, $(\qstate_{\text{out}}, [\text\textcent\acont,\$], w) \longmapsto (\qstate_{\text{out}}, \ksplit(\acont), w)$ and since $[\acont_+\stackact]$ is defined, we reach the final state by $(\qstate_{\text{out}}, \epsilon,\epsilon,\cdot,\stackact,\dsstate')$.
\end{proof}

\begin{lemma}[Checking lemma]\label{lem:check}
  If $(\qstate,a,a,\uparrow,\epsilon,\qstate) \in \transfunction$ and
  $(\qstate,[\acont_B,\contappend{\acont_{T'}}{\acont_{T}}\$],w) \longmapsto^* (\qstate, [\contappend{\acont_B}{\acont_{T'}},\acont_T\$], w')$ (through the one rule)
  then $w = \acont_{T'}w'$.
\end{lemma}
\begin{proof}
  Simple induction.
\end{proof}

\begin{lemma}[Stack machine correctness]\label{lem:stack-correct}
  For all $M \in \mathbb{CPDS}$, $G \in \mathbb{CCPDS}$, $\qstate \in G$,
  if $G \sqsubseteq \fCCPDS(M)$ then $$\mathcal{L}(\Stacks(G)(\qstate)) = \setbuild{\acont}{(\qstate_0,\vect{}) \overset{\overrightarrow{\hat K, \stackact}}{\underset{G}{\longmapsto^*}} (\qstate,\acont)}.$$
\end{lemma}
\begin{proof}
  ($\subseteq$):
  Let $(\dsstate_{\text{start}}, [\epsilon, \epsilon], \acont) \longmapsto^* (\dsstate_{\text{final}}, \ksplit(\acont), \epsilon)$ be an accepting path for $\acont \in \mathcal{L}(\Stacks(G)(\qstate))$.
  We inductively construct a corresponding path in $G$ that realizes $\acont$.
  We first see that the given path is split into three phases: setup, gadgetry, checking.
  The first step must be $(\dsstate_{\text{start}}, \epsilon, \epsilon, \cdot, \text\textcent_+, \dsstate_0)$, which we call setup. The only final state must be preceded by $\dsstate_{\text{check}}$, $\dsstate_{\text{down}}$, and the final occurrence of $\dsstate$, which we call checking. Thus the middle phase is a trace from $\dsstate_0$ to $\dsstate$. This must be through gadgets, which are disjoint for each rule of the IPDS, and thus each edge in $G$. \\
$(\dsstate_{\text{start}}, [\epsilon,\epsilon], \acont) \longmapsto (\dsstate_0, [\epsilon,\text\textcent],\acont) \longmapsto^* (\dsstate, \ksplit(\acont), \acont) \longmapsto$ \\
  $(\dsstate_{\text{down}}, [\text\textcent\acont,\$], \acont) \longmapsto^* (\dsstate_{\text{down}}, [\epsilon,\text\textcent\acont\$], \acont) \longmapsto$ \\
   $(\dsstate_{\text{check}}, [\text\textcent,\acont\$], \acont) \longmapsto^* (\dsstate_{\text{check}}, [\text\textcent\acont,\$], \epsilon) \longmapsto (\dsstate_{\text{final}}, \ksplit(\acont), \epsilon)$ \\
  We induct on the path through gadgets, $\dsstate_0$ to $\dsstate$ in the above path, invoking \autoref{lem:gadget} at each step.
 \\

\noindent($\supseteq$): Simple induction between setup and teardown, applying \autoref{lem:gadget}.
\end{proof}

\noindent{}Proof of \autoref{thm:ripds-to-icrpds}
\begin{proof}
  The finiteness of the state space and monotonicity of $\mkCCPDS$ ensures the least fixed point exists.
  $\lfp(\mkCCPDS(M)) \subseteq \fCCPDS(M)$ follows from \autoref{lem:stack-correct} and the definition of $\mkCCPDS$.
  To prove $\fCCPDS(M) \subseteq \lfp(\mkCCPDS(M))$, suppose not. Then there must be a path $(\dsstate_0,\vect{}) \PDTranssOU{\overrightarrow{\hat K, \stackact}}{M} (\dsstate, \acont) \PDTransOU{\hat K', \stackact'}{M} (\dsstate', [\acont_+\stackact'])$ where the final edge is the first edge not in $\lfp(\mkCCPDS(M))$.
  
  By definition of $\fCCPDS$, $\acont \in \hat K'$ and $(\dsstate,\hat K', \stackact, \dsstate') \in \transfunction$.
  Since $\acont$ is realizable at $\dsstate$ in $G$, by definition of $\mkCCPDS$ and \autoref{lem:stack-correct}, $(\dsstate, \acont) \overset{\hat K',\stackact}{\underset{G}{\longmapsto}} (\dsstate',[\acont_+\stackact])$ contra the assumption. Thus $\fCCPDS(M) \subseteq \lfp(\mkCCPDS(M))$ holds by contradiction.
\end{proof}
\newcommand{\invgamma}{\mathit{inv}_\Gamma}
\newcommand{\invagamma}{\mathit{inv}_{\hat\Gamma}}
\newcommand{\injgamma}{{\mathcal I}_\Gamma}
\newcommand{\sysgamma}{\sa{System}_\Gamma}
We first prove an invariant of $\atf : \syn{Exp} \to \sysgamma \monto \sysgamma$, where $\sysgamma = \mathbb{ICRPDS} \times \PowSm{\sa{OPState} \times \sa{OPState}}$
\begin{align*}
  \injgamma(\expr) &= ((\expr,\bot,\bot),\vect{}) \\
  \injgamma'(\expr) &= (((\expr,\bot,\bot),\emptyset), \vect{})
  \\[2pt]
  \vect{\aphrame_1, \ldots, \aphrame_n}_+ &= \aphrame_1{}_+\ldots\aphrame_n{}_+
  \\[2pt]
  \invgamma &: \syn{Exp} \to \sysgamma \to \Prop \\
  \invgamma(\expr)(\overbrace{(\hat P, \hat E), \hat H}^G) &= (\hat P = \bigcup\setbuild{\set{\aopstate,\aopstate'}}{\triedge{\aopstate}{\stackact}{\aopstate'} \in \hat E}) \\
                                                        &\wedge \forall \triedge{(\apstate,A)}{\stackact}{(\apstate',A')} \in \hat E.
                                                        \text{let } \hat K = \setbuild{\acont}{\StackRoot(\acont) = A} \text{ in } \\
                                                        &\phantom{\text{let }}
                                                          \forall \acont \in \setbuild{\acont \in \hat K}{[\acont_+\stackact] \text{ defined}}. \StackRoot([\acont_+\stackact]) = A'
                                                          \\ & ~~~~~~~~~\wedge (\apstate,\acont) \PDTransOU{\hat K,\stackact}{M} (\apstate',[\acont_+\stackact])\\
                                                        &\wedge \forall \biedge{(\apstate,\_)}{(\apstate',\_)} \in \hat H. \exists \overrightarrow{\hat K,\stackact}.
                                                          [\vec{\stackact}] = \epsilon \wedge
                                                          (\apstate,\vect{}) \PDTranssOU{\overrightarrow{\hat K,\stackact}}{M} (\apstate',\vect{}) \\
   \text{where } M &= \afIPDS'(\expr)
\end{align*}

\begin{lemma}[$\atf$ invariant]\label{lem:atf-inv}
  For all $\expr$, if $\invgamma(\expr)(G)$ then $\invgamma(\expr)(\atf_\expr(G))$
\end{lemma}
\begin{proof}
  Same structure as in Lemma~\autoref{lem:mkcrpdsp-inv} without reasoning about worklists.
\end{proof}

\noindent{}Proof of~\autoref{thm:gc-specialization}
\begin{proof}
  Let $M = \afIPDS'(\expr)$, $G = \fCCPDS(M)$ and $G' = ((\hat P, \hat E), \hat H) = \lfp(\atf_\expr)$.
  \\
  ($G'$ approximates $G$): \\
  We strengthen the statement to
  $\mtrace \equiv \injgamma(\expr) \overset{\overrightarrow{\hat K, \stackact}}{\underset{G}{\longmapsto^*}} (\apstate,\acont)$
  implies
  \begin{itemize}
  \item {$\injgamma'(\expr) \overset{\vec{\stackact}}{\underset{G'}{\longmapsto^*}} ((\apstate,\StackRoot(\acont)), \acont)$.}
  \item {for all $(\apstate,\acont) \PDTranssOU{\overrightarrow{\hat K, \stackact}}{G} (\apstate',[\acont_+\vec{\stackact}]) \PDTransOU{\hat K', \stackact'}{G} (\apstate'',[\acont_+\vec{\stackact}\stackact'])
                   \sqsubseteq \mtrace$,
         if $[\vec{\stackact}\stackact'] = \epsilon$, then $\exists \acont \in \hat K$ and $\biedge{(\apstate, \StackRoot(\acont))}{(\apstate'', \StackRoot([\acont_+\vec{\stackact}\stackact']))} \in \hat H$}
  \end{itemize}
  
  By induction on $\mtrace$,
  \begin{byCases}
    \case{\text{Base: } \injgamma(\expr)}{
      By definition of $\atf_\expr$, $\injgamma'(\expr) = (\aopstate_0,\vect{})$, $\aopstate_0 \in \hat P$.
      First goal holds by definition of $\StackRoot(\vect{})$ and reflexivity. Second goal vacuously true.}
    \case{\text{Induction step: }
          ((\expr,\bot,\bot),\vect{}) \overset{\overrightarrow{\hat K',\stackact'}}{\underset{G}{\longmapsto^*}}
          (\apstate',[\vec{\stackact'}]) \overset{\hat K'',\stackact''}{\underset{G}{\longmapsto}}
          (\apstate, \acont)}{
      Let $A = \StackRoot([\vec{\stackact'}])$.
      By IH, $\injgamma'(\expr) \overset{\vec{\stackact'}}{\underset{G'}{\longmapsto^*}} ((\apstate', A), [\vec{\stackact'}])$.
      
      Let $\hat K_{\text{root}} = \setbuild{\acont}{\StackRoot(\acont) = A}$
      By definition of $\afIPDS'$ and the case assumption, $(\apstate', \hat K_{\text{root}}, \stackact'', \apstate) \in \transfunction$.
      By cases on $(\apstate',[\vec{\stackact'}]) \overset{\hat K'',\stackact''}{\underset{G}{\longmapsto}} (\apstate, \acont)$:
      \begin{byCases}
        \case{(\apstate',[\vec{\stackact'}]) \overset{\hat K'',\aphrame_+}{\underset{G}{\longmapsto}} (\apstate, \acont)}{
          By definition of $\atf$, $\triedge{(\apstate', A)}{\aphrame_+}{(\apstate, A\cup\touches(\aphrame_+))} \in G'$.
          By definition of $\StackRoot$ and $A$, $\StackRoot([\vec{\stackact'}\stackact'']) = \StackRoot([\vec{\stackact}]) = \StackRoot(\acont)$.}
        \case{(\apstate',[\vec{\stackact'}]) \overset{\hat K'',\epsilon}{\underset{G}{\longmapsto}} (\apstate, \acont)}{Similar to previous case.}
        \case{(\apstate',[\vec{\stackact'}]) \overset{\hat K'',\aphrame_-}{\underset{G}{\longmapsto}} (\apstate, \acont)}{
          Since is $[\vec{\stackact'}\aphrame_-]$ is defined, there is an $i$ such that $\stackact_i = \aphrame_+$, which is witness to an edge in the trace with that action, $\apstate_b \RPDTransOU{\hat K''',\aphrame_+}{G} \apstate_e$
          By definition of $[\_]$, the actions from $\apstate_e$ to $\apstate'$ cancel to $\epsilon$, meaning by IH $\biedge{(\apstate_e,A)}{(\apstate',A)} \in H$, and $\triedge{(\apstate_b,A')}{\aphrame_+}{(\apstate_e, A)} \in E$.
          Thus the pop edge is added by definition of $\atf'$.
          The new balanced path $\biedge{(\apstate_b,A')}{(\apstate,A')}$ is also added, and extended paths get added with propagation.}
      \end{byCases}}
  \end{byCases}
  Approximation follows by composition with \autoref{thm:ripds-to-icrpds}. \\

  ($G$ approximates $G'$): Directly from \autoref{lem:atf-inv}.
\end{proof}

The approximate GC has a similar invariant, except the sets of addresses are with respect to the $\att$ computation.
\begin{align*}
  \invagamma &: \syn{Exp} \to \sysgamma \to \Prop \\
  \invagamma(\expr)(\overbrace{(\hat P, \hat E), \hat H}^G) &= (\hat P = \bigcup\setbuild{\set{\apstate,\apstate'}}{\triedge{\apstate}{A, \stackact}{\apstate'} \in \hat E}) \\
                                                        &\wedge \forall \triedge{\apstate_0}{A,\stackact}{\apstate_1} \in \hat E.
                                                          \exists \triedge{(\apstate^{\Gamma}_0, A_\Gamma)}{\stackact}{(\apstate^\Gamma_1, A'_\Gamma)} \in \hat E.
                                                          (\forall i. \apstate^\Gamma_i \sqsubseteq \apstate_i) \wedge A_\Gamma \subseteq A
                                                          \\
                                                        &\wedge \forall \biedge{\apstate_0}{\apstate_1} \in \hat H.
                                                          \exists \biedge{(\apstate_0^\Gamma,A_\Gamma)}{(\apstate_1^\Gamma, A_\Gamma)} \in \hat H'.
                                                          \forall i. \apstate^\Gamma_i \sqsubseteq \apstate_i
                                                          \\
                                                        &\wedge \forall \apstate \in \hat P. \exists (\apstate^\Gamma,A) \in \hat P'. \apstate^\Gamma \sqsubseteq \apstate \wedge \lfp(\att)(\apstate) \subseteq A
\\
   \text{where } ((\hat P', \hat E'), \hat H') &= \lfp(\atf_\expr)
\end{align*}

\begin{lemma}[Approx GC invariant]\label{lem:approx-gc-inv}
  For all $\expr$, if $\invagamma(\expr)(G)$ then $\invagamma(\expr)(\atf'_\expr(G))$
\end{lemma}
\begin{proof}
  Straightforward case analysis.
\end{proof}

\noindent{}Proof of \autoref{thm:approx-gc-sound}
\begin{proof}
  Induct on path in $\lfp(\atf_\expr)$ and apply \autoref{lem:approx-gc-inv}.
\end{proof}

\section{Haskell implementation of CRPDSs}
\label{haskell}
Where it is critical to understanding the details of the analysis,
we have transliterated the formalism into Haskell.
We make use of a two extensions in GHC: 
\begin{code}
 -XTypeOperators -XTypeSynonymInstances 
\end{code}
All code is in the context of the following header, and we'll assume
the standard instances of type classes like {\tt Ord} and {\tt Eq}.
\begin{code}
import Prelude hiding ((!!))

import Data.Map as Map hiding (map,foldr)
import Data.Set as Set hiding (map,foldr)
import Data.List as List hiding ((!!))

type \(\mathbb{P}\) s = Set.Set s
type k :-> v = Map k v

(==>) :: a -> b -> (a,b)
(==>) x y = (x,y)

(//) :: Ord a => (a :-> b) -> [(a,b)] -> (a :-> b)
(//) f [(x,y)] = Map.insert x y f

set x = Set.singleton x\end{code}

\subsection{Transliteration of NFA formalism}

We represent an NFA as a set of labeled forward edges, the inverse of
those edges (for convenience), a start state and an end state:
\begin{code}
type NFA state char =
 (NFAEdges state char,NFAEdges state char,state,state)
type NFAEdges state char = state :-> \(\mathbb{P}\)(Maybe char,state)
\end{code}

\subsection{ANF}

\begin{code}
data Exp    = Ret AExp
            | App Call
            | Let1 Var Call Exp
data AExp   = Ref Var
            | Lam Lambda
data Lambda = Var :=> Exp
data Call   = AExp :@ AExp
type Var    = String\end{code}

\begin{code}
-- Abstract state-space:
type AConf  = (Exp,AEnv,AStore,AKont)
type AEnv   = Var :-> AAddr
type AStore = AAddr :-> AD
type AD     = ℙ (AVal)
data AVal   = AClo (Lambda, AEnv)
type AKont  = [AFrame]
type AFrame = (Var,Exp,AEnv)

data AAddr = ABind Var AContext
type AContext = [Call]\end{code}
Abstract configuration space transliterated into Haskell.
In the code, we defined abstract addresses to be able to support 
$k$-CFA-style polyvariance.

Atomic expression evaluation implementation:
\begin{code}
aeval :: (AExp,AEnv,AStore) -> AD
aeval (Ref v, \(\rho\), \(\sigma\)) = \(\sigma\)!!(\(\rho\)!v)
aeval (Lam l, \(\rho\), \(\sigma\)) = set \$ AClo (l, \(\rho\))\end{code}

We encode the transition relation it as a function that returns lists of states:
\begin{code}
 astep :: AConf -> [AConf]\end{code}

\begin{code}
astep (App (f :@ ae), \(\rho\), \(\sigma\), \(\kappa\)) = [(e, \(\rho\)'', \(\sigma\)', \(\kappa\)) | 
   AClo(v :=> e, \(\rho\)') <- Set.toList \$ aeval(f, \(\rho\), \(\sigma\)),
   let a = aalloc(v, App (f :@ ae)),
   let \(\rho\)'' = \(\rho\)' // [v ==> a],
   let \(\sigma\)' = \(\sigma\) ⨆ [a ==> aeval(ae, \(\rho\), \(\sigma\))] ]
astep (Let1 v call e, \(\rho\), \(\sigma\), \(\kappa\)) =
  [(App call, \(\rho\), \(\sigma\), (v, e, \(\rho\)) : \(\kappa\))]
astep (Ret ae, \(\rho\), \(\sigma\), (v, e, \(\rho\)') : \(\kappa\)) = [(e, \(\rho\)'', \(\sigma\)', \(\kappa\))]
 where a = aalloc(v, Ret ae) 
       \(\rho\)'' = \(\rho\)' // [v ==> a]
       \(\sigma\)' = \(\sigma\) ⨆ [a ==> aeval(ae, \(\rho\), \(\sigma\))]
\end{code}

\subsection{Partial orders}
We define a typeclass for lattices:
\begin{code}
class Lattice a where
 bot :: a
 top :: a
 (\(\sqsubseteq\)) :: a -> a -> Bool
 (\(\sqcup\)) :: a -> a -> a
 (\(\sqcap\)) :: a -> a -> a
\end{code}

And, we can lift instances to sets and maps:

\begin{code}
instance (Ord s, Eq s) => Lattice (\(\mathbb{P}\) s) where
 bot = Set.empty
 top = error "no representation of universal set"
 x \(\sqcup\) y = x `Set.union` y
 x \(\sqcap\) y = x `Set.intersection` y
 x \(\sqsubseteq\) y = x `Set.isSubsetOf` y

instance (Ord k, Lattice v) => Lattice (k :-> v) where
 bot = Map.empty
 top = error "no representation of top map"
 f \(\sqsubseteq\) g = Map.isSubmapOfBy (\(\sqsubseteq\)) f g
 f \(\sqcup\) g = Map.unionWih (\(\sqcup\)) f g
 f \(\sqcap\) g = Map.intersectionWith (\(\sqcap\)) f g

(\(\sqcup\)) :: (Ord k, Lattice v) => (k :-> v) -> [(k,v)] -> (k :-> v)
f \(\sqcup\) [(k,v)] = Map.insertWith (\(\sqcup\)) k v f

(!!) :: (Ord k, Lattice v) => (k :-> v) -> k -> v
f !! k = Map.findWithDefault bot k f
\end{code}

\subsection{Reachability}
We can turn any data type to a stack-action alphabet:
\begin{code}
data StackAct frame = Push \{ frame :: frame \}
                    | Pop  \{ frame :: frame \}
                    | Unch \end{code}

\begin{code}
type CRPDS control frame = (Edges control frame, control) 
type Edges control frame = control :-> ℙ (StackAct frame,control)
\end{code}

We split the encoding of $\delta$ into two functions for efficiency purposes:
\begin{code}
type Delta control frame =
 (TopDelta control frame, NopDelta control frame)
type TopDelta control frame =
 control -> frame -> [(control,StackAct frame)]
type NopDelta control frame =
 control -> [(control,StackAct frame)]\end{code}
If we only want to know push and no-change transitions,
we can find these with a {\tt NopDelta} function without providing
the frame that is currently on top of the stack.
If we want pop transitions as well, we can find these with a {\tt TopDelta} 
function, but of course, it must have access to the top of the stack.
In practice, a {\tt TopDelta} function would suffice, but there are 
situations where only push and no-change transitions are needed, and having access to {\tt NopDelta} avoids
 extra computation.

At this point, we must clarify how to embed the abstract transition relation
into a pushdown transition relation:
\begin{code}
adelta :: TopDelta AControl AFrame
adelta (e, \(\rho\), \(\sigma\)) \(\gamma\) = [ ((e', \(\rho\)', \(\sigma\)'), g) | 
 (e', \(\rho\)', \(\sigma\)', \(\kappa\)) <- astep (e, \(\rho\), \(\sigma\), [\(\gamma\)]),
 let g = case \(\kappa\) of
          []        -> Pop \(\gamma\)
          [\(\gamma\)1 , \_ ] -> Push \(\gamma\)1
          [ \_ ]     -> Unch ]

adelta' :: NopDelta AControl AFrame
adelta' (e, \(\rho\), \(\sigma\)) = [ ((e', \(\rho\)', \(\sigma\)'), g) | 
 (e', \(\rho\)', \(\sigma\)', \(\kappa\)) <- astep (e, \(\rho\), \(\sigma\), []),
 let g = case \(\kappa\) of
          [\(\gamma\)1]  -> Push \(\gamma\)1
          [ ]   -> Unch ]
\end{code}

The function {\tt crpds} will invoke the fixed point solver:
\begin{code}
crpds :: (Ord control, Ord frame) =>
       (Delta control frame) ->
       control ->
       frame ->
       CRPDS control frame
crpds (\(\delta\),\(\delta\)') q0 γ0 =
 (summarize (\(\delta\),\(\delta\)') etg1 ecg1 [] dE dH, q0) where
 etg1 = (Map.empty // [q0 ==> Set.empty], 
         Map.empty // [q0 ==> Set.empty])
 ecg1 = (Map.empty // [q0 ==> set q0], 
         Map.empty // [q0 ==> set q0])
 (dE,dH) = sprout (\(\delta\),\(\delta\)') q0 \end{code}

Figure~\ref{fig:mkcrpds-ecg-code} provides the code for {\tt
  summarize}, which conducts the fixed point calculation, the
executable equivalent of Figure~\ref{fig:mkcompact-ecg}:
\begin{code}
summarize :: (Ord control, Ord frame) =>
             (Delta control frame) ->
             (ETG control frame) ->
             (ECG control) ->
             [control] ->
             [Edge control frame] ->
             [EpsEdge control] ->
             (Edges control frame)
\end{code}
To expose the structure of the computation, we've added a few types:
\begin{code}
-- A set of edges, encoded as a map:
type Edges control frame =
 control :-> \(\mathbb{P}\) (StackAct frame,control)

-- Epsilon edges:
type EpsEdge control = (control,control) 

-- Explicit transition graph:
type ETG control frame =
 (Edges control frame, Edges control frame)

-- Epsilon closure graph:
type ECG control =
 (control :-> \(\mathbb{P}\)(control), control :-> \(\mathbb{P}\)(control))
\end{code}

\begin{figure}
\begin{small}
\begin{code}
summarize (\(\delta\),\(\delta\)') (fw,bw) (fe,be) [] [] [] = fw

summarize (\(\delta\),\(\delta\)') (fw,bw) (fe,be) (q:dS) [] []
 | fe `contains` q = summarize (\(\delta\),\(\delta\)') (fw,bw) (fe,be) dS [] []
summarize (\(\delta\),\(\delta\)') (fw,bw) (fe,be) (q:dS) [] [] = 
 summarize (\(\delta\),\(\delta\)') (fw',bw') (fe',be') dS dE' dH' where
  (dE',dH') = sprout (\(\delta\),\(\delta\)') q
  fw' = fw \(\sqcup\) [q ==> Set.empty]
  bw' = bw \(\sqcup\) [q ==> Set.empty]
  fe' = fe \(\sqcup\) [q ==> set q] 
  be' = be \(\sqcup\) [q ==> set q]

summarize (\(\delta\),\(\delta\)') (fw,bw) (fe,be) dS ((q,g,q'):dE) []
 | (q,g,q') `isin'` fw  = summarize (\(\delta\),\(\delta\)') (fw,bw) (fe,be) dS dE []
summarize (\(\delta\),\(\delta\)') (fw,bw) (fe,be) dS ((q,Push γ,q'):dE) [] =
 summarize (\(\delta\),\(\delta\)') (fw',bw') (fe',be') dS' dE'' dH' where
  (dE',dH') = addPush (fw,bw) (fe,be) (\(\delta\),\(\delta\)') (q,Push γ,q')
  dE'' = dE' ++ dE''
  dS' = q':dS
  fw' = fw \(\sqcup\) [q  ==> set (Push γ,q')]
  bw' = bw \(\sqcup\) [q' ==> set (Push γ,q) ]
  fe' = fe \(\sqcup\) [q  ==> set q ]
  be' = fe \(\sqcup\) [q' ==> set q']

summarize (\(\delta\),\(\delta\)') (fw,bw) (fe,be) dS ((q,Pop γ,q'):dE) [] =
 summarize (\(\delta\),\(\delta\)') (fw',bw') (fe',be') dS' dE'' dH' where
  (dE',dH') = addPop (fw,bw) (fe,be) (\(\delta\),\(\delta\)') (q,Pop γ,q')
  dE'' = dE ++ dE'
  dS' = q':dS
  fw' = fw \(\sqcup\) [q  ==> set (Pop γ,q')]
  bw' = bw \(\sqcup\) [q' ==> set (Pop γ,q) ]
  fe' = fe \(\sqcup\) [q  ==> set q ]
  be' = fe \(\sqcup\) [q' ==> set q']

summarize (\(\delta\),\(\delta\)') (fw,bw) (fe,be) dS ((q,Unch,q'):dE) [] =
 summarize (\(\delta\),\(\delta\)') (fw',bw') (fe',be') dS' dE [(q,q')] where
  dS' = q':dS
  fw' = fw \(\sqcup\) [q  ==> set (Unch,q')]
  bw' = bw \(\sqcup\) [q' ==> set (Unch,q) ]
  fe' = fe \(\sqcup\) [q  ==> set q ]
  be' = fe \(\sqcup\) [q' ==> set q']

summarize (\(\delta\),\(\delta\)') (fw,bw) (fe,be) dS dE ((q,q'):dH)
 | (q,q') `isin` fe  = summarize (\(\delta\),\(\delta\)') (fw,bw) (fe,be) dS dE dH
summarize (\(\delta\),\(\delta\)') (fw,bw) (fe,be) dS dE ((q,q'):dH) = 
 summarize (\(\delta\),\(\delta\)') (fw,bw) (fe',be') dS dE' dH' where
   (dE',dH') = addEmpty (fw,bw) (fe,be) (\(\delta\),\(\delta\)') (q,q')
   fe' = fe \(\sqcup\) [q  ==> set q ]
   be' = fe \(\sqcup\) [q' ==> set q']\end{code}
\caption{An implementation of pushdown control-state reachability.}
\label{fig:mkcrpds-ecg-code}
\end{small}
\end{figure}

An explicit transition graph is an explicit encoding
of the reachable subset of the transition relation.
The function {\tt summarize} takes six parameters:
\begin{enumerate}
 \item the pushdown transition function;
 \item the current explicit transition graph;
 \item the current $\epsilon$-closure graph;
 \item a work-list of states to add;
 \item a work-list of explicit transition edges to add; and
 \item a work-list of $\epsilon$-closure transition edges to add.
\end{enumerate}
The function {\tt summarize} processes $\epsilon$-closure edges first,
then explicit transition edges and then individual states.
It \emph{must} process $\epsilon$-closure edges first to ensure that the
$\epsilon$-closure graph is closed when considering the implications of
other edges.

\paragraph{Sprouting}
\begin{code}
sprout :: (Ord control) =>
          Delta control frame ->
          control ->
          ([Edge control frame], [EpsEdge control])
sprout (\(\delta\),\(\delta\)') q = (dE, dH) where
  edges = \(\delta\)' q
  dE = [ (q,g,q') | (q',g) <- edges, isPush g ]
  dH = [ (q,q')   | (q',g) <- edges, isUnch g ]
\end{code}

\paragraph{Pushing}

\begin{code}
addPush :: (Ord control) =>
           ETG control frame ->
           ECG control ->
           Delta control frame ->
           Edge control frame ->
           ([Edge control frame], [EpsEdge control])
addPush (fw,bw) (fe,be) (\(\delta\),\(\delta\)') (s,Push \(\gamma\),q) = (dE,dH) where
 qset' = Set.toList \$ fe!q 
 dE = [ (q',g,q'') | q' <- qset', (q'',g) <- \(\delta\) q' \(\gamma\), isPop g ]
 dH = [ (s,q'')    | (q',Pop \_,q'') <- dE ] \end{code}

\paragraph{Popping}

\begin{code}
addPop :: (Ord control) =>
          ETG control frame ->
          ECG control ->
          Delta control frame ->
          Edge control frame ->
          ([Edge control frame], [EpsEdge control])
addPop (fw,bw) (fe,be) (\(\delta\),\(\delta\)') (s'',Pop \(\gamma\),q) = (dE,dH) where
 sset' = Set.toList \$ be!s''
 dH = [ (s,q) | s' <- sset', 
                (g,s) <- Set.toList \$ bw!s', isPush g ]
 dE = []\end{code}
Clearly, we could eliminate the new edges parameter {\tt dE} for 
the function {\tt addPop},
but we have retained it for stylistic symmetry.

\paragraph{Adding empty edges}

The function {\tt addEmpty} 
has many cases to consider:
\begin{code}
addEmpty :: (Ord control) =>
            ETG control frame ->
            ECG control ->
            Delta control frame ->
            EpsEdge control ->
            ([Edge control frame], [EpsEdge control])
addEmpty (fw,bw) (fe,be) (\(\delta\),\(\delta\)') (s'',s''') = (dE,dH) where
 sset'    = Set.toList \$ be!s''
 sset'''' = Set.toList \$ fe!s'''
 dH'   = [ (s',s'''')  | s' <- sset', s'''' <- sset'''' ]
 dH''  = [ (s',s''')   | s' <- sset' ]
 dH''' = [ (s'',s'''') | s'''' <- sset'''' ]

 sEdges = [ (g,s) | s' <- sset', (g,s) <- Set.toList \$ bw!s' ]

 dE = [ (s'''',g',q) | s'''' <- sset'''',
                       (g,s) <- sEdges,
                       isPush g, let Push \(\gamma\) = g,
                       (q,g') <- \(\delta\) s'''' \(\gamma\),
                       isPop g' ]

 dH'''' = [ (s,q) | (\_,s) <- sEdges, (\_,\_,q) <- dE ] 

 dH = dH' ++ dH'' ++ dH''' ++ dH''''\end{code}







\end{document}